\documentclass[11pt,a4paper]{article}

\usepackage[margin=1in]{geometry}
\usepackage{palatino}
\usepackage[mathlines]{lineno}
\usepackage{mathpazo}
\usepackage[usenames,dvipsnames]{xcolor}
\usepackage[bookmarksdepth=3]{hyperref}
\hypersetup{%
    colorlinks=true,
    linkcolor=darkgray,
    citecolor=darkgray,
   urlcolor  = gray}

\usepackage{amsmath,amssymb,amsthm}
\usepackage{mathtools}
\mathtoolsset{centercolon}

\usepackage[affil-it]{authblk}
\usepackage[titlenumbered, linesnumbered, ruled]{algorithm2e}
\usepackage{bbold,adjustbox,multicol,multirow}

\newtheorem{theorem}{Theorem}[section]
\newtheorem*{theorem*}{Theorem}
\newtheorem{lemma}{Lemma}[section]
\newtheorem{proposition}{Proposition}[section]
\newtheorem{corollary}[theorem]{Corollary}

\theoremstyle{definition}
\newtheorem{definition}[theorem]{Definition}
\newtheorem{remark}[theorem]{Remark}
\newtheorem{example}[theorem]{Example}

\newcommand{\ML}[1]{\mathsf{ML}(#1)}
\newcommand{\MLANG}{\mathsf{MATLANG}}

\newcommand{\ones}{\mathbb{1}}

\newcommand{\tr}{\mathsf{tr}}

\newcommand{\tp}[1]{#1^{\mathsf{t}}}

\newcommand{\rowdom}{\mathbb{1}}
\DeclareMathOperator{\diag}{\mathsf{diag}}
\DeclareMathOperator{\Apply}{\mathsf{apply}}

\newcommand{\one}{\rowdom}

\newcommand{\degr}{\mathsf{deg}}

\newcommand{\C}{\mathbb{C}}
\newcommand{\R}{\mathbb{R}}
\newcommand{\CLK}[1]{\mathsf{C}^{#1}}

\let\originalcdot\cdot
\renewcommand*{\cdot}{\mskip2mu {\originalcdot}\mskip 2mu\relax}
\newcommand{\mdot}{\cdot}

\newcommand{\change}[1]{#1}

\title{On the expressive power of linear algebra on graphs}

\author{Floris Geerts}
\affil{University of Antwerp, Antwerp, Belgium}

\date{}

\begin{document}

\maketitle
\begin{abstract}
	There is a long tradition in understanding graphs by investigating their adjacency matrices by means of linear algebra. Similarly, logic-based graph query languages are commonly used to explore graph properties. In this paper, we bridge these two approaches by regarding linear algebra as a graph query language. 
	
	More specifically, we consider $\MLANG$, a matrix query language recently introduced, in which some basic linear algebra functionality is supported. We investigate the problem of characterising the equivalence of graphs, represented by their adjacency matrices, for various fragments of $\MLANG$. That is, we are interested in understanding when two graphs cannot be distinguished by posing queries in $\MLANG$ on their adjacency matrices.
	
	Surprisingly, a complete picture can be painted of the impact of each of the linear algebra operations supported in $\MLANG$ on their ability to distinguish graphs. Interestingly, these characterisations can often be phrased in terms of spectral and combinatorial properties of graphs. 
	
	Furthermore, we also establish links to logical equivalence of graphs. In particular, we show that $\MLANG$-equivalence of graphs corresponds to equivalence by means of sentences in the three-variable fragment of first-order logic with counting. Equivalence with regards to a smaller $\MLANG$ fragment is shown to correspond to equivalence by means of sentences in the two-variable fragment of this logic. 
\end{abstract}


\section{Introduction}\label{sec:intro}
Motivated by the importance of linear algebra for machine learning on big data \cite{systemml,Boehm:2018,naughton_la,ngolteanu_learning,olteanu_regression} there is a current interest in languages that combine matrix operations with relational query languages in database systems \cite{Elgohary2017,hutchison,LARA_Berlin_2016,KunftKSRM17,Jermaine/17/LAonRA}. Such hybrid languages raise many interesting questions from a database theoretical point of view. 
\change{The Lara language is one such proposal~\cite{hutchison} and its connections to classical database query languages has been recently explored~\cite{Barcelo-abs-1909-11693}.}
It seems natural, however, to first consider query languages for matrices alone. These are the focus of this paper.

More precisely, we continue the investigation of the expressive power of the matrix query language $\MLANG$, recently introduced by Brijder et al.~\cite{Brijder2018,Brijder2019}, as an analog for matrices of the relational algebra on relations. Intuitively, queries in $\MLANG$ are built up by composing several linear algebra operations commonly found in linear algebra packages. When arbitrary matrices are concerned, it is known that $\MLANG$ is subsumed by aggregate logic with only three non-numerical variables. This implies, among other things, that when evaluated on adjacency matrices of graphs, $\MLANG$ cannot compute the transitive closure of a graph and neither can it express the four-variable query asking if a graph contains a four-clique~\cite{Brijder2018,Brijder2019}. 

In fact, it is implicit in the work by Brijder et al. that when two graphs $G$ and $H$ are indistinguishable by sentences in the three-variable fragment $\CLK{3}$ of first-order logic with counting, denoted by $G\equiv_{\CLK{3}} H$, then their adjacency matrices cannot be distinguished by $\MLANG$ expressions that return scalars, henceforth referred to as \textit{sentences} in $\MLANG$. The equivalence with respect to such sentences is denoted by $G\equiv_{\MLANG} H$. A natural question is whether the converse implication also holds, i.e., does $G\equiv_{\MLANG} H$ also imply $G\equiv_{\CLK{3}} H$? We answer this question affirmatively.

The underlying proof technique relies on a close connection between $\CLK{3}$-e\-qui\-va\-le\-nce and the indistinguishability of graphs by the $2$-dimensional Weisfeiler-Lehman ($2\mathsf{WL}$) algorithm, a result dating back to the seminal paper by Cai, F\"urer and Immerman~\cite{Cai1992,Immerman1990}. Indeed, as we will see, the linear algebra operations supported in $\MLANG$ have sufficient power to simulate the $2\mathsf{WL}$ algorithm. Hence, when $G\equiv_{\MLANG} H$, then $G$ and $H$ cannot by distinguished by the $2\mathsf{WL}$ algorithm. 

This \textit{combinatorial interpretation} of $\MLANG$-equivalence immediately provides an insight in which graph properties are preserved under $\MLANG$-equivalence (see e.g., the work by F\"urer~\cite{Furer2010,Furer2017}). For example, when $G\equiv_{\MLANG} H$, then $G$ and $H$ must be co-spectral (that is, their adjacency matrices have the same multi-set of eigenvalues) and have the same number of $s$-cycles, for $s\leq 6$, but not necessarily for $s> 7$. As observed in the conference version of this paper~\cite{Geerts19}, the case of $7$-cycles easily follows from the connection with $\MLANG$. Indeed, the linear algebra expressions for counting $s$-cycles, for $s\leq 7$, given in Noga et al.~\cite{Alon1997} are expressible in $\MLANG$ and hence, $7$-cycles are preserved by $2\mathsf{WL}$-equivalence. This has recently been verified using other techniques by Arvind et al.~\cite{Arvind2019}. Although formulas exist for counting cycles of length greater than $7$~\cite{Alon1997}, they require counting the number of $k$-cliques, for $k\geq 4$, which is not possible in $\MLANG$, as observed earlier.

Apart from the logical and spectral/combinatorial characterisation of $\MLANG$-e\-qui\-va\-le\-nce, we also point out the correspondence between $\CLK{3}$-equivalence (and thus also $2\mathsf{WL}$-equivalence and $\MLANG$-equivalence) and \textit{conjugacy conditions} between adjacency matrices. \change{Roughly speaking, a conjugacy condition refers to a relationship between adjacency matrices of the form $A_G\mdot T=T\cdot A_H$ for some matrix $T$.} Here, $A_G$ and $A_H$ denote the adjacency matrices of $G$ and $H$, respectively.
As observed by Dawar et al.~\cite{Dawar2016,DAWAR2019993}, $G\equiv_{\CLK{3}} H$ if and only if there exists a unitary matrix $U$ such that $A_G\mdot U=U\mdot A_H$ and moreover, $U$ induces an algebraic isomorphism between the so-called coherent algebras of $A_G$ and $A_H$.  We recall that a unitary matrix $U$ is a complex matrix whose inverse is its complex conjugate transpose $U^*$. Coherent algebras and their isomorphisms are detailed later in the paper.

All combined, we have a logical, combinatorial and conjugation-based characterisation of $\MLANG$-equivalence. Surprisingly, similar characterisations hold also for \textit{fragments of $\MLANG$}. We define fragments of $\MLANG$ by allowing only certain linear algebra operations in our expressions. Such fragments are denoted by $\mathsf{ML}({\cal L})$, with ${\cal L}$ the list of allowed operations. The corresponding notion of equivalence of graphs $G$ and $H$ will be denoted by $G\equiv_{\ML{\cal L}} H$. That is, $G\equiv_{\ML{\cal L}} H$ if any sentence in $\ML{\cal L}$ results in the same scalar when evaluated on $A_G$ and $A_H$. We investigate equivalence for all sensible $\MLANG$ fragments. Our results are as follows:

For starters, we consider the fragment $\ML{\cdot,\tr}$ that allows for matrix multiplication ($\cdot$) and trace ($\tr$) computation (i.e., taking the sum of the diagonal elements of a matrix). Then, $G\equiv_{\ML{\cdot,\tr}} H$ if and only if $G$ and $H$ are co-spectral, or equivalently, they have the same number of closed walks of any length, or $A_G\mdot O=O\mdot A_H$ for some orthogonal matrix $O$. We recall that an orthogonal matrix $O$ is a matrix over the real numbers such that its inverse coincides with the transpose matrix $\tp{O}$ (Section~\ref{subsec:cospectral}). 

Another small fragment, $\ML{\cdot,{}^*,\ones}$, allows for matrix multiplication, conjugate transposition (${}^*$) and the use of the \change{column} vector $\ones$, consisting of all ones. Then, $G\equiv_{\ML{\cdot,{}^*,\ones}} H$ if and only if $G$ and $H$ are co-main (roughly speaking, they are co-spectral only for special ``main'' eigenvalues), or equivalently, they have the same number of (not necessarily closed) walks of any length, or $A_G\mdot Q=Q\mdot A_H$ for some doubly quasi-stochastic matrix $Q$. A doubly quasi-stochastic matrix $Q$ is a matrix over the real numbers such that every of its columns and rows sums up to one (Section~\ref{subsec:comain}). 

When allowing matrix multiplication, \change{$\tr$, ${}^*$, and $\ones$, equivalence of graphs relative to $\ML{\cdot,\tr,\allowbreak {}^*,\ones}$} coincides, not surprisingly, to the graphs being both co-spectral and co-main, or equivalently, having the same number of closed walks of any length and the same number of non-closed walks of any length, or such that $A_G\mdot O=O\mdot A_H$, for an orthogonal doubly quasi-stochastic matrix $O$ (Section~\ref{subsec:comain}). 

More interesting is the fragment $\ML{\cdot,{}^*,\ones,\diag}$, which additionally allows for the operation $\diag(\cdot)$ that turns a column vector into a diagonal matrix with that vector on its diagonal. For this fragment we can tie equivalence to indistinguishability by the $1$-dimensional Weisfeiler-Lehman ($1\mathsf{WL}$) algorithm (or colour refinement). This is known to coincide with the graphs having a common equitable partition, or the existence of a doubly stochastic matrix $S$ such that $A_G\mdot S=S\mdot A_H$ (a.k.a. as a fractional isomorphism), or $\CLK{2}$-equivalence. Here, $\CLK{2}$ denotes the two-variable fragment of first-order logic with counting. We recall that a doubly stochastic matrix is a doubly quasi-stochastic matrix whose entries are all non-negative (Section~\ref{sec:diag}). 

In the former fragment, replacing the operation $\diag(\cdot)$ with an operation ($\odot_v$) which pointwise multiplies vectors results in the same distinguishing power. By contrast, the combination of $\tr$ and the ability to pointwise multiply vectors results in a stronger notion of equivalence. That is, $G\equiv_{\ML{\cdot,\tr,\change{{}^*,}\ones,\odot_v}} H$ if and only if $G$ and $H$ are co-spectral \text{and} indistinguishable by $1\mathsf{WL}$. Also in this case, $A_G\mdot O=O\mdot A_H$ for an orthogonal matrix $O$ that, in addition, needs to preserve equitable partitions. We define this preservation condition later in the paper (Section~\ref{sec:pw}).

For the larger fragment $\ML{\cdot,\tr,{}^*,\ones,\diag}$, no elegant combinatorial characterisation is obtained. Nevertheless, for equivalent graphs $G$ and $H$, $A_G\mdot O=O\mdot A_H$ where $O$ is an orthogonal matrix that can be block-structured according to the equitable partitions. This is a stronger notion than the preservation of equitable partitions. Graphs equivalent with respect to this fragment have, for example, the same number of spanning trees. This is not necessarily true for all previous fragments (Section~\ref{sec:diag}).

Finally, as we already mentioned, equivalence relative to $\MLANG$ is shown to correspond to $\CLK{3}$-equivalence and $2\mathsf{WL}$-equivalence. We additionally refine the conjugation-based characterisation given by Dawar et al.~\cite{Dawar2016,DAWAR2019993} so that it compares more easily to the conjugacy notions used for all previous fragments. Furthermore, we show that pointwise multiplication of matrices (the Schur-Hadamard product) is crucial in this setting  (Section~\ref{sec:C3}).

Each of these fragments can be extended with addition and scalar multiplication at no increase in distinguishing power. It is also shown when fragments can be extended to accommodate for \textit{arbitrary} pointwise function applications, on scalars, vectors or matrices. We furthermore exhibit example graphs \emph{separating} all fragments.

For many of our characterisations we rely on the rich literature on spectral graph theory~\cite{Brouwer2012,Dragos1978,Dragos1997,Dragos2009,Godsil2001,Harary1979,Rowlinson2007,Vandam2007} and the study of the equivalence by the Weisfeiler-Lehman algorithms and fixed-variable fragments of first-order logic with counting~\cite{Dawar2016,DAWAR2019993,Dell2018,Grohe2014,Immerman1990,Ramana1994,Tinhofer1986,Tinhofer1991,Weisfeiler1968}. We describe the relevant results in these papers in due course. We also refer to work by F\"urer~\cite{Furer2010,Furer2017} for more examples of connections to graph invariants and to Dawar et al.~\cite{Dawar2016,DAWAR2019993} for connections between logic, combinatorial and spectral invariants. 

In some sense, we provide a unifying view of various existing results in the literature by grouping them according to the operations supported in $\MLANG$. We remark that, recently, another unifying approach has been put forward by Dell et al.~\cite{Dell2018}. In that work, one considers indistinguishability of graphs in terms of \textit{homomorphism vectors}. That is, one defines $\textsf{HOM}_{\cal F}(G):=(\textsf{Hom}(F,G))_{F\in{\cal F}}$ for some class ${\cal F}$ of graphs, where $\textsf{Hom}(F,G)$ is the number of homomorphisms from $F$ to $G$. Then $G$ and $H$ are indistinguishable for some class ${\cal F}$ of graphs when $\textsf{HOM}_{\cal F}(G)=\textsf{HOM}_{\cal F}(H)$. When ${\cal F}$ consists of all cycles, this notion of equivalence corresponds to $\ML{\cdot,\tr}$-equivalence (recall the closed walk characterisation of the latter); when ${\cal F}$ consists of all paths, we have a correspondence with $\ML{\cdot,{}^*,\ones}$-equivalence (recall the walk characterisation of the latter); when ${\cal F}$ consists of trees, $G$ and $H$ are equivalent for the $1\mathsf{WL}$-algorithm and thus also for $\CLK{2}$ and $\ML{\cdot,{}^*,\ones,\diag}$, and finally, when ${\cal F}$ consists of all graphs of tree-width at most $2$, $G$ and $H$ are equivalent for the $2\mathsf{WL}$-algorithm and thus also for $\CLK{3}$ and $\MLANG$. Our results can thus be regarded as a re-interpretation of the results in Dell et al.~\cite{Dell2018} in terms of $\MLANG$.

We also remark that $\CLK{k}$-equivalence, for $k\geq 4$, can  be characterised in terms of solutions to linear problems which resemble conjugation-based characterisations~\cite{AtseriasM13,grohe_otto_2015,MALKIN2014}. We leave it to future work to identify which additional linear algebra operations to include in $\MLANG$ such that $\CLK{k}$-equivalence can be captured, for $k\geq 4$.

Although we made links to logics such as $\CLK{2}$ and $\CLK{3}$, the connection between $\MLANG$, rank logics and fixed-point logics with counting, as studied in the context of the descriptive complexity of linear algebra~\cite{dghl_rank,Dawar2008,Dawar2017,GradelP15,GroheP17,holm_phd}, is yet to be explored. Similarly for connections to logic-based graph query languages~\cite{Angles:2017,Barcelo2013}. \change{We also
mention that $\MLANG$ can be interpreted as a relational query language on so-called $K$-relations. Such $K$-relations are standard database relations which are annotated with values from a semiring $K$~\cite{Brijder-abs-1904-03934}. This connection provides an elegant formalism of bridging  linear algebra and relational algebra. It further opens the way to explore $\MLANG$ for matrices whose elements are semiring values.}

\change{We want to emphasise that in this work we only consider \textit{indistinguishability} of graphs by means of matrix query languages. 
As such, our results do not directly imply which matrix functions
can be computed by $\MLANG$ expressions, in a uniform dimension-independent way.
This is in contrast to the expressiveness results in Brijder et al.~\cite{Brijder2018,Brijder2019}. Furthermore, we focus only on \textit{undirected graphs} in this paper. Such graphs have symmetric adjacency matrices which have many desirable linear algebra properties,  diagonalisability being the most important one.  Finally, $\MLANG$ is a language in which expressions can combine \textit{multiple input matrices}. Since our focus is on distinguishing graphs, in this work we restrict $\MLANG$ such that its expressions only take a \textit{single matrix}, i.e., the adjacency matrix, as input. Some of our results  generalise to directed graphs (with asymmetric adjacency matrices) or even arbitrary matrices. This is explored in an upcoming paper~\cite{Geerts20}. A full treatment of the general setting with  multiple inputs is left as future work.
 }
 
 This paper is an extended version of the ICDT 2019 conference paper~\cite{Geerts19}. It extends that version by including all proofs in detail. Furthermore, the overall presentation and underlying proof techniques have been simplified. In addition, a new section (Section~\ref{sec:pw}) has been added in which the difference between the $\diag(\cdot)$ and the $\odot_v$ operations is investigated. 

\section{Background} We denote the set of real numbers by $\R$ and the set of complex numbers by $\C$. The set of $m\times n$-matrices over the real (resp., complex) numbers is denoted by $\R^{m\times n}$ (resp., $\C^{m\times n}$). \change{Column} vectors are elements of $\R^{m\times 1}$ (or $\C^{m\times 1}$). \change{Row vectors are elements of $\R^{1\times m}$ (or $\C^{1\times m}$)}. The entries of an $m\times n$-matrix $A$ are denoted by $A_{ij}$, for $i=1,\ldots,m$ and $j=1,\ldots,n$. The entries of a \change{(column or row)} vector $v$ are denoted by $v_i$, for $i=1,\ldots,m$. We often identify $\R^{1\times 1}$ with $\R$, and $\C^{1\times 1}$ with $\C$ and refer to these as \textit{scalars}. \change{Moreover, the $i$th row and $j$th column of a matrix $A\in\R^{m\times n}$ (or in $\C^{m\times n}$) are denoted by $A_{i*}$ and $A_{*j}$, respectively, for $i=1,\ldots,m$ and $j=1,\ldots,n$.} 

The following classes of matrices are of interest in this paper: \textit{square} matrices (elements in $\R^{n\times n}$ or $\C^{n\times n}$), \change{\textit{invertible} matrices (square matrices $A$ for which there exists an \textit{inverse} matrix $B$ such that $A\mdot B=I=B\mdot A$, where $\cdot$ denotes matrix multiplication, and $I$ is the identity matrix in $\C^{n\times n}$)},
 \textit{symmetric} matrices (such that $A_{ij}=A_{ji}$ for all $i$ and $j$),
 \change{\textit{stochastic} matrices ($A_{ij}\in\R$, $A_{ij}\geq 0$, $\sum_{j=1}^n A_{ij}=1$ for all $i$), }
    \textit{doubly stochastic} matrices ($A_{ij}\in\R$, $A_{ij}\geq 0$, $\sum_{j=1}^n A_{ij}=1$ and $\sum_{i=1}^m A_{ij}=1$ for all $i$ and $j$),
  \change{ \emph{quasi-stochastic} matrices ($A_{ij}\in\R$, $\sum_{j=1}^n A_{ij}=1$ for all $i$),  } \emph{doubly quasi-stochastic} matrices ($A_{ij}\in\R$, $\sum_{j=1}^n A_{ij}=1$ and $\sum_{i=1}^m A_{ij}=1$ for all $i$ and $j$), and \emph{orthogonal} matrices (invertible real matrices whose inverse matrix is $\tp{O}$, where $\tp{O}$ denotes the transpose of $O$ obtained by switching rows and columns). The matrix \change{$J\in \R^{n\times n}$} denotes the matrix consisting of all ones and $Z\in\R^{n\times n}$ denotes the zero matrix. We often do not specify the dimensions of matrices and vectors, as these will be clear from the context.

We consider undirected graphs without self-loops. Let $G=(V,E)$ be such a graph with vertices $V=\{1,\ldots,n\}$ and unordered edges $E\subseteq \{\{i,j\}\mid i,j\in V\}$. The \emph{order} of $G$ is simply the number of vertices. Then, an \emph{adjacency matrix} of a graph $G$ of order $n$, denoted by $A_G$, is an $n\times n$-matrix whose entries $(A_G)_{ij}$ are set to $1$ if and only if $\{i,j\}\in E$, all other entries are set to $0$. Strictly speaking, an adjacency matrix requires an ordering on the vertices in $G$. In this paper, this ordering is irrelevant and we often speak about ``the'' adjacency matrix of a graph. For undirected graphs $G=(V,E)$, the adjacency  matrix $A_G$ is a symmetric binary matrix with zeroes on its diagonal.

An \emph{eigenvalue} of a matrix $A$ is a scalar $\lambda$ in $\C$ for which there is a non-zero vector $v$ satisfying $A\mdot v=\lambda v$. Such a vector is called an \emph{eigenvector} of $A$ for eigenvalue $\lambda$. The \emph{eigenspace} of an eigenvalue is the vector space obtained as the span of a maximal set of linear independent eigenvectors for this eigenvalue. Here, the \emph{span} of a set of vectors just refers to the set of all linear combinations of vectors in that set. A set of vectors is linear independent if no vector in that set can be written as a linear combination of other vectors. The \textit{dimension} of an eigenspace is the number of linearly independent eigenvectors spanning that space.
The \emph{spectrum} of an undirected graph can be represented as $ \mathsf{spec}(G)= 
\begin{pmatrix}
	\lambda_1 & \lambda_2 & \cdots & \lambda_p\\
	m_1 & m_2 & \cdots & m_p 
\end{pmatrix}
$, where $\lambda_1<\lambda_2<\cdots<\lambda_p$ are the distinct real eigenvalues of the adjacency matrix $A_G$ of $G$, and where $m_1, m_2,\ldots, m_p$ denote the dimensions of the corresponding eigenspaces. Two graphs are said to be \emph{co-spectral} if they have the same spectrum.

We use $\CLK{k}$ to denote the \textit{$k$-variable fragment of first-order logic with counting}. More precisely, formulas in $\CLK{k}$ are built up from a binary relation $R(x,y)$ (encoding the edge relation of a graph), disjunction, conjunction, negation and counting quantifiers $\exists^{\geq m}$ and use at most $k$ distinct variables~\cite{otto_bounded}. A \textit{sentence} in $\CLK{k}$ is a formula without free variables. 
Two graphs $G$ and $H$ are equivalent with regards to $\CLK{k}$, denoted by $G\equiv_{\CLK{k}} H$, if $G\models\varphi$ if and only if $H\models\varphi$ for every sentence $\varphi$ in $\CLK{k}$. Here, $\models$ denotes the standard notion of satisfaction of logical formulas (see e.g., \cite{Libkin2004,otto_bounded}).

\section{Matrix query languages}\label{sec:mql}

As described in Brijder et al.~\cite{Brijder2018}, matrix query languages can be formalised as compositions of linear algebra operations. Intuitively, a linear algebra operation takes a number of matrices as input and returns another matrix. Examples of operations are matrix multiplication, conjugate transposition, computing the trace, just to name a few. By closing such operations under composition ``matrix query languages'' are formed. More specifically, for linear algebra operations $\mathsf{op}_1,\ldots,\mathsf{op}_k$ the corresponding matrix query language is denoted by $\ML{\mathsf{op}_1,\ldots,\mathsf{op}_k}$ and consists of expressions formed by the following grammar:
\begin{linenomath*}
	$$ e:= X \,|\, \mathsf{op}_1\bigl(e_1,\ldots,e_{p_1}\bigr) \,|\, \cdots \, | \, \mathsf{op}_k \bigr(e_1,\ldots,e_{p_k}\bigr), $$
\end{linenomath*}
where \change{$X$ denotes a \textit{matrix variable} from an infinite set of variables,  which serves to indicate the inputs} to expressions, and $p_i$ denotes the number of inputs required by operation $\mathsf{op}_i$. 
\change{As mentioned in the introduction, we focus on the case when only a fixed single matrix variable, denoted by $X$, is allowed in expressions. We denote expressions by $e(X)$ to make this explicit. } 

The semantics of an expression $e(X)$ in $\ML{\mathsf{op}_1,\ldots,\mathsf{op}_k}$ is defined inductively, relative to an \textit{assignment} $\nu$ of $X$ to a matrix $\nu(X)\in\C^{m\times n}$, for some dimensions $m$ and $n$. We denote by $e\bigl(\nu(X)\bigr)$ the result of evaluating $e(X)$ on $\nu(X)$.  We define, as expected, $$\mathsf{op}_i(e_1(X),\ldots,e_{p_i}(X))(\nu(X)):=\mathsf{op}_i\bigl(e_1(\nu(X)),\ldots,e_{p_i}(\nu(X))\bigr)$$ for linear algebra operation $\mathsf{op}_i$.

\change{In this paper we regard $\MLANG$ as the matrix query language built-up  from the atomic operations listed in Table~\ref{tbl:mloperations} and in which only a single matrix variable $X$ is used.
}
In the table we also show the semantics of the atomic operations. We note that restrictions on the dimensions are in place to ensure that operations are well-defined. Using a simple type system one can formalise a notion of well-formed expressions which guarantees that the semantics of such expressions is well-defined. We refer to Brijder et al.~\cite{Brijder2018} for details. We only consider well-formed expressions from here on. 
\change{As one can observe from Table~\ref{tbl:mloperations}, $\MLANG$ is  parameterised by a set 
$\Omega$ of pointwise functions (see the last operation in Table~\ref{tbl:mloperations}). More specifically, $\Omega=\bigcup_{p>0} \Omega_p$, where $\Omega_p$ consists of some functions $f:\C^p\to \C$.  The choice of $\Omega$ does not impact our results. Hence, we can take $\Omega$ to consist of all possible pointwise functions.}

\begin{table}
	[t]  \adjustbox{max width= 
	\textwidth}{ 
	\begin{tabular}
		{lcc} \multicolumn{2}{l}{\textbf{conjugate transposition} ($\mathsf{op}(e)=e^*$)} & \\
		$e(\nu(X))=A\in \C^{m\times n}$ & $e(\nu(X))^*=A^*\in \C^{n\times m}$& $(A^*)_{ij}=\overline{A}_{ji}$\\\hline \multicolumn{2}{l}{\textbf{one-vector} ($\mathsf{op}(e)=\ones(e)$)} &\\
		$e(\nu(X))=A\in \C^{m\times n}$ & $\ones(e(\nu(X))=\ones\in \C^{m\times 1}$ & $\ones_i=1$\\\hline \multicolumn{2}{l}{\textbf{diagonalization of a vector} ($\mathsf{op}(e)=\diag(e)$) } &\\
		\multirow{2}{*}{$e(\nu(X))=A\in \C^{m\times 1}$} & \multirow{2}{*}{$\diag(e(\nu(X))=\diag(A)\in \C^{m\times m}$} & $\diag(A)_{ii}=A_i$,\\
		& & $\diag(A)_{ij}=0$, $i\neq j$\\\hline \multicolumn{2}{l}{\textbf{matrix multiplication} ($\mathsf{op}(e_1,e_2)=e_1\mdot e_2$)} &\\
		$e_1(\nu(X))=A\in \C^{m\times n}$ & \multirow{2}{*}{$e_1(\nu(X))\mdot e_2(\nu(X))=C\in \C^{m\times o}$} & \multirow{2}{*}{$C_{ij}=\sum_{k=1}^n A_{ik}\times B_{kj}$}\\
		$e_2(\nu(X))=B\in \C^{\change{n\times o}}$ & & \\\hline \multicolumn{2}{l}{\textbf{matrix addition} ($\mathsf{op}(e_1,e_2)=e_1+ e_2$)} &\\
		\change{$e_1(\nu(X))=A\in \C^{m\times n}$ }& \multirow{2}{*}{\change{$e_1(\nu(X))+ e_2(\nu(X))=C\in \C^{m\times n}$}} & \multirow{2}{*}{\change{$C_{ij}=A_{ij}+ B_{ij}$}}\\
		\change{$e_2(\nu(X))=B\in \C^{m\times n}$} & & \\\hline 		
		\multicolumn{2}{l}{\textbf{scalar multiplication} ($\mathsf{op}(e)=c\times {e}$, $c\in \C$)} &\\
		$e(\nu(X))=A\in \C^{m\times n}$ & $c\times e(\nu(X))=B\in \C^{m\times n}$& $B_{ij}=c\times A_{ij}$\\\hline \multicolumn{2}{l}{\textbf{trace} ($\mathsf{op}(e)=\tr(e)$)} &\\
		$e(\nu(X))=A\in \C^{m\times m}$ & $\tr(e(\nu(X))=c\in\C$& $c=\sum_{i=1}^m A_{ii}$\\\hline
\multicolumn{2}{l}{\textbf{\change{pointwise vector multiplication}} ($\mathsf{op}(e_1,e_2)=e_1\odot_v e_2$)} &\\
		\change{$e_1(\nu(X))=A\in \C^{m\times 1}$ }& \multirow{2}{*}{\change{$e_1(\nu(X))\odot_v e_2(\nu(X))=C\in \C^{m\times 1}$}} & \multirow{2}{*}{\change{$C_{i}=A_{i}\times B_{i}$}}\\
		\change{$e_2(\nu(X))=B\in \C^{m\times 1}$} & & \\\hline 		
\multicolumn{2}{l}{\textbf{\change{pointwise matrix multiplication (Schur-Hadamard)}} ($\mathsf{op}(e_1,e_2)=e_1\odot e_2$)} &\\
		\change{$e_1(\nu(X))=A\in \C^{m\times n}$ }& \multirow{2}{*}{\change{$e_1(\nu(X))\odot e_2(\nu(X))=C\in \C^{m\times n}$}} & \multirow{2}{*}{\change{$C_{ij}=A_{ij}\times B_{ij}$}}\\
		\change{$e_2(\nu(X))=B\in \C^{m\times n}$} & & \\\hline 				
		   \multicolumn{3}{l}{\textbf{pointwise function application} ($\mathsf{op}(e_1,\ldots,e_p)=\Apply[f](e_1, \ldots, e_p)$), $f:\C^p\to\C\in\Omega$}\\
\change{$e_1(\nu(X))=A^{(1)}\in \C^{m\times n}$ }& \multirow{3}{*}{\change{$\Apply[f]\bigl(e_1(\nu(X)), \ldots, e_p(\nu(X))\bigr)=B\in \C^{m\times n}$}} & \multirow{3}{*}{\change{$B_{ij}=f(A_{ij}^{(1)},\ldots,A^{(p)}_{ij})$}}\\
				\change{\qquad\quad$\vdots$} & & \\
		\change{$e_p(\nu(X))=A^{(p)}\in \C^{m\times n}$} & & \\\hline 	
	\end{tabular}
	}
	
	\caption{Linear algebra operations supported in $\MLANG$ and their semantics. In the last column, $\bar{\hspace{1ex}}$, $+$ and $\times$ 
	denote complex conjugation, addition and multiplication in $\C$, respectively.} \label{tbl:mloperations} \end{table}

\begin{remark}
	The list of operations in Table~\ref{tbl:mloperations} differs slightly from the list presented in Brijder et al.~\cite{Brijder2018}: We explicitly mention  the trace operation ($\tr$) and pointwise function applications for:	
\change{scalar multiplication ($\times$), addition ($+$), pointwise product of vectors ($\odot_v$) and pointwise product of matrices, also called the Schur-Hadamard product ($\odot$). We remark that all of these operations can be expressed in the language proposed in Brijder et al.~\cite{Brijder2018}. Conversely, any single matrix variable expression in the language of Brijder et al.~\cite{Brijder2018} can be expressed in our $\MLANG$ language. } 
\end{remark}

\begin{remark}
\change{The choice of operations included in $\MLANG$ is motivated by operations supported in linear algebra packages such as \textsf{MAPLE}, \textsf{MATLAB}, \textsf{MATHEMATICA}, \textsf{R}, and others~\cite{Brijder2018}. In $\MLANG$, we only include what we believe to be atomic operations, from which more complex operations can be derived. Many other, more complex, operations could of course be added. Hence $\MLANG$ is just a starting point. We mention that extensions of $\MLANG$ with an $\mathsf{inverse}$ operation (taking the inverse of a matrix if it exists) and an $\mathsf{eigen}$ operation (returning eigenvalues and eigenvectors) are consider Brijder et al.~\cite{Brijder2018}. The precise impact of these operations on the expressive power is yet to be understood.}
\end{remark}

\change{In the following, when $\mathcal{L}$ is a subset of the operations from 
Table~\ref{tbl:mloperations}, we denote by $\ML{\cal L}$ the \textit{fragment} of $\MLANG$ in which only the operations in ${\cal L}$ are
supported.}

\section{Expressive power of matrix query languages}\label{sec:exppow}
As mentioned in the introduction, we are interested in the expressive power of matrix query languages. In analogy with indistinguishability notions used in logic, we consider \textit{sentences} in our matrix query languages. Let ${\cal L}$ be a subset of the operations supported in $\MLANG$. We define an expression $e(X)$ in $\ML{\cal L}$ to be a \textit{sentence} if $e(\nu(X))$ returns a $1\times 1$-matrix (i.e., a scalar) for any assignment $\nu$ of the matrix variable $X$ in $e(X)$. We note that the type system of $\MLANG$ allows to easily check whether an expression in $\ML{\cal L}$ is a sentence (see Brijder et al.~\cite{Brijder2018} for more details). Having defined sentences, a notion of equivalence naturally follows. 
\begin{definition}\label{def:matrixequiv} 
	Two matrices $A$ and $B$ in $\C^{m\times n}$ are said to be \textit{$\ML{\cal L}$-equivalent}, denoted by $A\equiv_{\ML{\cal L }}\!B$, if and only if $e(A)=e(B)$ for all sentences $e(X)$ in $\ML{\cal L}$. 
\end{definition}
In other words, equivalent matrices cannot be distinguished by sentences in the matrix query language fragment under consideration. One could imagine defining equivalence with regards to arbitrary expressions, i.e., expressions in $\MLANG$ that are not necessarily sentences. Such a notion would be too strong, however. Indeed, requiring that $e(A)=e(B)$ for arbitrary expressions $e(X)$ would imply that $A=B$ (just consider $e(X):=X)$) and then the story ends. 

We aim to \textit{characterise} equivalence of matrices for various matrix query languages. We will, however, not treat this problem in full generality and instead only consider equivalence of \textit{adjacency matrices of undirected graphs}. \change{The reason for this limitation is two-fold. First, we can benefit from existing results from graph theory; second, undirected graphs have symmetric adjacency matrices and symmetric matrices have desirable linear algebra properties. For example, every symmetric matrix is diagonalisable. We rely on the properties of symmetric matrices in the proofs of our results.}

Definition~\ref{def:matrixequiv}, when applied to  adjacency matrices, naturally results in the following notion of \textit{equivalence of graphs}.

\begin{definition}
	Two graphs $G$ and $H$ of the same order are said to be \textit{$\ML{\cal L}$-equivalent}, denoted by $G\equiv_{\ML{\cal L}} H$, if and only if their adjacency matrices are $\ML{\cal L}$-e\-qui\-va\-lent. 
\end{definition}

In the following sections we consider equivalence of graphs for various fragments $\ML{\cal L}$, starting from simple fragments only supporting a couple of linear algebra operations, up to the full $\MLANG$ matrix query language. 

\section{Expressive power of the matrix query language  \texorpdfstring{$\ML{\cdot,\tr}$}{ML(.,tr)}}\label{subsec:cospectral} 

\change{We start, in Section~\ref{subsec:multtrace}, by considering the  equivalence of graphs for $\ML{\cdot,\tr}$, i.e., the matrix query language in which only matrix multiplication and the trace operation are supported. We further introduce the notion of \textit{conjugacy} of matrices, which will be used throughout the paper. In Section~\ref{subsec:multrace_ext}, we then explore which operations can be added to $\ML{\cdot,\tr}$ without increasing its distinguishing power.}

\subsection{\texorpdfstring{$\ML{\cdot,\tr}$}{ML(.,tr)}-equivalence}\label{subsec:multtrace}
 
The matrix query language $\ML{\cdot,\tr}$ is a very restrictive fragment. Indeed, the only sentences that one can formulate in $\ML{\cdot,\tr}$ are of the form (i)~$\#\mathsf{cwalk}_k(X):=\tr(X^k)$, where $X^k$ stands for the $k^{\text{th}}$ power of $X$, i.e., $X$ multiplied $k$ times with itself, and (ii)~products of such sentences. To make the connection to graphs, we recall the following notions. A \emph{walk of length} $k$ in a graph $G=(V,E)$ is a sequence $(v_0,v_1,\ldots,v_k)$ of vertices of $G$ such that consecutive vertices are adjacent in $G$, i.e., $\{v_{i-1},v_i\}\in E$ for all $i=1,\ldots,k$. Furthermore, a \emph{closed walk} is a walk that starts in and ends at the same vertex. Closed walks of length $0$ correspond, as usual, to vertices in $G$.
We note that, when evaluated on the adjacency matrix $A_G$ of $G$, $\#\mathsf{cwalk}_k(A_G)$ is equal to the number of closed walks of length $k$ in $G$. Indeed, an entry $(A_G^k)_{vw}$ of the $k^{\text{th}}$ power $A_G^k$ of adjacency matrix $A_G$ can be easily seen to correspond to the number of walks from $v$ to $w$ of length $k$ in $G$. Hence, $\#\mathsf{cwalk}_k(A_G)=\tr(A_G^k)=\sum_{v\in V} (A_G^k)_{vv}$ indeed corresponds to the number of closed walks of length $k$ in $G$.

The following (folklore) characterisations are known to hold.
\begin{proposition}\label{prop:tracesim} 
	Let $G$ and $H$ be two graphs of the same order. The following statements are equivalent: 
	\begin{enumerate}
		\item[(1)] $G$ and $H$ have the same number of closed walks of length $k$, for all $k\geq 0$; 
		\item[(2)] $\tr(A_G^k)=\tr(A_H^k)$ for all $k\geq 0$; 
		\item[(3)] $G$ and $H$ are co-spectral; and 
		\item[(4)] there exists an orthogonal matrix $O$ such that $A_G\mdot O=O\mdot A_H$. 
	\end{enumerate}
\end{proposition}
\begin{proof}
	For a proof of the equivalences (1) $\Leftrightarrow$ (2) $\Leftrightarrow$ (3) we refer to Proposition 1 in~\cite{Dawar2016} (although these equivalences appeared in the literature many times before). The equivalence (3) $\Leftrightarrow$ (4) is also known (see e.g., Theorem 9-12 in~\cite{Perlis1952}). \end{proof}

\begin{example}\label{ex:cospectral} 
	The graphs $G_1$ (\!\raisebox{-0.4ex}{\mbox{ 
	\includegraphics[height=0.3cm]{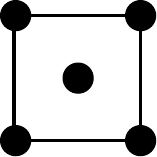}}}\hspace{.08em}) and $H_1$ (\!\raisebox{-0.4ex}{\mbox{ 
	\includegraphics[height=0.3cm]{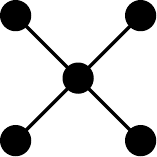}}}\hspace{.08em}) are the smallest pair (in terms of number of vertices) of non-isomorphic co-spectral graphs of the same order (see e.g., Figure 6.2 in~\cite{Dragos1971}). From the previous proposition we then know that $G_1$ and $H_1$ have the same number of closed walks of any length. We note that the isolated vertex in $G_1$ ensures that $G_1$ and $H_1$ have the same number of vertices (and thus the same number of closed walks of length $0$). ~\hfill$\qed$ 
\end{example}

As anticipated, sentences in $\ML{\cdot, \tr}$ can only extract information from adjacency matrices related to the number of closed walks in graphs. More precisely, we can add to Proposition~\ref{prop:tracesim} a fifth equivalent condition based on $\ML{\cdot,\tr}$-equivalence: 
\begin{proposition}\label{theorem:trace} 
	For two graphs $G$ and $H$ of the same order, $G\equiv_{\ML{\cdot,\, \tr}} H$ if and only if $G$ and $H$ have the same number of closed walks of any length. 
\end{proposition}
\begin{proof}
	By definition, if $G\equiv_{\ML{\cdot,\,\tr}} H$, then $e(A_G)=e(A_H)$ for any sentence $e(X)$ in $\ML{\cdot,\tr}$. This holds in particular for the sentences $\#\mathsf{cwalk} _k(X):=\tr(X^k)$ in $\ML{\cdot,\tr}$, for $k\geq 1$. Hence, $G$ and $H$ have indeed the same number of closed walks of length $k$, for $k\geq 1$. Furthermore, since $G$ and $H$ are of the same order and $A_G^0=A_H^0=I$ (by convention), $G$ and $H$ have also the same number of closed walks of length $0$. 
	
	For the converse, if $G$ and $H$ have the same number of closed walks of any length, then the previous proposition tells that $A_G\mdot O=O\mdot A_H$ for some orthogonal matrix $O$. We next claim that when $A_G\mdot O=O\mdot A_H$ holds for some orthogonal matrix $O$, then $e(A_G)=e(A_H)$ for any sentence $e(X)$ in $\ML{\cdot,\tr}$. In fact, this claim will follow from the more general Lemmas~\ref{lem:multp-sim} and~\ref{lem:trace-sim} below. We separate these \change{lemmas} from the current proof since we also need them later in the paper. 
\end{proof}
We note that yet another interpretation of $G\equiv_{\ML{\cdot, \tr}} H$ can be given in terms of the homomorphism vectors mentioned in the introduction. That is,  $G\equiv_{\ML{\cdot, \tr}} H$  if and only if $\textsf{HOM}_{\cal F}(G)=\textsf{HOM}_{\cal F}(H)$ where ${\cal F}$ is the set of all cycles~\cite{Dell2018}. 

As mentioned in the proof of Proposition~\ref{theorem:trace}, we still need to show that if $A_G\mdot O=O\mdot A_H$ holds for some orthogonal matrix $O$, then $e(A_G)=e(A_H)$ for any sentence $e(X)$ in $\ML{\cdot,\tr}$. 

In more generality, when $A_G\mdot T=T\mdot A_H$ holds for a (\change{not necessarily invertible}) matrix $T$, we say that $A_G$ and $A_H$ are \change{\textit{$T$-conjugate}. We remark that conjugation  is not necessarily a symmetric relation, i.e., $A_G$ and $A_H$ can be $T$-conjugate whereas $A_H$ and $A_G$ may not be $T$-conjugate for the same matrix $T$.}
We also define the notion of $T$-conjugation for vectors and scalars, as is shown next.
\begin{definition}\label{def:sim} \change{Let $n>1$.}
	Let $T$ be a matrix in $\C^{n\times n}$.
	Two matrices $A$ and $B$ in $\C^{n\times n}$ are called \textit{$T$-conjugate} if $A\mdot T=T\mdot B$. Two \change{column} vectors $A$ and $B$ in $\C^{n\times 1}$ are \textit{$T$-conjugate} if $A=T\mdot B$. Similarly, two \change{row} vectors $A$ and $B$ in $\C^{1\times n}$ are \textit{$T$-conjugate} if $A\mdot T=B$. Finally, if $A$ and $B$ are scalars in $\C$ (or elements in $\C^{1\times 1}$), then $A$ and $B$ are $T$-conjugate if $A=B$ (i.e., $T$-conjugation of scalars is simply equality). 
\end{definition}

In $\ML{\cdot,\tr}$ we allow for matrix multiplication and the trace operation. We first show that $T$-conjugation is preserved by matrix multiplication. 
\begin{lemma}\label{lem:multp-sim} Let $\ML{\cal L}$ be any matrix query language fragment and let $G$ and $H$ be two graphs of the same order.
	Consider expressions  $e_1(X)$ and $e_2(X)$ in $\ML{\cal L}$. If $e_i(A_G)$ and $e_i(A_H)$ are $T$-conjugate, for $i=1,2$, for some matrix $T$, then $e_1(A_G)\mdot e_2(A_G)$ is also $T$-conjugate to $e_1(A_H)\mdot e_2(A_H)$ (provided, of course, that the multiplication is well-defined). 
\end{lemma}
\begin{proof}[\change{sketch}]
	The proof consists of a simple case analysis depending on the dimensions of $e_1(A_G)$ and $e_2(A_G)$ (or equivalently, the dimensions of $e_1(A_H)$ and $e_2(A_H)$) and by using the definition of $T$-conjugation. We refer for the full proof to the appendix. 
\end{proof}

When considering the trace operation, we observe that $T$-conjugation is preserved by the trace operation, provide that  $T$ is an \textit{invertible} matrix.
 
\begin{lemma}\label{lem:trace-sim} Let $\ML{\cal L}$ be any matrix query language fragment and let $G$ and $H$ be two  graphs of the same order. Let $e_1(X)$ be an expression in $\ML{\cal L}$. If $e_1(A_G)$ and $e_1(A_H)$ are $T$-conjugate for an \textit{invertible} matrix $T$, 
	 then $\tr(e_1(A_G))$ and $\tr(e_1(A_H))$ are also $T$-conjugate. 
\end{lemma}
\begin{proof}
	Let $e(X):=\tr(e_1(X))$. By assumption, $e_1(A_G)\mdot T=T\mdot e_1(A_H)$ for an invertible matrix $T$ in case that $e_1(A_G)$ is an $n\times n$-matrix \change{with $n>1$}, and $e_1(A_G)=e_1(A_H)$ in case that $e_1(A_G)$ is a sentence. In the latter case, clearly also $e(A_G)=\tr(e_1(A_G))=\tr(e_1(A_H))=e(A_H)$. In the former case, we use the property that $\tr(T^{-1}\mdot A\mdot T)=\tr(A)$ for any matrix $A$ and invertible matrix $T$ (see e.g., Chapter 10 in \cite{ladoneright} for a proof of this property). Hence, we have that $ e(A_G)=\tr(e_1(A_G))=\tr(T^{-1} \mdot e_1(A_G)\mdot T)=\tr(T^{-1}\mdot T \mdot e_1(A_H))\allowbreak=\allowbreak\tr(I\mdot e_1(A_H))=\allowbreak \tr(e_1(A_H))=e(A_H)$ holds, as desired. 
\end{proof}
We remark that Lemmas~\ref{lem:multp-sim} and~\ref{lem:trace-sim} hold for \textit{any} fragment $\ML{\cal L}$.

The claim at the end of the proof of Proposition~\ref{theorem:trace}, i.e., that $O$-conjugation of $A_G$ and $A_H$ indeed implies that $e(A_G)=e(A_H)$ for any sentence $e(X)\in\ML{\cdot,\tr}$, now easily follows by induction on the structure of expressions,  Indeed, since orthogonal matrices are invertible, Lemmas~\ref{lem:multp-sim} and~\ref{lem:trace-sim} imply that when $e_1(A_G)$ and $e_1(A_H)$, and $e_2(A_G)$ and $e_2(A_H)$ are $O$-conjugate for an orthogonal matrix $O$, then also $e_1(A_G)\mdot e_2(A_G)$ and $e_1(A_H)\mdot e_2(A_H)$ are $O$-conjugate, and $\tr(e_1(A_G))$ and $\tr(e_1(A_H))$ are $O$-conjugate (i.e., equal). Hence, when $A_G$ and $A_H$ are $O$-conjugate, $e(A_G)$ and $e(A_H)$ are $O$-conjugate for any sentence $e(X)\in\ML{\cdot,\tr}$. That is, $e(A_G)=e(A_H)$ for any sentence in $\ML{\cdot,\tr}$. 

\subsection{Adding operations to \texorpdfstring{$\ML{\cdot,\tr}$}{ML(.,tr)} without increasing its distinguishing power}\label{subsec:multrace_ext} We next investigate how much more $\ML{\cdot,\tr}$ can be extended whilst preserving the characterisation given in Proposition~\ref{theorem:trace}. Some more general observations will be made in this context, which will be used for other fragments later in the paper as well.

First, we consider the extension with scalar multiplication ($\times$) and addition ($+$). 
\begin{lemma}\label{cor:ridoflinear} 
	Let $\ML{\cal L}$ be any matrix query language fragment. Let $e_1(X)$ and $e_2(X)$ be two expressions in $\ML{\cal L}$ and consider two graphs $G$ and $H$ of the same order. Then, if $e_1(A_G)$ and $e_1(A_H)$, and $e_2(A_G)$ and $e_2(A_H)$ are $T$-conjugate for some matrix $T$, then also $e_1(A_G)+e_2(A_G)$ and $e_1(A_H)+e_2(A_H)$ are $T$-conjugate, and $a\times e_1(A_G)$ and $a\times e_1(A_H)$ are $T$-conjugate for any scalar \change{$a\in \C$}. 
\end{lemma}
\begin{proof}
	This is an immediate consequence of the definition of $T$-conjugation and that matrix multiplication is a bilinear operation, i.e., $(a\times A+b\times B)\mdot (c\times C+d\times D)=(a\times c)\times (A\mdot C)+ (a\times d)\times (A\mdot D)+ (b\times c)\times (B\mdot C)+(b\times d)\times (B\mdot D)$, for scalars $a$, $b$, $c$, $d\in\C$ and matrices or vectors $A,B,C$ and $D$. 
\end{proof}

We next consider complex conjugate transposition (${}^*$). 
\begin{lemma}\label{lem:complextranspose-elim1} 
	Let $\ML{\cal L}$ be any matrix query language fragment. Let $e(X)$ be an expression in $\ML{\cal L}$ and consider two graphs $G$ and $H$ of the same order. Then, if $e(A_G)$ and $e(A_H)$ are $T$-conjugate, and $e(A_H)$ and $e(A_G)$ are $T^*$-conjugate for some matrix $T$, then also $(e(A_G))^*$ and $(e(A_H))^*$ are $T$-conjugate, and $(e(A_H))^*$ and $(e(A_G))^*$ are $T^*$-conjugate. 
\end{lemma}

\begin{proof}
	We distinguish between a number of cases, depending on the dimensions of $e(A_G)$ (and hence also of $e(A_H)$). Suppose that $e(A_G)$ returns an $n\times n$-matrix \change{for $n>1$}. Then, by assumption $e(A_G)\mdot T=T\mdot e(A_H)$ and $e(A_H)\mdot T^*=T^*\mdot e(A_G)$. It then follows, using that the operation ${}^*$ is an involution ($(A^*)^*=A$) and $(A\mdot B)^*=B^*\mdot A^*$, that 
	\begin{linenomath}
		\[(e(A_G))^*\mdot T=(T^*\mdot e(A_G))^*=(e(A_H)\mdot T^*)^*=T\mdot (e(A_H))^*, \]
	\end{linenomath}
	and similarly, 
	\begin{linenomath}
		\[ (e(A_H))^*\mdot T^*=(T\mdot e(A_H))^*=(e(A_G)\mdot T)^*=T^*\mdot (e(A_G))^*. \]
	\end{linenomath}
	Furthermore, when $e(A_G)$ is an $n\times 1$-vector \change{for $n>1$}, we have by assumption that $e(A_G)=T\mdot e(A_H)$ and $e(A_H)=T^*\mdot e(A_G)$. Hence, 
	$(e(A_G))^*\mdot T=(T^*\!\mdot e(A_G))^*=(e(A_H))^*$ and $(e(A_H))^*\mdot T^*=(T\mdot e(A_H))^*=(e(A_G))^*$. 
	Similarly, when $e(A_G)$ is a $1\times n$-vector \change{for $n>1$}, one can verify that $((e(A_G))^*=T\mdot (e(A_H))^*$ and $(e(A_H))^*=T^*\mdot (e(A_G))^*$. Finally, if $e(A_G)$ is a sentence then clearly $(e(A_G))^*=(e(A_H))^*$. 
\end{proof}

We next consider pointwise function applications. Later in the paper we show that pointwise function applications on vectors or matrices do add expressive power. \change{This is particularly true for
pointwise multiplication of vectors ($\odot_v$) and of matrices ($\odot$).}
 By contrast, when function applications are \textit{only allowed on scalars} they do not add any expressive power. \change{More specifically, let $\Omega$ be an arbitrary set of pointwise functions and} let $f:\C^k \to \C$ be a function in $\Omega$. We denote by $\Apply_{\mathsf{s}}[f](e_1,\ldots,\allowbreak e_k)$ the application of $f$ on $e_1(X),\ldots, e_k(X)$ when each $e_i(X)$ is a sentence. 
\begin{lemma}\label{lem:ridofscalarpointw} 
	Let $\ML{\cal L}$ be any matrix query language fragment. Consider two graphs $G$ and $H$ of the same order and sentences $e_1(X), e_2(X),\ldots,e_k(X)$ in $\ML{\cal L}$. Let $f:\C^k\to\C$ be a function in $\Omega$. Suppose that for each $i=1,\ldots,k$, $e_i(A_G)=e_i(A_H)$ (i.e., they are $T$-conjugate for an arbitrary matrix $T$). Then also $\Apply_{\mathsf{s}}[f](e_1(A_G),\ldots,\allowbreak e_k(A_G))=\Apply_{\mathsf{s}}[f](e_1(A_H),\ldots,\allowbreak e_k(A_H))$ (i.e., they are $T$-conjugate as well). 
\end{lemma}
\begin{proof}
	This is straightforward to verify since the result of a function $f:\C^k\to \C$ is fully determined by its input values. 
\end{proof}

Given these lemmas, we can infer that the characterisation given in Proposition~\ref{theorem:trace} remains to hold for $\ML{\cdot,\tr,+,\times,{}^*,\Apply_{\mathsf{s}}[f],\allowbreak f\in\Omega}$-equivalence. 
\begin{corollary}\label{corr:tracelincomb} 
	For two graphs $G$ and $H$ of the same order, we have that  $G\equiv_{\ML{\cdot,\,\tr}} H$ if and only if $G\equiv_{\ML{\cdot,\,\tr,{}^*,+,\allowbreak\times,\allowbreak\Apply_{\mathsf{s}}[f],f\in\Omega}} H$.~\hfill$\qed$ 
\end{corollary}
\begin{proof}
	We only need to show that $G\equiv_{\ML{\cdot,\,\tr}} H$ implies $G\equiv_{\ML{\cdot,\,\tr,{}^*,+,\times,\allowbreak\Apply_{\mathsf{s}}[f],f\in\Omega}} H$. By Proposition~\ref{theorem:trace}, there exists an orthogonal matrix $O$ such that $A_G\mdot O=O\mdot A_H$. Furthermore, we have that $O^*\mdot A_G=(A_G\mdot O)^*=(O\mdot A_H)^*=A_H\mdot O^*$ since $A_G$ and $A_H$ are symmetric real matrices. Hence, $A_H$ and $A_G$ are $O^*$-conjugate. We also, importantly, observe that $O^*$ is an orthogonal matrix as well. Lemmas~\ref{lem:multp-sim} and~\ref{lem:trace-sim} then imply that $e(A_G)$ and $e(A_H)$ are $O$-conjugate, and $e(A_H)$ and $e(A_G)$ are $O^*$-conjugate for any expression $e(X)$ in $\ML{\cdot,\tr}$. Furthermore, Lemmas~\ref{cor:ridoflinear},~\ref{lem:complextranspose-elim1} and~\ref{lem:ridofscalarpointw} imply that addition, scalar multiplication, complex conjugate transposition and pointwise function applications on scalars preserve $O$ and $O^*$-conjugation. This in turn implies that $e(A_G)=e(A_H)$ for any sentence $e(X)\in\ML{\cdot,\,\tr,{}^*,+,\times,\allowbreak\Apply_{\mathsf{s}}[f],f\in\Omega}$. 
\end{proof}
As a consequence, the graphs $G_1$ (\!\raisebox{-0.4ex}{\mbox{ 
\includegraphics[height=0.3cm]{graphG1}}}\hspace{.08em}) and $H_1$ (\!\raisebox{-0.4ex}{\mbox{ 
\includegraphics[height=0.3cm]{graphH1}}}\hspace{.08em}) from Example~\ref{ex:cospectral} cannot be distinguished by sentences in $\ML{\cdot,\tr,{}^*,+,\times,\Apply_{\mathsf{s}}[f],f\in\Omega}$. As we will see later, including any other operation from Table~\ref{tbl:mloperations}, such as $\ones(\cdot)$, $\diag(\cdot)$ or pointwise function applications on vectors or matrices, allows us to distinguish $G_1$ and $H_1$.

\section{The impact of the \texorpdfstring{$\ones(\cdot)$}{1(.)} and ${}^*$ operations}\label{subsec:comain} 
\change{We next consider two fragments that support 
complex conjugate transposition  ${}^*$ and the operation $\ones(\cdot)$, which returns the all-ones column vector $\ones$\footnote{We use $\ones$ to denote the all-ones \emph{vector} (of appropriate dimension) and use $\ones(\cdot)$ (with brackets) for the corresponding one-vector \textit{operation}.}. More specifically, we consider the fragments $\ML{\cdot,{}^*,\ones}$ and $\ML{\cdot,\tr,{}^*,\ones}$.}

\change{The presence of $\ones(\cdot)$ allows to extract, in combination with ${}^*$, other information from graphs than just the number of closed walks. Indeed, consider the sentence
\begin{linenomath}
	\[ \#\mathsf{walk}_k(X):=(\ones(X))^*\mdot X^k \mdot\ones(X)\]
\end{linenomath}
in $\ML{\cdot,{}^*,\ones}$ and $\ML{\cdot,\tr,{}^*,\ones}$.}
 When applied on the adjacency matrix $A_G$ of a graph $G$, $\#\mathsf{walk}_k(A_G)$ 
 returns the number of (not necessarily closed) walks in $G$ of length $k$. In relation to the previous section, co-spectral graphs have the same number of closed walks of any length, yet do not necessarily have the same number of walks of any length. Similarly, graphs with the same number of walks of any length are not necessarily co-spectral. We illustrate this by the following example.
\begin{example}\label{ex:cospecvsallwalks} 
	\normalfont It can be verified that the co-spectral graphs $G_1$ (\!\raisebox{-0.4ex}{\mbox{ 
	\includegraphics[height=0.3cm]{graphG1}}}\hspace{.08em}) and $H_1$ (\!\raisebox{-0.4ex}{\mbox{ 
	\includegraphics[height=0.3cm]{graphH1}}}\hspace{.08em}) of Example~\ref{ex:cospectral} have $16$ versus $20$ walks of length $2$, respectively. As a consequence, $\ML{\cdot,{}^*,\ones}$ and $\ML{\cdot,\tr,\change{{}^*},\ones}$ can distinguish $G_1$ from $H_1$ by means of the \change{sentence $\#\mathsf{walk}_2(X)$.}
	 By contrast, the graphs $G_2$ (\!\raisebox{-1.4ex}{\mbox{ 
	\includegraphics[height=0.6cm]{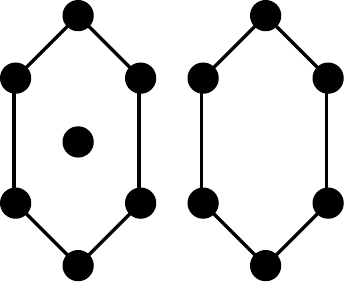}}}\hspace{.08em}) and $H_2$ (\!\raisebox{-1.4ex}{\mbox{ 
	\includegraphics[height=0.6cm]{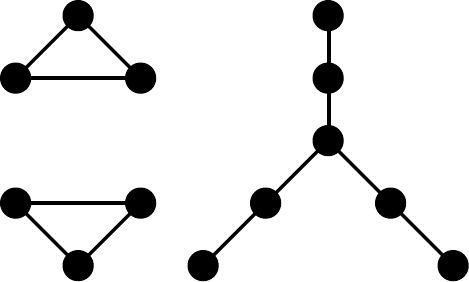}}}) are not co-spectral, yet have the same number of walks of any length. It is easy to see that $G_2$ and $H_2$ are not co-spectral (apart from verifying that their spectra are different): $H_2$ has $12$ closed walks of length $3$ (because of the triangles), whereas $G_2$ has no closed walks of length $3$. As a consequence, $\ML{\cdot,\tr}$ (and thus also $\ML{\cdot,\tr,\change{{}^*},\ones}$) can distinguish $G_2$ and $H_2$. We argue below that $G_2$ and $H_2$ have the same number of walks of any length and show that $\ML{\cdot,{}^*,\ones}$ cannot distinguish $G_2$ and $H_2$. ~\hfill$\qed$ 
\end{example}
The previous example illustrates the key difference between the fragments $\ML{\cdot,{}^*,\ones}$ and $\ML{\cdot,\tr,\change{{}^*},\ones}$. The former can only detect differences in the number of walks of certain lengths, the latter can detect differences in both the number of walks and the number of closed walks of certain lengths.

Graphs sharing the same number of walks of any length have been investigated before in spectral graph theory~\cite{Dragos1978,Dragos1997,Dragos2009,Harary1979,Rowlinson2007}. To state a spectral characterisation, the so-called \emph{main spectrum} of a graph needs to be considered. The main spectrum of a graph is the set of eigenvalues whose eigenspace is not orthogonal to the $\ones$ vector. More formally, consider an eigenvalue $\lambda$ and its corresponding eigenspace, represented by a matrix $V$ whose columns are eigenvectors of $\lambda$ that span the eigenspace of $\lambda$. Then, the \emph{main angle} $\beta_\lambda$ of $\lambda$'s eigenspace is $\frac{1}{\sqrt{n}}\| \tp{V}\mdot \ones\|_2$, where $\|\mdot \|_2$ is the Euclidean norm. The \textit{main eigenvalues} are now simply those eigenvalues with a non-zero main angle. Furthermore, two graphs are said to be \emph{co-main} if they have the same set of main eigenvalues and corresponding main angles. Intuitively, the importance of the orthogonal projection on $\ones$ stems from the observation that $\#\mathsf{walk}_k(A_G)$ can be expressed as $\sum_{i} \lambda_i^k \beta_{\lambda_i}^2$ where the $\lambda_i$'s are the distinct  eigenvalues of $A_G$.\footnote{\change{Underlying this observation is that $A_G$ is a symmetric matrix and hence diagonalisable.}} Clearly, only those eigenvalues $\lambda_i$ for which $\beta_{\lambda_i}$ is non-zero matter when computing $\#\mathsf{walk}_k(A_G)$. This results in the following characterisation. 
\begin{proposition}[Theorem 1.3.5 in Cvetkovi\'c et al.~\cite{Dragos2009}]\label{prop:comain} 
	Two graphs $G$ and $H$ of the same order are co-main if and only if they have the same number of walks of length $k$, for every $k\geq 0$. ~\hfill$\qed$ 
\end{proposition}
Furthermore, the following proposition follows implicitly from the proof of Theorem 3 in van Dam et al.~\cite{Vandam2007}. This proposition is also explicitly proved more recently in Theorem 1.2 in Dell et al.~\cite{Dell2018} in the context of distinguishing graphs by means of homomorphism vectors
$\textsf{HOM}_{\cal F}(G)$ and $\textsf{HOM}_{\cal F}(H)$ where ${\cal F}$ consists of all paths.
\begin{proposition}\label{prop:qds} 
	Two graphs $G$ and $H$ of the same order have the same number of walks of length $k$, for every $k\geq 0$, if and only if there is a doubly quasi-stochastic matrix $Q$ such that $A_G\mdot Q=Q\mdot A_H$. ~\hfill$\qed$ 
\end{proposition}
\begin{example}[Continuation of Example~\ref{ex:cospecvsallwalks}]\label{ex:stoch} 
	\normalfont Consider the subgraph $G_3$ (\!\raisebox{-1.3ex}{\mbox{ 
	\includegraphics[height=0.6cm]{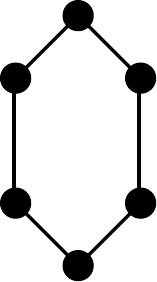}}}\hspace{.08em}) of $G_2$ and the subgraph $H_3$ (\!\raisebox{-1.3ex}{\mbox{ 
	\includegraphics[height=0.6cm]{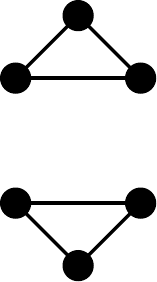}}}\hspace{.08em}) of $H_2$. It is readily verified that there exists a doubly quasi-stochastic matrix $Q$ such that $A_{G_3}\mdot Q=Q\mdot A_{H_3}$. Indeed, $A_{G_3}\mdot Q$ is equal to 
	\begin{linenomath}
		\[ 
		\begin{bmatrix}
			0 & 0 & 1 & 0 & 0& 1\\
			0 & 0 & 0 & 1 & 1 & 0\\
			1 & 0 & 0 & 0 & 1 & 0\\
			0 & 1 & 0 & 0 & 0& 1\\
			0 &1 & 1 & 0 & 0 & 0\\
			1 & 0 & 0 & 1 & 0 & 0\\
		\end{bmatrix}
		\mdot 
		\begin{bmatrix}
			0 & 0 & \frac{1}{2} & 0 & 0& \frac{1}{2}\\
			0 & 0 & 0 & \frac{1}{2} & \frac{1}{2} & 0\\
			\frac{1}{2} & 0 & 0 & 0 & \frac{1}{2} & 0\\
			0 & \frac{1}{2} & 0 & 0 & 0& \frac{1}{2}\\
			0 &\frac{1}{2} & \frac{1}{2} & 0 & 0 & 0\\
			\frac{1}{2} & 0 & 0 & \frac{1}{2} & 0 & 0\\
		\end{bmatrix}
		= 
		\begin{bmatrix}
			0 & 0 & \frac{1}{2} & 0 & 0& \frac{1}{2}\\
			0 & 0 & 0 & \frac{1}{2} & \frac{1}{2} & 0\\
			\frac{1}{2} & 0 & 0 & 0 & \frac{1}{2} & 0\\
			0 & \frac{1}{2} & 0 & 0 & 0& \frac{1}{2}\\
			0 &\frac{1}{2} & \frac{1}{2} & 0 & 0 & 0\\
			\frac{1}{2} & 0 & 0 & \frac{1}{2} & 0 & 0\\
		\end{bmatrix}
		\mdot 
		\begin{bmatrix}
			0 & 0 & 1 & 0 & 0& 1\\
			0 & 0 & 0 & 1 & 1 & 0\\
			1 & 0 & 0 & 0 & 1 & 0\\
			0 & 1 & 0 & 0 & 0& 1\\
			0 &1 & 1 & 0 & 0 & 0\\
			1 & 0 & 0 & 1 & 0 & 0\\
		\end{bmatrix}
		,\]
	\end{linenomath}
	which is equal to $Q\mdot A_{H_3}$. Hence by Proposition~\ref{prop:qds}, $G_3$ and $H_3$ have the same number of walks on any length. ~\hfill$\qed$ 
\end{example}

Just as for the fragment $\ML{\cdot,\tr}$ (Proposition~\ref{theorem:trace}), it turns out that sentences in $\ML{\cdot,{}^*,\ones}$ can only extract information from adjacency matrices related to the number of walks in graphs.
\change{More precisely, we have the following proposition.}
\begin{proposition}\label{prop:quasistoch} 
	Let $G$ and $H$ be two graphs of the same order. Then, $G\equiv_{\ML{\cdot,{}^*,\ones}} H$ if and only if $G$ and $H$ have the same number of walks of any length. 
\end{proposition}
\begin{proof}
	It is straightforward to show that $G\equiv_{\ML{\cdot,{}^*,\ones}} H$ implies that $G$ and $H$ must have the same number of walks of any length. This follows from the same argument as given in the proof of Proposition~\ref{theorem:trace}. For the converse, we use the characterisation given in Proposition~\ref{prop:qds}. That is, if $G$ and $H$ have the same number of walks of any length, then there exists a doubly quasi-stochastic matrix $Q$ such that $A_G\mdot Q=Q\mdot A_H$. In other words, $A_G$ and $A_H$ are $Q$-conjugate. We now show that when $A_G$ and $A_H$ are $Q$-conjugate, for a doubly quasi-stochastic matrix $Q$, then $e(A_G)=e(A_H)$ for all sentences $e(X)$ in $\ML{\cdot,{}^*,\ones}$. We here rely on a more general result (Lemma~\ref{lem:ones-sim} below), which states that $T$-conjugation is preserved by the operation $\ones(\cdot)$ provided that $T$ is  quasi-stochastic.
We again separate this lemma from the current proof because we need it also later in the paper. This suffices to conclude that expressions in $\ML{\cdot,{}^*,\ones}$ preserve $Q$-conjugation for a doubly quasi-stochastic matrix $Q$. Indeed, to deal with complex conjugate transposition, we note that $A_G\mdot Q=Q\mdot A_H$ implies that $A_H\mdot Q^*=(Q\mdot A_H)^*=(A_G\mdot Q)^*=Q^*\mdot A_G$ since $A_G$ and $A_H$ are symmetric real matrices. Hence, $A_H$ and $A_G$ are $Q^*$-conjugate. Furthermore, since $Q$ is a real matrix and quasi doubly-stochastic, also $Q^*\cdot\ones=\one$ holds. That is, $Q^*$ is a (doubly) quasi-stochastic matrix as well. Hence, Lemmas~\ref{lem:multp-sim} and~\ref{lem:ones-sim} imply that $Q$-conjugation and $Q^*$-conjugation are preserved by matrix multiplication and the one-vector operation. Combined with Lemma~\ref{lem:complextranspose-elim1}, we may  conclude that $Q$-conjugation and $Q^*$-conjugation is also preserved by complex conjugate transposition. Hence, by induction on the structure of expressions, $e(A_G)=e(A_H)$ for any sentence $e(X)\in \ML{\cdot,{}^*,\ones}$. 
\end{proof}

We now show that $T$-conjugation is preserved under the one-vector operation for any quasi-stochastic matrix $T$. In fact, since the result of $\ones(\cdot)$ is only dependent on the dimensions of the input, we have do not even need the $T$-conjugation assumption on the inputs.
\begin{lemma}\label{lem:ones-sim} Let $\ML{\cal L}$ be any matrix query language fragment and consider two  graphs $G$ and $H$ of the same order. Let $e_1(X)$ be an expression in $\ML{\cal L}$. Then, $\ones(e_1(A_G))$ and $\ones(e_1(A_H))$ are $T$-conjugate for any quasi-stochastic matrix $T$. 
\end{lemma}
\begin{proof}
	The proof is straightforward. Let $e(X):=\ones(e_1(X))$. We distinguish between the following cases, depending on the dimensions of $e_1(A_G)$. If $e_1(A_G)$ is an $n\times n$-matrix or $n\times 1$-vector, \change{for $n>1$,} then $e(A_G)=e(A_H)=\ones$ and $e(A_G)=\ones=T\cdot\ones=T\mdot e(A_H)$. Furthermore, if $e_1(A_G)$ is a $1\times n$-vector or sentence, then $e(A_G)=e(A_H)=1$ and thus these agree and are $T$-conjugate. 
\end{proof}

We next turn our attention to $\ML{\cdot,\tr,\change{{}^*},\ones}$. We know from Propositions~\ref{prop:tracesim} and~\ref{theorem:trace} that $G\equiv_{\ML{\cdot,\,\tr,\change{{}^*},\ones}} H$ implies that $G$ and $H$ are co-spectral. Combined with Proposition~\ref{prop:comain} and the fact that the sentence $\#\mathsf{walk}_k(X)$ counts the number of walks of length $k$, we have that $G\equiv_{\ML{\cdot,\,\tr,\change{{}^*},\ones}} H$ implies that $G$ and $H$ are co-spectral and co-main. The following is known about such graphs. 
\begin{proposition}[Corollary to Theorem 2 in Johnson and Newman~\cite{Johnson1980}]\label{prop:cospeccomain} 
	Two co-spectral graphs $G$ and $H$ of the same order are co-main if and only if there exists a doubly quasi-stochastic orthogonal matrix $O$ such that $A_G\mdot O=O\mdot A_H$. ~\hfill$\qed$ 
\end{proposition}

In other words, $G\equiv_{\ML{\cdot,\,\tr,\change{{}^*},\ones}} H$ implies the existence of a doubly quasi-stochastic orthogonal matrix $O$ such that  $A_G\mdot O=O\mdot A_H$. We further observe that $A_H\mdot O^*=O^*\mdot A_G$ and that $O^*$ is a doubly quasi-stochastic orthogonal matrix as well.
 We can now use Lemmas~\ref{lem:multp-sim},~\ref{lem:trace-sim} and~\ref{lem:ones-sim} to show the converse. Indeed, these lemmas combined tell us that $A_G\mdot O=O\mdot A_H$ implies that $e(A_G)=e(A_H)$ for any sentence $e(X)$ in $\ML{\cdot,\,\tr,\change{{}^*,}\ones}$. As a consequence, we have the following proposition.
\begin{proposition}\label{prop:traceones} 
	For two graphs $G$ and $H$ of the same order, $G\equiv_{\ML{\cdot,\,\tr,\change{{}^*},\ones}} H$ if and only if $G$ and $H$ have the same number of closed walks of any length, and the same number of walks of any length, if and only if $A_G\mdot O=O\mdot A_H$ for a doubly quasi-stochastic orthogonal matrix $O$. ~\hfill$\qed$ 
\end{proposition}
We can also phrase $\ML{\cdot,\,\tr,\change{{}^*},\ones}$-equivalence in terms of homomorphism vectors. That is, 
$G\equiv_{\ML{\cdot,\,\tr,\change{{}^*},\ones}} H$ if and only if $\textsf{HOM}_{\cal F}(G)=\textsf{HOM}_{\cal F}(H)$, where ${\cal F}$ now consists of all cycles and paths. This complements the results in Dell et al.~\cite{Dell2018}.

As a note aside, an alternative characterisation to Proposition~\ref{prop:cospeccomain} (Theorem 3 in van Dam et al.~\cite{Vandam2007}) is that $G$ and $H$ are co-spectral and co-main if and only if both $G$ and $H$ \emph{and} their complement graphs $\bar G$ and $\bar H$ are co-spectral. Here, the complement graph $\bar G$ of $G$ is the graph with adjacency matrix given by $J-A_G-I$, and similarly for $\bar H$. 
\begin{example}[Continuation of Example~\ref{ex:cospecvsallwalks}] \label{ex:graphs4}
	\normalfont Consider the subgraph $G_4$ (\!\raisebox{-1.3ex}{\mbox{ 
	\includegraphics[height=0.6cm]{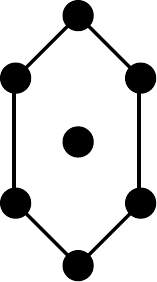}}}\hspace{.08em}) of $G_2$ and the subgraph $H_4$ (\!\raisebox{-1.1ex}{\mbox{ 
	\includegraphics[height=0.6cm]{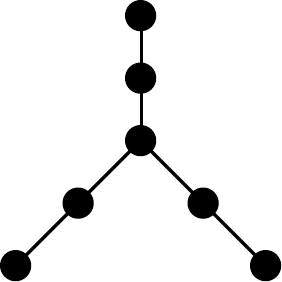}}}\hspace{.08em}) of $H_2$. These are known to be the smallest non-isomorphic co-spectral graphs with co-spectral complements (see e.g., Figure 4 in~\cite{Haemers2004}). From the previous remark it  follows that $G_4$ and $H_4$ have the same number of  walks of any length and the same number of closed walks of any length. These graphs are thus indistinguishable by sentences in $\ML{\cdot,{}^*,\ones}$ and $\ML{\cdot,\tr,\change{{}^*},\ones}$. Combined with our observation in Example~\ref{ex:stoch} that also $G_3$ and $H_3$ have the same number of walks, we conclude that the disjoint unions $G_2=G_3\cup G_4$ (\!\raisebox{-1.3ex}{\mbox{ 
	\includegraphics[height=0.6cm]{graphG2}}}\hspace{.08em}) and $H_2=H_3\cup H_4$ (\!\raisebox{-1.3ex}{\mbox{ 
	\includegraphics[height=0.6cm]{graphH2}}}) have the same number of walks of any length, as anticipated in Example~\ref{ex:cospecvsallwalks}. ~\hfill$\qed$ 
\end{example}

Clearly,
$G\equiv_{\ML{\cdot,\,\tr,\change{{}^*},\ones}} H$ implies $G\equiv_{\ML{\cdot,^*,\ones}} H$. We already mentioned in Example~\ref{ex:cospecvsallwalks} that the graphs $G_2$ (\!\raisebox{-1.3ex}{\mbox{ 
\includegraphics[height=0.6cm]{graphG2}}}) and $H_2$ (\!\raisebox{-1.3ex}{\mbox{ 
\includegraphics[height=0.6cm]{graphH2}}}) show that the converse does not hold.

We conclude again by observing that addition, scalar multiplication and pointwise function application on scalars can be added to $\ML{\cdot,^*,\ones}$ and $\ML{\cdot,\tr,\change{{}^*},\ones}$ at no increase in expressiveness. 
\begin{corollary}\label{cor:oneplusall} 
	Let $G$ and $H$ be two graphs of the same order. Then, 
	\begin{enumerate}
		\item[(1)] $G\equiv_{\ML{\cdot,{}^*,\ones,+,\times,\Apply_{\mathsf{s}}[f], f\in\Omega}} H$ if and only if $G\equiv_{\ML{\cdot,^*,\ones}} H$; and 
		\item[(2)] $G\equiv_{\ML{\cdot,\,\tr,{}^*,\ones,+,\times,\Apply_{\mathsf{s}}[f],f\in\Omega}} H$ if and only if $G\equiv_{\ML{\cdot,\,\tr,\change{{}^*,}\ones}} H$. 
	\end{enumerate}
\end{corollary}
\begin{proof}
	(1) We only need to show that $G\equiv_{\ML{\cdot,^*,\ones}} H$ implies $G\equiv_{\ML{\cdot,{}^*,\ones,+,\times,\allowbreak \Apply_{\mathsf{s}}[f], \allowbreak f\in\Omega}} H$. We have that $G\equiv_{\ML{\cdot,^*,\ones}} H$ implies $A_G\mdot Q=Q\mdot A_H$ for a doubly quasi-stochastic matrix $Q$ (Proposition~\ref{prop:quasistoch}). Furthermore, in the proof of Proposition~\ref{prop:quasistoch} we have shown that $A_H\mdot Q^*=Q^*\mdot A_G$ where $Q^*$ is again a doubly quasi-stochastic matrix. Lemmas~\ref{lem:multp-sim},~\ref{cor:ridoflinear},~\ref{lem:complextranspose-elim1},~\ref{lem:ridofscalarpointw} and~\ref{lem:ones-sim} imply that $Q$-conjugation and $Q^*$-conjugation are preserved by all operations in the fragment $\ML{\cdot,{}^*,\ones,+,\times,\Apply_{\mathsf{s}}[f], f\in\Omega}$.
	
	\smallskip 
	\noindent (2) We only need to show that $G\equiv_{\ML{\cdot,\,\tr,\change{{}^*,}\ones}} H$ implies $G\equiv_{\ML{\cdot,\,\tr,\,{}^*,\ones,+,\times,\Apply_{\mathsf{s}}[f],f\in\Omega}} H$. We have that $G\equiv_{\ML{\cdot,\,\tr,\change{{}^*,}\ones}} H$ implies $A_G\mdot O=O\mdot A_H$ for a doubly quasi-stochastic orthogonal matrix $O$ (Proposition~\ref{prop:traceones}). We observed earlier that $A_H\mdot O^*=O^*\mdot A_G$  and that $O^*$ is a doubly quasi-stochastic orthogonal matrix as well. Lemmas~\ref{lem:multp-sim}, \ref{lem:trace-sim}, \ref{cor:ridoflinear}, \ref{lem:complextranspose-elim1}, \ref{lem:ridofscalarpointw} and~\ref{lem:ones-sim}, imply that $O$-conjugation and $O^*$-conjugation are preserved by all operations in the fragment $\ML{\cdot,\,{}^*,\tr,\ones,+,\times,\Apply_{\mathsf{s}}[f],f\in\Omega}$.
	
	In both cases, we can therefore conclude, by induction on the structure of expressions, that for any sentence $e(X)$, $e(A_G)$ and $e(A_H)$ are conjugate and hence, $e(A_G)=e(A_H)$. 
\end{proof}

As we will see later, including any other operation from Table~\ref{tbl:mloperations}, such as $\diag(\cdot)$ or pointwise function applications on vectors or matrices, allows us to distinguish $G_4$ and $H_4$. We recall from Example~\ref{ex:graphs4} that these graphs cannot be distinguished by sentences in $\ML{\cdot,{}^*,\ones}$ and $\ML{\cdot,\tr,\change{{}^*,}\ones}$.

\section{The impact of the \texorpdfstring{$\diag(\cdot)$}{diag(.)}  operation}\label{sec:diag} 
We next consider the operation $\diag(\cdot)$ which takes as a column vector as input and returns the diagonal matrix with the input vector on its diagonal.
\footnote{The $\diag(\cdot)$ operation is also defined for $1\times1$-matrices (scalars) in which case it just returns that scalar.} 
\change{From the fragments considered so far, 
$\ML{\cdot,{}^*,\ones}$ and $\ML{\cdot,\tr,\change{{}^*,}\ones}$ are the only fragments in which vectors can be defined and for which the inclusion of $\diag(\cdot)$ has an impact.
Therefore, in this section we consider equivalence with regards to  $\ML{\cdot,{}^*,\ones,\diag}$ and $\ML{\cdot,\tr,\change{{}^*,}\ones,\diag}$. This section is organised as follows. First, we illustrate what information can be obtained from graphs using the
$\diag(\cdot)$ operation (Section~\ref{subsec:diag_ex}). We then show in Section~\ref{subsec:diag_equit} that one can compute so-called equitable partitions of graphs. From this, we can infer that equivalence of graphs with regards to the fragments under consideration implies
that the graphs have a common equitable partition. We use this observation in Sections~\ref{subsec:diag_notr} and~\ref{sub-sec:trdiag} to characterise $\ML{\cdot,{}^*,\ones}$- and $\ML{\cdot,\tr,\change{{}^*,}\ones}$-equivalence, respectively.
Finally, in Section~\ref{subsec:diag_funct_w} we show that we can add pointwise function applications on vectors without increasing the distinguishing power of the fragments.}

\subsection{\change{Example of the impact of the presence of  \texorpdfstring{$\diag(\cdot)$}{diag(.)}}}\label{subsec:diag_ex}
Using $\diag(\cdot)$ we can extract new information from graphs, as is illustrated in the following example. 
\begin{example}\label{ex:degree3} 
	\normalfont Consider graphs $G_4$ (\!\raisebox{-1.3ex}{\mbox{ 
	\includegraphics[height=0.6cm]{graphG4}}}\hspace{.08em}) and $H_4$ (\!\!\raisebox{-1.3ex}{\mbox{ 
	\includegraphics[height=0.6cm]{graphH4}}}\!). In $G_4$ we have vertices of degrees $0$ and $2$, and in $H_4$ we have vertices of degrees $1$, $2$ and $3$. We will count the number of vertices of degree $3$. Given that we know that $3$ is an upper bound on the degrees of vertices in $G_4$ and $H_4$, we consider the sentence $\#\mathsf{3degr}(X)$ given by
	\begin{linenomath}
		\postdisplaypenalty=0 
		\begin{multline*}
			\left(\frac{1}{6}\right)\times (\ones(X))^*\cdot\bigl(\diag(X\cdot\ones(X)- 0\times\ones(X))\cdot\diag(X\cdot\ones(X)-1\times \ones(X))\\
			{}\mdot \diag(X\cdot\ones(X)-2\times\ones(X))\bigr)\cdot\ones(X), 
		\end{multline*}
	\end{linenomath}
	in which we, for convenience, allow addition and scalar multiplications. Each of the subexpressions $\diag(X\cdot\ones(X)-d\times \ones(X))$, for $d=0,1$ and $2$, sets the diagonal entry corresponding to vertex $v$ to $0$ when $v$ has degree $d$. By taking the product of these diagonal matrices, entries that are set to $0$ will remain zero in the resulting diagonal matrix. This implies that, after taking these products, the only non-zero diagonal entries are those corresponding to vertices of degree different from $0$, $1$ or $2$. In other words, only for vertices of degree $3$ the diagonal entries carry a non-zero value, i.e., the value $6=(3-0)(3-1)(3-2)$. By appropriately rescaling by the factor $\frac{1}{6}$, the diagonal entries for the degree three vertices are set to $1$, and then summed up. Hence, $\#\mathsf{3degr}(X)$ indeed counts the number vertices of degree three when evaluated on adjacency matrices of graphs with vertices of maximal degree $3$. Since $\#\mathsf{3degr}(A_{G_4})=[0]\neq [1]=\#\mathsf{3degr}(A_{H_4})$ we can distinguish $G_4$ and $H_4$. We can obtain similar expressions for $
	\#d\!\mathsf{degr}(X)$ for arbitrary $d$, provided that we know the maximal degree of vertices in the graph. The way that these expressions are constructed is similar to the so-called Schur-Wielandt Principle indicating how to extract entries from a matrix that hold specific values by means of pointwise multiplication of matrices (see e.g., Proposition 1.4 in~\cite{Pech2002}). Here, we do not have pointwise matrix multiplication available but since we extract information from vectors, pointwise multiplication of vectors is simulated by normal matrix multiplication of diagonal matrices with the vectors on their diagonals. ~\hfill$\qed$ 
\end{example}

The use of the diagonal matrices and their products as in our example sentence $\#\mathsf{3degr}(X)$ can also be generalised to obtain information about so-called \emph{iterated degrees} of vertices in graphs, e.g., to identify and/or count vertices that have a number of neighbours each of which have neighbours of specific degrees, and so on. Such iterated degree information is closely related to \emph{equitable partitions} and \emph{fractional isomorphisms} of graphs (see e.g., Chapter 6 in~\cite{Scheinerman1997}). We phrase our results in terms of equitable partitions instead of iterated degree sequences.

\subsection{Equitable partitions} \label{subsec:diag_equit}
Formally, an \emph{equitable partition} ${\cal V}=\{V_1,\ldots,V_\ell\}$ of $G$ is  a partition of the vertex set $V$ of $G$ such that for all $i,j=1,\ldots,\ell$ and $v,v'\in V_i$, $\degr(v,V_j)=\degr(v',V_j)$. Here, $\degr(v,V_j)$ is the number of vertices in $V_j$ that are adjacent to $v$. In other words, an equitable partition is such that the graph is regular within each part, i.e., all vertices in a part have the same degree, and is bi-regular between any two different parts, i.e., the number of edges between any two vertices in two different parts is constant. A graph always has a \emph{trivial} equitable partition: simply treat each vertex as a part by its own. More interesting is the \emph{coarsest} equitable partition of a graph, i.e., the \emph{unique} equitable partition for which any other equitable partition of the graph is a refinement thereof~\cite{Scheinerman1997}. 

\change{The conditions underlying equitable partitions can be equivalently stated in terms of adjacency matrices and indicator vectors describing the partitions. We first introduce the notion of indicator vector.
Let $G=(V,E)$ be a graph of order $n$ with $V=\{1,\ldots,n\}$.
Let $V'\subseteq V$. We denote the \textit{indicator vector of $V'$} as the column vector $\ones_{V'}$ in  $\R^{n\times 1}$ and defined such that $(\ones_{V'})_v=1$ if $v\in V'$ and $(\ones_{V'})_v=0$ otherwise. Then, given a partition ${\cal V}=\{V_1,\ldots,V_\ell\}$ of $V$ we represent ${\cal V}$ by the $\ell$ indicator vectors $\ones_{V_i},\ldots,\ones_{V_\ell}$.
We observe that $\ones=\sum_{i=1}^\ell \ones_{V_i}$ due to ${\cal V}$ being a partition. We can now express that 
 ${\cal V}$ is an equitable partition in linear algebra terms. More precisely, ${\cal V}$ is an equitable partition  of $G$ if and only if for all $i,j=1,\ldots,\ell$, 
\begin{linenomath}
	\[\diag(\ones_{V_i})\mdot A_G\cdot\ones_{V_j}=\degr(v,V_j)\times\ones_{V_i},\]
\end{linenomath}
for some (arbitrary) vertex $v\in V_i$.}

\change{We next use the standard terminology for relating two graphs in terms of their equitable partitions~\cite{Scheinerman1997}.} More specifically, two graphs $G$ and $H$ are said to have a \emph{common equitable partition} if there exists an equitable partition ${\cal V}=\{V_1,\ldots,V_\ell\}$ of $G$ and an equitable partition ${\cal W}=\{W_1,\ldots,W_\ell\}$ of $H$ such that (a)~the sizes of the parts agree, i.e., $|V_i|=|W_i|$ for each $i=1,\ldots,\ell$, and (b)~$\degr(v,V_j)=\degr(w,W_j)$ for any $v\in V_i$ and $w\in W_i$ and any $i,j=1,\ldots,\ell$. \change{When ${\cal V}$ and ${\cal W}$
satisfy these conditions, we say that these partitions \textit{witness} that $G$ and $H$ have a common equitable partition.}
\change{Similarly, two graphs are said to have a \textit{common coarsest equitable partition} if the partitions ${\cal V}$ of $G$ and ${\cal W}$ of $H$ mentioned above are the coarsest equitable partitions of $G$ and $H$, respectively.}
Proposition~\ref{prop:fractiso} below characterises when two graphs do have a common equitable partition. Furthermore, when two graphs have a common equitable partition they also have a common coarsest equitable partition (see e.g., Theorem 6.5.1 in~\cite{Scheinerman1997}).

Equitable partitions naturally arise as the result of the \emph{colour refinement procedure}~\cite{Bastert2001,Grohe2014,Weisfeiler1968}, also known as the $1$-dimensional Weisfeiler-Lehman (\textsf{1WL}) algorithm, used as a subroutine in graph isomorphism solvers. Furthermore, there is a close connection to the study of \emph{fractional isomorphisms} of graphs~\cite{Scheinerman1997,Tinhofer1986}, as already mentioned in the introduction. We recall: two graphs $G$ and $H$ are said to be fractional isomorphic if there exists a doubly stochastic matrix $S$ such that $A_G\mdot S=S\mdot A_H$. Furthermore, a logical characterisation of graphs with a common equitable partition exists, as is stated next. 
\begin{proposition}[Theorem 1 in Tinhofer~\cite{Tinhofer1986}, Section 4.8 in Immerman and Lander~\cite{Immerman1990}]\label{prop:fractiso} 
	Let $G$ and $H$ be two graphs of the same order. Then, $G$ and $H$ are fractional isomorphic if and only if $G$ and $H$ have a common equitable partition if and only if $G\equiv_{\CLK{2}} H$. ~\hfill$\qed$ 
\end{proposition}
\begin{example}\label{ex:fractiso} 
	\normalfont The matrix linking the adjacency matrices of $G_3$ (\!\raisebox{-1.3ex}{\mbox{ 
	\includegraphics[height=0.6cm]{graphG3}}}\hspace{.08em}) and $H_3$ (\!\raisebox{-1.3ex}{\mbox{ 
	\includegraphics[height=0.6cm]{graphH3}}}\hspace{.08em}) in Example~\ref{ex:stoch} is in fact a doubly stochastic matrix (all its entries are either $0$ or $\frac{1}{2}$). Hence, $G_3$ and $H_3$ have a common equitable partition by Proposition~\ref{prop:fractiso}. 
	\change{One can alternatively verify that the partitions of $G_3$ and $H_3$
	consisting of a single part containing all the vertices of $G_3$ and $H_3$, respectively, witness that $G_3$ and $H_3$ have a common equitable partition. }
	 By contrast, graphs $G_2$ (\!\raisebox{-1.3ex}{\mbox{ 
	\includegraphics[height=0.6cm]{graphG2}}}\hspace{.08em}) and $H_2$ (\!\raisebox{-1.3ex}{\mbox{ 
	\includegraphics[height=0.6cm]{graphH2}}}) do not have a common equitable partition. Indeed, fractional isomorphic graphs must have the same multi-set of degrees, i.e., the same multi-set consisting of the degrees of vertices (Proposition 6.2.6 in~\cite{Scheinerman1997}), which does not hold for $G_2$ and $H_2$. Indeed, we note that there is an isolated vertex in $G_2$ but not in $H_2$. For the same reason, $G_1$ (\!\raisebox{-0.4ex}{\mbox{ 
	\includegraphics[height=0.3cm]{graphG1}}}\hspace{.08em}) and $H_1$ (\!\raisebox{-0.4ex}{\mbox{ 
	\includegraphics[height=0.3cm]{graphH1}}}\hspace{.08em}), and $G_4$ (\!\raisebox{-1.3ex}{\mbox{ 
	\includegraphics[height=0.6cm]{graphG4}}}\hspace{.08em}) and $H_4$ (\!\raisebox{-1.3ex}{\mbox{ 
	\includegraphics[height=0.6cm]{graphH4}}}) are not fractional isomorphic. ~\hfill$\qed$ 
\end{example}

To relate equitable partitions to $\ML{\cdot,{}^*,\ones,\diag}$- and $\ML{\cdot,\tr,\change{{}^*,}\ones,\diag}$-equi\-va\-lence, we show that the presence of $\diag(\cdot)$ allows us to formulate a number of expressions, denoted by $\mathsf{eqpart}_i(X)$, for $i=1,\ldots,\ell$, that together extract the \emph{coarsest equitable partition} from a given graph. By evaluating these expressions on $A_G$ and $A_H$, one can use additional sentences to detect whether these partitions witness that $G$ and $H$ have a common equitable partition. In this subsection, ${\cal L}$ can be either  $\{\cdot,{}^*,\ones,\diag\}$ or $\{\cdot,\tr,\change{{}^*,}\ones,\diag\}$. 

\SetAlgoCaptionLayout{small} \SetAlgoVlined 
\begin{algorithm}
	[t] 
	\caption{\normalfont\small Computing the coarsest equitable partition based on algorithm \textsc{GDCR}~\cite{Kersting2014}.} \label{alg:eqpart} \SetKwInOut{Input}{Input}\SetKwInOut{Output}{Output} \ResetInOut{Output} \Input{Adjacency matrix $A_G$ of $G$ of dimension $n\times n$.} \Output{Indicator vectors of coarsest equitable partition of $G$.} Let $B^{(0)}:=\ones$\; Let $i=1$\; \While{$i\leq n$}{ Let $M^{(i)}:=A_G\mdot B^{(i-1)}$\; Let ${\cal V}^{(i)}:=\{V_1^{(i)},\ldots,V_{\ell_i}^{(i)}\}$ a partition such that $v,w \in V_j^{(i)}$ if and only if $M^{(i)}_{v*}=M^{(i)}_{w*}$\; Let $B^{(i)}:=[\ones_{V_1^{(i)}},\ldots,\ones_{V_{\ell_i}^{(i)}}]$\; Let $i=i+1$\; } Return $B^{(n)}$. 
\end{algorithm}
\begin{proposition}\label{prop:equipart} 
	Let $G$ and $H$ be two graphs of the same order. Then, $G\equiv_{\ML{{\cal L}}} H$ implies that $G$ and $H$ have a common equitable partition. 
\end{proposition}
\begin{proof}
\change{The proof is quite lengthy so we first describe its structure.
\begin{itemize}
\item[(a)] We first argue that we can use addition and scalar multiplication at no increase in distinguishing power. This will simplify the construction of the expressions later on. We denote by ${\cal L}^+$ the extension of ${\cal L}$ with $+$ and $\times$. 
\item[(b)] We then show how to construct a number of expressions, denoted by $\mathsf{eqpart}_i(X)$, for $i=1,\ldots,\ell$, in $\ML{{\cal L}^+}$. The key property of these expressions is that when they are evaluated on the adjacency matrix $A_G$ of $G$, 
$\mathsf{eqpart}_i(A_G)$, for $i=1,\ldots,\ell$, correspond to indicator vectors representing an equitable partition of $G$. 
\item[(c)] 
The construction of the expressions $\mathsf{eqpart}_i(X)$, for $i=1,\ldots,\ell$, depend on $A_G$. As such, it is not guaranteed that $\mathsf{eqpart}_i(A_H)$, for $i=1,\ldots,\ell$, correspond to indicator vectors representing an equitable partition of $H$. We show, however, that when  $G\equiv_{\ML{{{\cal L}^+}}} H$ holds, then $\mathsf{eqpart}_i(A_H)$, for $i=1,\ldots,\ell$, indeed correspond to indicator vectors representing an equitable partition of $H$. To show this, we construct a number of sentences in $\ML{{\cal L}^+}$.
\item[(d)] Finally, we observe that the partitions represented by 
$\mathsf{eqpart}_i(A_G)$ and $\mathsf{eqpart}_i(A_H)$, for $i=1,\ldots,\ell$, witness that $G$ and $H$ have a common equitable partition.
\end{itemize}
Hence, all combined, this suffices to conclude that $G\equiv_{\ML{{\cal L}^+}} H$ implies that $G$ and $H$ have a common equitable partition. Given (a), the same conclusion holds for $G\equiv_{\ML{{\cal L}}} H$.}

\medskip	
\noindent
\change{\textbf{(a)~Showing that $G\equiv_{\ML{{\cal L}}} H$ if and only if $G\equiv_{\ML{{\cal L}^+}} H$.}
As mentioned earlier, it will be convenient to use addition and scalar multiplication
in our expressions. Clearly, $G\equiv_{\ML{{\cal L}^+}} H$ implies $G\equiv_{\ML{\cal L}} H$. Hence, it suffices to show that $G\equiv_{\ML{{\cal L}}} H$  implies $G\equiv_{\ML{{\cal L}^+}} H$.}\footnote{We remark that we cannot rely yet on the conjugation-preservation Lemma~\ref{cor:ridoflinear} to show that $G\equiv_{\ML{{\cal L}^+}} H$ if and only if $G\equiv_{\ML{\cal L}} H$. Indeed, at this point we do not know yet for what kind of matrices $T$, $T$-conjugation is preserved by the $\diag(\cdot)$-operation. This will only be settled in Lemma~\ref{lem:diag-sim} later in this section.}
\change{To see this, we observe that any expression in $\ML{{\cal L}^+}$ 
can be written as a linear combination of expressions in $\ML{\cal L}$. For completeness, we verify this  in the appendix.}
 \change{Consider now  a sentence $e(X)$ in $\ML{{\cal L}^+}$.
Hence, $e(X)$ can be written as $\sum_{i=1}^p a_i\times e_i(X)$
with $a_i\in \C$ and sentences $e_i(X)\in\ML{\cal L}$, for $i=1,\ldots,p$. By assumption, $e_i(A_G)=e_i(A_H)$, for $i=1,\ldots,p$, Hence also $e(A_G)=\sum_{i=1}^p a_i\times e_i(A_G)=\sum_{i=1}^p a_i\times e_i(A_H)=e(A_H)$.}
As a consequence, $G\equiv_{\ML{{\cal L}}} H$ indeed implies $G\equiv_{\ML{{{\cal L}^+}}} H$.
	
\medskip	
\noindent
\change{\textbf{(b)~Computing the coarsest equitable partition of $G$ by expression in $\ML{{\cal L}^+}$.}}
\change{We next show that we can compute the indicator vectors of an equitable partition of $G$ by means of expressions in $\ML{{\cal L}^+}$. } To see this, we implement the algorithm \textsc{GDCR} for finding this partition~\cite{Kersting2014}. We recall this algorithm (in a slightly different form than presented in Kersting et al.~\cite{Kersting2014}) in Algorithm~\ref{alg:eqpart}. In a nutshell, the algorithm takes as input $A_G$, the adjacency matrix of $G$, and returns a matrix whose columns hold indicator vectors that represent an equitable partition of $G$.
	
	The algorithm starts, on line 1, by creating a partition consisting of a single part containing all vertices, represented by the indicator vector $\ones$, and stored in vector $B^{(0)}$. Then, in the $i^{\text{th}}$ step, the current partition is represented by $\ell_{i-1}$ indicator vectors $\ones_{V_1^{(i-1)}},\ldots,\ones_{V_{\ell_{i-1}}^{(i-1)}}$ which constitute the columns of matrix $B^{(i-1)}$. The refinement of this partition is then computed in two steps. First, the matrix $M^{(i)}:=A_G\mdot B^{(i-1)}$ (line 4) is computed; \change{s}econd, each $\ones_{V_j^{(i-1)}}$ is refined by putting vertices $v$ and $w$ in the same part if and only if they have the same rows in $M^{(i)}$, i.e., when $M_{v*}^{(i)}=M_{w*}^{(i)}$ holds (line 5). The corresponding partition ${\cal V}^{(i)}$ is then represented by, say $\ell_i$, indicator vectors and stored as the columns of $B^{(i)}$ (line 6). This is repeated until no further refinement of the partition is obtained. At most $n$ iterations are needed. The correctness of the algorithm is established in~\cite{Kersting2014}. That is, the resulting indicator vectors represent an equitable partition of $G$. In fact, they represent the coarsest equitable partition of $G$.
	
	We next detail how a run of the algorithm on adjacency matrix $A_G$ can be simulated using expressions in $\ML{{\cal L}^+}$. The initialisation step is easy: We compute $B^{(0)}$ by means of the expression $b^{(0)}(X):=\ones(X)$. Clearly, $B^{(0)}=b^{(0)}(A_G)$. Next, suppose by induction that we have $\ell_{i-1}$ expressions $b^{(i-1)}_1(X),\ldots, b^{(i-1)}_{\ell_{i-1}}(X)$ such that when these expressions are evaluated on $A_G$, they return the indicator vectors stored in the columns of $B^{(i-1)}$. That is, $\ones_{V_j^{(i-1)}}=b^{(i-1)}_j(A_G)$ for all $j=1,\ldots,\ell_{i-1}$. We next show how the $i^{\text{th}}$ iteration is simulated. 
	
	We first compute the $\ell_{i-1}$ vectors stored in the columns of $M^{(i)}$ (line 4). We compute  these column vectors one at a time. To this aim, we consider expressions 
	\begin{linenomath}
		\[ m^{(i)}_j(X):=X\mdot b^{(i-1)}_j(X), \quad \text{for } j=1,\ldots,\ell_{i-1}. \]
	\end{linenomath}
	Clearly, $m^{(i)}_j(A_G)=M^{(i)}_{*j}$, as desired.
	
	A bit more challenging is the computation of the refined partition in ${\cal V}^{(i)}$ (line 5) since we need to inspect all columns $M^{(i)}_{*j}$ and identify rows on which all these columns agree, as explained above. It is here that the $\diag(\cdot)$ operation plays a crucial role. Moreover, to compute this refined partition we need to know all values occurring in $M^{(i)}$. The expressions below depend on these values and hence on the input adjacency matrix $A_G$.

	Let $D^{(i)}_j$ be the set of values occurring in the column vector $M^{(i)}_{*j}$, for $j=1,\ldots,\ell_{i-1}$. We compute, by means of an $\ML{{\cal L}^+}$ expression, an indicator vector which identifies the rows in $M^{(i)}_{*j}$ that hold a specific value $c\in D^{(i)}_j$. This expression is similar to the one used in Example~\ref{ex:degree3} to extract vertices of degree 3 from the degree vector. More precisely, we consider expressions 
	\begin{linenomath}
		\[ \ones_{=c}^{(i),j}(X)=\left(\frac{1}{\prod_{c'\in D^{(i)}_j,c'\neq c}(c-c')}\right)\times\biggl(\Bigl(\!\!\!\!\!\!\prod_{c'\in D^{(i)}_j\!\!, c'\neq c}\diag\bigl(m^{(i)}_j(X)- c'\times\ones(X)\bigr)\Bigr)\cdot\ones(X)\biggr), \]
	\end{linenomath}
	for the current iteration $i$, column $j$ in $M^{(i)}$, and value $c\in D^{(i)}_j$. The correctness of these expressions follows from a similar explanation as given in Example~\ref{ex:degree3}. Given these expressions, one can now easily obtain an indicator vector identifying all rows in $M^{(i)}$ that hold a specific value combination $(c_1,\ldots,c_{\ell_{i-1}})$ in their columns, where each $c_j\in D^{(i)}_j$, as follows: 
	\begin{linenomath}
		\[ \ones_{=(c_1,\ldots,c_{\ell_{i-1}})}^{(i)}(X)=\diag(\ones_{=c_1}^{(i),1}(X))\cdot\cdots\mdot \diag(\ones_{=c_{\ell_{i-1}}}^{(i),\ell_{i-1}}(X))\cdot\ones(X).\]
	\end{linenomath}
	That is, we simply take the boolean conjunction of all indicator vectors $\ones_{=c_j}^{(i),j}(X)$, for $j=1,\ldots,\ell_{i-1}$. We note that $\ones_{=(c_1,\ldots,c_{\ell_{i-1}})}^{(i)}(A_G)$ may return the zero vector, i.e., when $(c_1,\ldots,c_{\ell_{i-1}})$ does not occur as a row in $M^{(i)}$.
	We are only interested in  value combinations that do occur. Suppose that there are $\ell_i$ distinct value combinations $(c_1,\ldots,c_{\ell_{i-1}})$ for which $\ones_{=(c_1,\ldots,c_{\ell_{i-1}})}^{(i)}(A_G)$ returns a non-zero indicator vector. We denote by $b^{(i)}_1(X),\ldots, \allowbreak b^{(i)}_{\ell_i}(X)$ the corresponding expressions of the form $\ones_{=(c_1,\ldots,c_{\ell_{i-1}})}^{(i)}(X)$. It should be clear that $b^{(i)}_1(A_G),\ldots, \allowbreak b^{(i)}_{\ell_i}(A_G)$ are indicator vectors corresponding to the refined partition ${\cal V}^{(i)}$ as stored in $B^{(i)}$. This concludes the simulation of the $i^{\text{th}}$ iteration of the algorithm.
	
Finally, after the $n^{\text{th}}$ iteration we define 
	\begin{linenomath}
		\[ \mathsf{eqpart}_i(X):=b^{(n)}_i(X), \]
	\end{linenomath}
	for $i=1,\ldots,\ell_n$. In the following, we denote $\ell_n$ by $\ell$. We remark once more that all expressions defined above depend on the input $A_G$, as their definitions rely on the values occurring in the matrices $M^{(i)}$ computed along the way. 

\medskip	
\noindent
\change{\textbf{(c)~Checking that $G\equiv_{\ML{{\cal L}^+}} H$ implies that 
$\mathsf{eqpart}_i(A_H)$, for $i=1,\ldots,\ell$, also represent an equitable partition of $H$.}	
	Recall that we want to show that if $G\equiv_{\ML{{\cal L}^+}} H$ holds, then $G$ and $H$ have a common equitable partition. As a first step we verify that 
	$G\equiv_{\ML{{\cal L}^+}} H$ implies that the vectors $\mathsf{eqpart}_i(A_H)$, for $i=1,\ldots,\ell$, correspond to an equitable partition of $H$.}
	The challenge is to check all this by means of sentences in $\ML{{\cal L}^+}$.
We will use the following sentences.

	\begin{enumerate}
		\item For each $i=1,\ldots,\ell$, we first check whether $\mathsf{eqpart}_i(A_H)$ is also a binary vector.
		 We note that, by construction of the expression $\mathsf{eqpart}_i(X)$, $\mathsf{eqpart}_i(A_H)$ returns a real vector. To check whether every entry in $\mathsf{eqpart}_i(A_H)$ is either $0$ or $1$, we show that all of its entries must satisfy the equation $x(x-1)=0$. To this aim, we consider the $\ML{{\cal L}^+}$ sentence
		\begin{linenomath}
			\[ \mathsf{binary\_diag}(X):=(\ones(X))^*\cdot\bigl( (X\mdot X- X)\mdot (X\mdot X-X)\bigr)\mdot \ones(X). \]
		\end{linenomath}
		We claim that if $X$ is assigned a diagonal real matrix, say $\Delta$, then $\mathsf{binary\_diag}(\Delta)=[0]$ if and only if $\Delta$ is a \textit{binary} diagonal matrix. 
		
		Indeed, if $\Delta$ is a binary diagonal matrix, then $\Delta\cdot\Delta=\Delta$, $\Delta\cdot\Delta-\Delta=Z$, where $Z$ is the zero matrix, and hence $\mathsf{binary\_diag}(\Delta)=\tp{\ones}\mdot Z\mdot Z\cdot\ones=[0]$. Conversely, suppose that $\mathsf{binary\_diag}(\Delta)=[0]$. We observe that $(\Delta\mdot \Delta- \Delta)\mdot (\Delta\mdot \Delta-\Delta)$ is a diagonal matrix with squared real numbers on its diagonal. Hence, $\mathsf{binary\_diag}(\Delta)=[0]$ implies that the sum of the (squared real) diagonal elements in $\Delta\mdot \Delta- \Delta$ is $0$. This in turn implies that every element on the diagonal in $\Delta\cdot\Delta-\Delta$ must be zero. Hence, every element on $\Delta$'s diagonal must satisfy the equation $x^2-x=0$, implying that either $x=0$ or $x=1$. As a consequence, $\Delta$ is a binary diagonal matrix.
		
We now observe that $ \mathsf{binary\_diag}(\diag(\mathsf{eqpart}_i(A_G)))=[0]$ since  $\mathsf{eqpart}_i(A_G)$ returns an indicator vector. Then, $G\equiv_{\ML{{\cal L}^+}} H$ implies that the equality 
		\begin{linenomath}
			\[ \mathsf{binary\_diag}\bigl(\diag(\mathsf{eqpart}_i(A_G))\bigr)=[0]=\mathsf{binary\_diag}\bigl(\diag(\mathsf{eqpart}_i(A_H))\bigr),\]
		\end{linenomath}
		must hold, for all $i=1,\ldots,\ell$. Hence, the matrices $\diag(\mathsf{eqpart}_i(A_H))$ are indeed binary and so are its diagonal elements described by $\mathsf{eqpart}_i(A_H)$, as desired.

		\item We next verify that all indicator vectors $\mathsf{eqpart}_i(A_H)$ combined represent a partition of the vertex set of $H$. To verify this partition condition, 
\change{we simply consider the sentence $\ones(X)^*\mdot\bigl(\sum_{i=1}^\ell \mathsf{equit}_i(X)\bigr)$. Since $G\equiv_{\ML{{\cal L}^+}} H$ implies
		\begin{linenomath}
			\[ \ones(X)^*\cdot\big(\sum_{i=1}^\ell\mathsf{equit}_{i}(A_G)\bigr)=[n]=
		\ones(X)^*\cdot\big(\sum_{i=1}^\ell\mathsf{equit}_{i}(A_H)\bigr),\]
 		\end{linenomath}
it must be the case that the indicator vectors $\mathsf{eqpart}_i(A_H)$, for $i=1,\ldots,\ell$, form a partition.}
			\item We finally verify that the partition of $H$ represented by the indicator vectors $\mathsf{equit}_i(A_H)$, for $i=1,\ldots,\ell$, is an equitable partition of $H$. Let ${\cal V}=\{V_1,\ldots,V_\ell\}$ be the  equitable partition of $G$, represented by the indicator vectors $\mathsf{eqpart}_i(A_G)$, for $i=1,\ldots,\ell$. Similarly, let 
		${\cal W}=\{W_1,\ldots,W_\ell\}$ be the partition of $H$, represented by the indicator vectors $\mathsf{eqpart}_i(A_H)$, for $i=1,\ldots,\ell$. We show that ${\cal W}$ is an equitable partition of $H$.
				As already mentioned at the beginning of this section, we can rephrase ``being equitable'' in linear algebra terms. In particular, we know that for any $i,j=1,\ldots,\ell$, 
		\begin{linenomath}
			\[ \diag(\mathsf{eqpart}_i(A_G))\mdot A_G\mdot \diag(\mathsf{eqpart}_j(A_G))\cdot\ones-\deg(v,V_j)\times\mathsf{eqpart}_i(A_G) \]
		\end{linenomath}
		returns the zero vector, where $v$ is an arbitrary vertex in $V_i$, the part corresponding to the indicator vector $\mathsf{eqpart}_i(A_G)$. We want to check whether the same condition holds for $A_H$. We therefore consider the expressions $\mathsf{equi\_test}_{ij}(X)$, for $i,j=1,\ldots,\ell$ given by 
		\begin{linenomath}
			\[ \diag\Bigl(\diag(\mathsf{eqpart}_i(X))\mdot X\mdot \diag(\mathsf{eqpart}_j(X))\cdot\ones(X)-\deg(v,V_j)\times\mathsf{eqpart}_i(X)\Bigr) \]
		\end{linenomath}
		and check whether, when evaluated on $A_H$, the obtained diagonal matrix is the zero matrix. This would imply that for all $i,j=1,\ldots,\ell$,
		\begin{linenomath}
			\[ \diag(\mathsf{eqpart}_i(A_H))\mdot A_H\mdot \diag(\mathsf{eqpart}_j(A_H))\cdot\ones-\deg(v,V_j)\times\mathsf{eqpart}_i(A_H)\]
		\end{linenomath}
		 also returns the zero vector. As a consequence, ${\cal W}$ is an equitable partition of $H$. It rests us only to show that we can check, by means of sentences, whether a diagonal matrix is the zero matrix. We use the sentence 
		\begin{linenomath}
			\[ \mathsf{zerotest\_diag}(X):=(\ones(X))^*\mdot X\mdot X\cdot\ones(X), \]
		\end{linenomath}
		for this purpose. A similar argument as for the expression $\mathsf{binary\_diag}(X)$ shows that the $\mathsf{zerotest\_diag}(X)$ expression returns $[0]$ on diagonal real matrices if and only if the diagonal matrix is the zero matrix. We here again use that a sum of squares equals zero if and only if each summand is zero. Since $G\equiv_{\ML{{\cal L}^+}} H$, we have that  $$\mathsf{zerotest\_diag}(\mathsf{equi\_test}_{ij}(A_G))=[0]=\mathsf{zerotest\_diag}(\mathsf{equi\_test}_{ij}(A_H)),$$ for all $i,j=1,\ldots,\ell$, as desired.
	\end{enumerate}

\medskip	
\noindent
\change{\textbf{(d)~Verifying that $G$ and $H$ have a common equitable partition.}	
To conclude the proof we need to show that the equitable partitions ${\cal V}$ of $G$
and ${\cal W}$ of $H$, as defined above, witness that $G$ and $H$ have a common equitable partition. In other words, }we must have that (i)~$|V_i|=|W_i|$ for every $i=1,\ldots,\ell$; and
(ii)~for any $i,j=1,\ldots,\ell$, $\degr(v,V_j)=\degr(w,W_j)$ for any $v\in V_i$ and any $w\in W_i$. \change{We observe that the expressions $\mathsf{equi\_test}_{ij}(X)$ used above already verify condition (ii). We next show that, for each $i=1,\ldots,\ell$, the indicator vector $\mathsf{equit}_i(A_G)$ 
contains the same number of $1$'s as $\mathsf{eqpart}_i(A_G)$.	}
To show this, it suffices to consider the sentences $(\ones(X))^*\cdot\mathsf{eqpart}_i(X)$, for $i=1,\ldots,\ell$. Clearly, $G\equiv_{\ML{{\cal L}^+}} H$ implies that $\ones^*\cdot\mathsf{eqpart}_i(A_G)=\ones^*\cdot\mathsf{eqpart}_i(A_H)$, for $i=1,\ldots,\ell$. Hence, $\mathsf{eqpart}_i(A_H)$ and $\mathsf{eqpart}_i(A_G)$ contain the same number of ones. So, we may conclude that $G\equiv_{\ML{{\cal L}^+}} H$ implies indeed that $G$ and $H$ have a common equitable partition. Given that 
 $G\equiv_{\ML{{\cal L}^+}} H$ if and only if $G\equiv_{\ML{{\cal L}}} H$, also 
 $G\equiv_{\ML{{\cal L}}} H$ implies that $G$ and $H$ have a common equitable partition, as desired. \end{proof}

\subsection{Characterisation of \texorpdfstring{$\ML{\cdot,{}^*,\ones,\diag}$}{ML(.,*,1,diag)}-equivalence} 
\label{subsec:diag_notr}
\change{We next consider $\ML{\cdot,{}^*,\ones,\diag}$-equivalence.} We have just shown that $\ML{\cdot,{}^*,\ones,\diag}$-e\-qui\-va\-lent graphs have a common equitable partition. The converse also holds, as will be shown below (Proposition~\ref{lem:eqpartequalsent}).

\change{We first introduce a notion of \textit{compatibility} that will be used to define the proper notion of conjugation for  $\ML{\cdot,{}^*,\ones,\diag}$-equivalence.
Let $G$ and $H$ be two graphs of order $n$.
Let ${\cal V}=\{V_1,\ldots,V_\ell\}$ and ${\cal W}=\{W_1,\ldots,W_\ell\}$ be partitions of $G$ and $H$, respectively.
Furthermore, let $\ones_{V_i}$ and $\ones_{W_i}$, for $i=1,\ldots,\ell$, denote the indicator vectors corresponding to $V_i\in{\cal V}$ and $W_i\in{\cal W}$, respectively. Consider a matrix $T$ in  $\C^{n\times n}$. Then, $T$ is said to be \textit{compatible} with the partitions ${\cal V}$ of $G$ and ${\cal W}$ of $H$, if for $i=1,\ldots,\ell$,
	\begin{linenomath}
		\[ \diag(\ones_{V_i})\mdot T=T\cdot\diag(\ones_{W_i}). \]
	\end{linenomath}
Intuitively, this condition implies that $T$ has a block structure determined by the partitions ${\cal V}$ and ${\cal W}$ and only has non-zero blocks for blocks corresponding to the same parts $V_i$ and $W_i$ in these partitions. }

\begin{proposition}\label{lem:eqpartequalsent} 
	Let $G$ and $H$ be two graphs of the same order. If $G$ and $H$ have a common equitable partition, then $e(A_G)=e(A_H)$ for every sentence $e(X)$ in $\ML{\cdot,{}^*,\ones,\diag}$. 
\end{proposition}
\begin{proof}
	By assumption, $G$ and $H$ have a common equitable partition. 
	 Let ${\cal V}=\{V_1,\ldots,\allowbreak V_\ell\}$ and ${\cal W}=\{W_1,\ldots,W_\ell\}$ be partitions of $G$ and $H$, respectively, that witness this fact.
	 As before, we denote by $\ones_{V_i}$ and $\ones_{W_i}$, for $i=1,\ldots,\ell$, the corresponding indicator vectors. We know from Proposition~\ref{prop:fractiso} that there exists a doubly stochastic matrix $S$ such that $A_G\mdot S=S\mdot A_H$. As previously observed, also $A_H\mdot S^*=S^*\mdot A_G$ holds, where $S^*$ is again doubly stochastic. Then, Lemmas~\ref{lem:multp-sim},~\ref{lem:complextranspose-elim1} and~\ref{lem:ones-sim} imply that $S$-conjugation and $S^*$-conjugation are preserved by matrix multiplication, complex conjugate transposition, and the one-vector operation. To conclude that $G\equiv_{\ML{\cdot,{}^*,\ones,\diag}} H$ holds, we verify that the $\diag(\cdot)$ operation also preserves $S$- and $S^*$-conjugation. We rely on a more general result (Lemma~\ref{lem:diag-sim} below), which states that $T$-conjugation, for a matrix $T$, is preserved by the $\diag(\cdot)$ operation provided that $T$ is doubly quasi-stochastic and compatible \change{with equitable partitions of $G$ and $H$ that witness that $G$ and $H$ have a common equitable partition}. We again separate this lemma from the current proof because we need it later in the paper. 
When considering the doubly stochastic matrix $S$ such that $A_G\mdot S=S\mdot A_H$ holds, the matrix $S$ can be assumed to be compatible with ${\cal V}$ and $\cal W$. To see this, we recall from the proof of Theorem 6.5.1 in~\cite{Scheinerman1997} that we can take $S$ to be such that for $i\neq j$, $\diag(\ones_{V_i})\mdot S\mdot \diag(\ones_{W_j})$ is the $|V_i|\times |W_j|$ zero matrix, and for $i=j$, $\diag(\ones_{V_i})\mdot S\mdot \diag(\ones_{W_i})$ is the square $|V_i|\times |W_i|$ matrix in which all entries are equal to $\frac{1}{|V_i|}$. Hence,
$\diag(\ones_{V_i})\mdot S=S\cdot\diag(\ones_{W_i})$ for all $i,j=1,\ldots,\ell$.
	
	As a consequence, if $e_1(A_G)$ and $e_1(A_H)$ are $S$-conjugate, then Lemma~\ref{lem:diag-sim} implies that $\diag(e_1(A_G))$ and $\diag(e_1(A_H))$ are $S$-conjugate. We also note that $\diag(\ones_{W_i})\mdot S^*=S^*\cdot\diag(\ones_{V_i})$. So, $S^*$ is compatible with ${\cal W}$ and ${\cal V}$. An inductive argument then shows that $e(A_G)$ and $e(A_H)$ are $S$-conjugate (and thus equal) for any sentence $e(X)$ in $\ML{\cdot,{}^*,\ones,\diag}$, as desired. 
\end{proof}

\change{Let ${\cal V}$ and ${\cal W}$ be equitable partitions of $G$ and $H$, respectively, that witness that $G$ and $H$ have a common equitable partition.}
We next show that $T$-conjugation,  by means \change{of} doubly quasi-stochastic matrices $T$ that are compatible with ${\cal V}$ and ${\cal W}$,
 is indeed preserved by the $\diag(\cdot)$ operation. Showing this requires a bit more work than our previous conjugation-preservation results. 
 
\change{More precisely, a \textit{key property} underlying this preservation property} is the following: Any vector that can be obtained by evaluating expressions in $\ML{\cdot,{}^*,\ones,\diag}$ on an adjacency matrix of a graph, can  be written in a canonical way in terms of the indicator vectors representing an equitable partitions of that graph. We state this requirement for general matrix query languages, as follows.

Let $\ML{\cal L}$ be a matrix query language. 
We say that \textit{$\ML{\cal L}$-vectors are constant on equitable partitions} if, for any expression $e(X)\in \ML{\cal L}$, \change{any graph $G$}, and
\change{any equitable partition ${\cal V}=\{V_1,\ldots,V_\ell\}$ of $G$,}
when $e(A_G)$ is an $n\times 1$-vector then there exists  scalars $a_i\in\C$, for $i=1,\ldots,\ell$,
such that
\begin{linenomath}
	\begin{equation}
		e(A_G)=\sum_{i=1}^\ell a_i\times \ones_{V_i},
			\label{eq:det} \end{equation}
\end{linenomath}
holds. Here, \change{$\ones_{V_1},\ldots,\ones_{V_\ell}$ represent the equitable partition ${\cal V}$ of $G$.}

Intuitively, this condition is important for the $\diag(\cdot)$ operation since it takes a vector as input and the linear combination (\ref{eq:det}) allows one to only reason about (linear combinations of) diagonal matrices obtained by the indicator vectors of the equitable partitions. Compatibility in turn implies conjugation preservation for such (indicator vector-based) diagonal matrices, which can then be lifted, due to linearity, to conjugation of arbitrary diagonal matrices. \change{We make this argument more formal in the proof of the following Lemma.}
\begin{lemma}\label{lem:diag-sim} 
\change{Let $\ML{\cal L}$ be a matrix query language fragment such that $\ML{\cal L}$-vectors are constant on equitable partitions.
	Let $G$ and $H$ be two graphs of the same order which have a common equitable partition. 
	Furthermore, let ${\cal V}$ and ${\cal W}$ be equitable partitions of $G$ and $H$, respectively, that witness that $G$ and $H$ have a common equitable partition, and	let $T$ be a doubly quasi-stochastic matrix which is compatible with ${\cal V}$ and ${\cal W}$.}
Finally, let $e(X)$ be an expression in $\ML{\cal L}$. Then, if $e(A_G)$ and $e(A_H)$ are $T$-conjugate, then also $\diag(e(A_G))$ and $\diag(e(A_H))$ are $T$-conjugate. 
\end{lemma}
\begin{proof}
	Let $e(X)$ be an expression in $\ML{\cal L}$. Consider now $e'(X):=\diag(e(X))$. We distinguish between two cases, depending on the dimensions of $e(A_G)$. First, if $e(A_G)$ is a sentence then we know by induction that $e(A_G)=e(A_H)$. Hence, 
	\begin{linenomath}
		\[ e'(A_G)=\diag(e(A_G))=e(A_G)=e(A_H)=\diag(e(A_H))=e'(A_H). \]
	\end{linenomath}
	Next, if $e(A_G)$ is a column vector, then we know that $e(A_G)=T\mdot e(A_H)$ and furthermore, since $\ML{{\cal L}}$-vectors are constant on equitable partitions, that $e(A_G)=\sum_{i=1}^\ell a_i\times \ones_{V_i}$ and $e(A_H)=\sum_{i=1}^\ell b_i\times \ones_{W_i}$ \change{for some scalars $a_i$ and $b_i$  in $\C$, for $i=1,\ldots,\ell$. Here, $\ones_{V_i}$ and $\ones_{W_i}$, for $i=1,\ldots,\ell$,
	are the indicator vectors representing the equitable partitions ${\cal V}=\{V_1,\ldots,V_\ell\}$ of $G$ and ${\cal W}=\{W_1,\ldots,W_\ell\}$ of $H$, respectively.}
	We first show that $a_i=b_i$, for $i=1,\ldots,\ell$. Indeed, since $T\cdot\ones=\ones$ and $T$ is \change{compatible with ${\cal V}$ and ${\cal W}$}, we have that	\begin{linenomath}
		\[ \ones_{V_i}=\diag(\ones_{V_i})\cdot\ones=\diag(\ones_{V_i})\mdot T\mdot \ones=T\cdot\diag(\ones_{W_i})\cdot\ones=T\cdot\ones_{W_i}. \]
	\end{linenomath}
	As a consequence, using that $\tp{\ones_{V_i}}\mdot \ones_{V_j}^{\phantom{a}}$ is $0$ if $i\neq j$ and $|V_i|$ if $i=j$, we obtain 
	\begin{linenomath}
		\postdisplaypenalty=0 
		\begin{align*}
	a_i\times |V_i|&=\tp{\ones_{V_i}}\mdot e(A_G)=\tp{\ones_{V_i}}\mdot T\mdot e(A_H)\\
			&=\sum_{j=1}^\ell b_j\times(\tp{\ones_{V_i}}\mdot T\mdot \ones_{W_j}^{\phantom{a}})=b_i\times |W_i|, 
		\end{align*}
	\end{linenomath}
	for all $i=1,\ldots,\ell$. Since $|V_i|=|W_i|\neq 0$, we indeed have that $a_i=b_i$ for all $i=1,\ldots,\ell$.
	
	We may now conclude that 
	\begin{linenomath}
		\postdisplaypenalty=0 
		\begin{align*}
			e'(A_G)\mdot T&=\diag(e(A_G))\mdot T=\sum_{i=1}^\ell a_i\times  (\diag(\ones_{V_i})\mdot T)\\
			& =\sum_{i=1}^\ell a_i\times (T\mdot \diag(\ones_{W_i}))=T\mdot \diag(e(A_H))=T\mdot e'(A_H). 
		\end{align*}
	\end{linenomath}
	Hence $e'(A_G)$ and $e'(A_H)$ are indeed $T$-conjugate. 
\end{proof}

In the context of Proposition~\ref{lem:eqpartequalsent}, i.e., to show that the $\diag(\cdot)$ operation preserves $S$-con\-ju\-ga\-tion (and $S^*$-conjugation), we need to verify that $\ML{\cdot,{}^*,\ones,\diag}$-vectors are constant on equitable partitions. We verify this, in the appendix, by induction on the structure of expressions in $\ML{\cdot,{}^*,\ones,\diag}$.
 In fact, we more generally show the following.
\begin{proposition}\label{prop:determined1} 
	$\ML{\cdot,\tr, {}^*,\ones,\diag,+,\times,\Apply_{\mathsf{s}}[f],f\in\Omega}$-vectors are constant on equitable partitions.~\hfill$\qed$ 
\end{proposition}

All combined, we obtain the following characterisations. 
\begin{theorem}\label{thm:C2} 
	Let $G$ and $H$ be two graphs of the same order. Then, $G\equiv_{\ML{\cdot,^*,\ones,\diag}} H$ if and only if there is doubly stochastic matrix $S$ such that $A_G\mdot S=S\mdot A_H$ if and only if $G\equiv_{\CLK{2}} H$ if and only if $G$ and $H$ have a common equitable partition.~\hfill$\qed$ 
\end{theorem}
\begin{proof}
	This is a direct consequence of Propositions~\ref{prop:fractiso},~\ref{prop:equipart},~\ref{lem:eqpartequalsent} and~\ref{prop:determined1}.
\end{proof}
We further complement Theorem~\ref{thm:C2} by observing that $G\equiv_{\ML{\cdot,^*,\ones,\diag}} H$ if and only if $\textsf{HOM}_{\cal F}(G)=\textsf{HOM}_{\cal F}(H)$ where
${\cal F}$ consists of all trees~\cite{Dell2018}.

As a consequence, following Example~\ref{ex:fractiso}, sentences in $\ML{\cdot,^*,\ones,\diag}$ can distinguish $G_1$ (\!\raisebox{-0.4ex}{\mbox{ 
\includegraphics[height=0.3cm]{graphG1}}}\hspace{.08em}) and $H_1$ (\!\raisebox{-0.4ex}{\mbox{ 
\includegraphics[height=0.3cm]{graphH1}}}\hspace{.08em}), $G_2$ (\!\raisebox{-1.3ex}{\mbox{ 
\includegraphics[height=0.6cm]{graphG2}}}\hspace{.08em}) and $H_2$ (\!\raisebox{-1.3ex}{\mbox{ 
\includegraphics[height=0.6cm]{graphH2}}}), $G_4$ (\!\raisebox{-1.3ex}{\mbox{ 
\includegraphics[height=0.6cm]{graphG4}}}\hspace{.08em}) and $H_4$ (\!\raisebox{-1.3ex}{\mbox{ 
\includegraphics[height=0.6cm]{graphH4}}}\hspace{.08em}), because all these pairs of graphs do not have a common equitable partition. By contrast, $G_3$ (\!\raisebox{-1.3ex}{\mbox{ 
\includegraphics[height=0.6cm]{graphG3}}}\hspace{.08em}) and $H_3$ (\!\raisebox{-1.3ex}{\mbox{ 
\includegraphics[height=0.6cm]{graphH3}}}\hspace{.08em}) cannot be distinguished by sentences in $\ML{\cdot,^*,\ones,\diag}$.

We remark that $G\equiv_{\ML{\cdot,^*,\ones,\diag}} H$ if and only if $G\equiv_{\ML{\cdot,^*,\ones,\diag,+,\times,\Apply_{\mathsf{s}}[f],f\in\Omega}} H$. This is again a direct consequence of the fact that $G\equiv_{\ML{\cdot,^*,\ones,\diag}} H$ implies that $A_G\mdot S=S\mdot A_H$ and $A_H\mdot S^*=S^*\mdot A_G$, for a doubly stochastic matrix $S$, and that all operations in $\ML{\cdot,^*,\ones,\diag,+,\times,\allowbreak \Apply_{\mathsf{s}}[f],f\in\Omega}$ preserve $S$-conjugation and $S^*$-conjugation.

\subsection{Characterisation of
\texorpdfstring{$\ML{\cdot,\tr,\change{{}^*,}\ones,\diag}$}{ML(.,tr,*,1,diag)}-equivalence}\label{sub-sec:trdiag} We next consider $\ML{\cdot,\tr,\change{{}^*,}\ones,\diag}$-equivalence. We already know a couple of implications when $G\equiv_{\ML{\cdot,\tr,\change{{}^*,}\ones,\diag}} H$ holds. For example, there must exist an orthogonal matrix $O$ such that $O\cdot\ones=\ones$ and $A_G\mdot O=O\mdot A_H$ (Propositions~\ref{prop:cospeccomain} and~\ref{prop:traceones}). Furthermore, we know that $G$ and $H$ must have a common equitable partition and hence, there exists a doubly stochastic matrix $S$ such that $A_G\mdot S=S\mdot A_H$ (Proposition~\ref{thm:C2}). It is tempting to conjecture that $G\equiv_{\ML{\cdot,\tr,\change{{}^*,}\ones,\diag}} H$ if and only if there exists an orthogonal doubly stochastic matrix $O$ such that $A_G\mdot O=O\mdot A_H$. This does not hold, however. Indeed, invertible doubly stochastic matrices are necessarily permutation matrices (see e.g., Theorem 2.1 in~\cite{farahat_1966}). Then, $A_G\mdot O=O\mdot A_H$ would imply that $G$ and $H$ are isomorphic, contradicting that our fragments cannot go beyond $\CLK{3}$-equivalence~\cite{Brijder2018}. Instead, we have the following characterisation. 
\begin{theorem}\label{thm:treqpart} 
	Let $G$ and $H$ be two graphs of the same order. Then the following hold: $G\equiv_{\ML{\cdot,\,\tr,\change{{}^*,}\ones,\diag}} H$ if and only if $G$ and $H$ have a common equitable partition and $A_G\mdot O=O\mdot A_H$ for some doubly quasi-stochastic orthogonal matrix $O$ which is compatible with some  \change{equitable partitions ${\cal V}$ of $G$ and ${\cal W}$ of $H$, that witness that $G$ and $H$ have a common equitable partition.}
\end{theorem}

\change{In the proof of Theorem~\ref{thm:treqpart} we will rely on Specht's Theorem (see e.g.,~\cite{Jing2015}), which we recall first. Let ${\cal A}=A_1,\ldots,A_p$ and ${\cal B}=B_1,\ldots,B_p$ be two sequences of complex matrices that are closed under complex conjugate transposition. That is, if $A_i$ occurs in ${\cal A}$ then so does $A_i^*$; \change{the same must hold for  ${\cal B}$.}
 The sequences ${\cal A}$ and ${\cal B}$ are called \textit{simultaneously unitary equivalent} if there exists a unitary matrix $U$ such that $A_i\mdot U=U\mdot B_i$, for $i=1,\ldots,p$. Specht's Theorem provides a means of checking simultaneous unitary equivalence in terms of \textit{trace identities}. Indeed, Specht's Theorem states that ${\cal A}$ and ${\cal B}$ are simultaneously unitary equivalent if and only if 
	\begin{linenomath}
		\[ \tr(w(A_1,\ldots,A_p))=\tr(w(B_1,\ldots,B_p)), \]
	\end{linenomath}
	for all words $w(x_1,\ldots,x_p)$ over the alphabet $\{x_1,\ldots,x_p\}$. In expression $w(A_1,\ldots,A_p)$ we instantiated $x_i$ with $A_i$ and interpret concatenation in the word $w$ as matrix multiplication; we interpret $w(B_1,\ldots,B_p)$ in a similar way. Specht's Theorem also holds when ${\cal A}$ and ${\cal B}$ consist of real matrices and unitary matrices are replaced by orthogonal matrices~\cite{Jing2015}.  The required condition is that ${\cal A}$ and ${\cal B}$ are closed under transposition. We next show Theorem~\ref{thm:treqpart}.}

\begin{proof}
	To show that the existence of a matrix $O$, as stated in the Theorem, implies the equivalence  $G\equiv_{\ML{\cdot,\,\tr,\change{{}^*,}\ones,\diag}} H$, we argue as before. More precisely, we show that $O$-conjugation is preserved by the operations in $\ML{\cdot,\,\tr,\change{{}^*,}\ones,\diag}$. This is, however, a direct consequence of Lemmas~\ref{lem:multp-sim},~\ref{lem:trace-sim},~\ref{lem:ones-sim} and~\ref{lem:diag-sim}. We remark that Proposition~\ref{prop:determined1} guarantees that Lemma~\ref{lem:diag-sim} can be applied. Indeed, Proposition~\ref{prop:determined1} implies that $\ML{\cdot,\,\tr,\change{{}^*,}\ones,\diag}$-vectors are constant on equitable partitions. We may thus conclude that all expressions in $\ML{\cdot,\,\tr,\change{{}^*,}\ones,\diag}$ preserve $O$-conjugation. Hence, $e(A_G)=e(A_H)$ for any sentence $e(X)$ in $\ML{\cdot,\,\tr,\change{{}^*,}\ones,\diag}$.
	
	For the converse direction, we need to show that $G\equiv_{\ML{\cdot,\,\tr,\change{{}^*,}\ones,\diag}} H$ implies that there exists an orthogonal matrix $O$ such that $A_G\mdot O=O\mdot A_H$, and where $O$ satisfies the conditions mentioned in the statement of the Theorem.
	
	The existence of the orthogonal matrix $O$ is shown using Specht's Theorem, which we just recalled.  \change{This is done as follows: We first rephrase the conditions required for $O$, i.e., that it is a doubly quasi-stochastic matrix which is compatible with \change{some equitable partitions ${\cal V}$ and ${\cal W}$ of $G$ and $H$,} respectively, in terms of such trace identities. Then we show that these trace identities can be expressed by sentences in  $\ML{\cdot,\,\tr,\change{{}^*,}\ones,\diag}$.}
	
	We note that our approach is inspired by Th\"une~\cite{Thune2012}. In that work, simultaneous unitary equivalence of graphs is studied with respect to their so-called $0$, $1$ and $2$-dimensional Weisfeiler-Lehman closures. Here, we consider simultaneous orthogonal equivalence with respect to all possible matrices that can be obtained using $\ML{\cdot,\,\tr,\change{{}^*,}\ones,\diag}$-expressions.
	
	We start by defining the sequences ${\cal A}$ and ${\cal B}$ on which we will apply Specht's Theorem. Consider the  sequences of real symmetric matrices: ${\cal A}:=A_G, J,\allowbreak \diag(\ones_{V_1}),\ldots, \allowbreak \diag(\ones_{V_\ell})$ and ${\cal B}:= A_H, J,\diag(\ones_{W_1}),\ldots, \diag(\ones_{W_\ell})$, where $\ones_{V_i}$ and $\ones_{W_i}$, for $i=1,\ldots,\ell$, denote the indicator vectors corresponding to the coarsest equitable partitions ${\cal V}=\{V_1,\ldots, V_\ell\}$ and ${\cal W}=\{W_1,\ldots,W_\ell\}$ of  $G$ and $H$, respectively. We observe that ${\cal A}$ and ${\cal B}$ are closed under transposition. By the real counterpart of Specht's Theorem we can check whether there exists an orthogonal matrix $O$ such that 
	\begin{linenomath}
		\postdisplaypenalty=0
		\begin{align}
			A_G\mdot O&=O\mdot A_H\label{eq:cond1}\\
			J\mdot O&= O\mdot J\label{eq:cond2}\\
			\diag(\ones_{V_i})\mdot O&=O\mdot \diag(\ones_{W_i}),
		\label{eq:cond3} \end{align}
	\end{linenomath}
	hold, for $i=1,\ldots, \ell$, in terms of trace identities. It is clear that conditions~(\ref{eq:cond1}) and~(\ref{eq:cond3}) express that $A_G$ and $A_H$ must be $O$-conjugate and that $O$ must be compatible with ${\cal V}$ and ${\cal W}$. The orthogonality of $O$ is implied by Specht's Theorem. Condition~(\ref{eq:cond2}) ensures that we can choose $O$ such that $O\mdot \ones=\ones$ holds. To see this, we next modify the proof of Lemma 4 in Th\"une~\cite{Thune2012}, stated for unitary matrices, so that it holds for orthogonal matrices. 
	
We first observe that $\ones$ is an eigenvector of $O$. Indeed, $J\mdot O\cdot\ones=\ones\mdot (\tp{\ones}\mdot O\mdot \ones)=\alpha \times \ones$ with $\alpha=\tp{\ones}\mdot O\cdot\ones$ and $J\mdot O\cdot\ones=O\mdot J\cdot\ones=(\tp{\ones}\cdot\ones)\times O\cdot\ones$. In other words, $O\mdot \ones=\frac{\alpha}{n}\times\ones$ since $\tp{\ones}\cdot\ones=n$. Furthermore, because $\tp{\ones}\mdot \tp{O}\cdot\ones$ is a scalar, $\tp{\ones}\mdot \tp{O}\cdot\ones=(\tp{\ones}\cdot\tp{O}\cdot\ones)^{\mathsf{t}}=\tp{\ones}\mdot O\cdot\ones=\alpha$. We next show that $\alpha=\pm n$. Indeed, since $O$ is an orthogonal matrix 
	\begin{linenomath}
		\[ n=\tp{\ones}\mdot I\cdot\ones=\tp{\ones}\mdot \tp{O}\mdot O\mdot \ones=\frac{\alpha}{n}\times (\tp{\ones}\mdot \tp{O}\mdot \ones)=\frac{\alpha^2}{n}, \]
	\end{linenomath}
	and thus $\alpha^2=n^2$ or $\alpha=\pm n$. Hence, $O\cdot\ones=\pm \ones$. When $O\cdot\ones=\ones$, $O$ is already doubly quasi-stochastic. In case \change{when} $O\cdot\ones=-\ones$, we simply replace $O$ by $(-1)\times O$ to obtain that $O\cdot\ones=\ones$. This rescaling does not impact the validity of conditions~(\ref{eq:cond1}) and (\ref{eq:cond3}). Hence, $O$ can indeed be assumed to be doubly quasi-stochastic.
	
	It remains to show that the trace identities implying the existence of an orthogonal matrix $O$ satisfying conditions~(\ref{eq:cond1}),~(\ref{eq:cond2}) and (\ref{eq:cond3}) can be expressed in $\ML{\cdot,\tr,\change{{}^*,}\ones,\diag}$. For every word $w(x,j,b_1,\allowbreak \ldots,b_\ell)$ we consider the sentence 
	\begin{linenomath}
		\[ e_w(X):=\tr(w(X,\ones(X)\mdot (\ones(X))^*,\diag(\mathsf{eqpart}_1(X)),\ldots ,\diag(\mathsf{eqpart}_\ell(X)))),\]
	\end{linenomath}
	in which variables $x, j, b_1,\ldots,b_\ell$ are assigned to matrix variable $X$, expression $\ones(X)\mdot (\ones(X))^*$ (which evaluates to $J$), and expressions $\diag(\mathsf{eqpart}_i(X))$, for $i=1,\ldots,\ell$, respectively. Here, the expressions $\mathsf{eqpart}_i(X)$ correspond to the expressions extracting the indicator vectors of the (coarsest) equitable partition of a graph, as defined in the proof of Proposition~\ref{prop:equipart}. We recall from that proof that $\mathsf{eqpart}_i(X)$ are expressible in $\ML{\cdot,\tr,{}^*,\ones,\diag,+,\times}$.
As a consequence, the sentences $e_w(X)$ belong to $\ML{\cdot,\tr,\change{{}^*,}\ones,\diag,+,\times}$. We have seen, however, in the proof of Proposition~\ref{prop:equipart}, that
$G\equiv_{\ML{\cdot,\tr,\change{{}^*},\ones,\diag}} H$ implies
$G\equiv_{\ML{\cdot,\tr,\change{{}^*,}\ones,\diag,+,\times}} H$.
Hence, $G\equiv_{\ML{\cdot,\,\tr,\change{{}^*,}\ones,\diag}} H$ implies that $e_w(A_G)=e_w(A_H)$ for every word $w$. 
\end{proof}

We note that $G\equiv_{\ML{\cdot,\tr,\change{{}^*,}\ones,\diag}} H$ implies $G\equiv_{\ML{\cdot,{}^*,\ones,\diag}} H$. The converse does not hold, as is illustrated next.
\begin{example}\label{ex:stochequiv}
	\normalfont Consider $G_3$ (\!\raisebox{-1.3ex}{\mbox{ 
	\includegraphics[height=0.6cm]{graphG3}}}\hspace{.08em}) and $H_3$ (\!\raisebox{-1.3ex}{\mbox{ 
	\includegraphics[height=0.6cm]{graphH3}}}\hspace{.08em}). These graphs are fractional isomorphic (and thus $\ML{\cdot,{}^*,\ones,\diag}$-equivalent) but are not co-spectral. Hence, $G_3\not\equiv_{\ML{\cdot,\,\tr,\change{{}^*,}\ones,\diag}} H_3$ since $\ML{\cdot,\tr,\change{{}^*,}\ones,\diag}$-equivalence implies co-spectrality. On the other hand, $G_5$ (\!\raisebox{-1.3ex}{\mbox{ 
	\includegraphics[height=0.6cm]{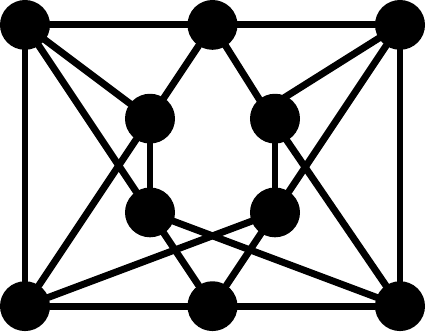}}}) and $H_5$ (\!\raisebox{-1.3ex}{\mbox{ 
	\includegraphics[height=0.6cm]{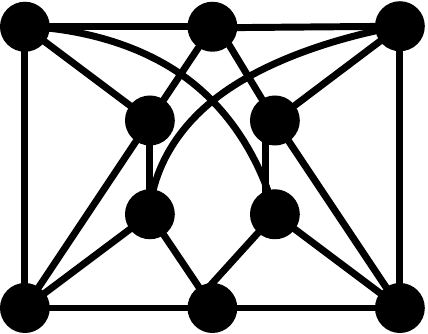}}}) are co-spectral regular graphs (Figure 2 in van Dam et al.~\cite{Vandam2003}), with co-spectral complements. \change{The equitable partitions ${\cal V}$ of $G_5$ and ${\cal W}$ of $H_5$, consisting of a single part containing all vertices, witness that $G_5$ and $H_5$ have a common equitable partition. } We thus know from before that there exists an orthogonal matrix $O$ such that $A_{G_5}\mdot O=O\mdot A_{H_5}$ and $O\mdot \ones=\ones$ (this follows from $G_5$ and $H_5$ being co-spectral and co-main). Moreover, the compatibility requirement \change{with ${\cal V}$ and ${\cal W}$} is vacuously satisfied. Indeed, these partitions are represented by the indicator vector $\one$ and we note that $\diag(\ones)\mdot O=O\mdot \diag(\ones)$ always holds.  Hence, $G_5$ and $H_5$ cannot be distinguished by $\ML{\cdot,\tr,\change{{}^*,}\ones,\diag}$ by Theorem~\ref{thm:treqpart}. ~\hfill$\qed$ 
\end{example}

It would be tempting to conjecture that $G\equiv_{\ML{\cdot,\,\tr,\change{{}^*,} \ones,\diag}} H$ if and only if $G$ and $H$ are co-spectral and have a common equitable partition. It is clearly a necessary condition. We see in the next Section, however, that there exists graphs that are co-spectral and have a common equitable partition, yet are not $\ML{\cdot,\,\tr,\change{{}^*,}\ones,\diag}$-equivalent.

We remark that $G\equiv_{\ML{\cdot,\,\tr,\change{{}^*,}\ones,\diag}} H$ if and only if $G\equiv_{\ML{\cdot.\tr,\change{{}^*,}\ones,\diag,+,\times,\Apply_{\mathsf{s}}[f],f\in\Omega}} H$. This is again a direct consequence of the fact that $G\equiv_{\ML{\cdot,\tr,\change{{}^*,}\ones,\diag}} H$ implies that $A_G\mdot O=O\mdot A_H$, $A_H\mdot O^*=O^*\mdot A_G$, for an orthogonal doubly quasi-stochastic matrix $O$ which is compatible \change{with equitable partitions of $G$ and $H$, that witness that $G$ and $H$ have a common equitable partition.} Furthermore, for such matrices $O$, we observe that all operations in $\ML{\cdot,\tr,\change{{}^*,}\ones,\diag,+,\times,\allowbreak \Apply_{\mathsf{s}}[f],f\in\Omega}$ preserve $O$-conjugation and $O^*$-conjugation.

\subsection{Pointwise function applications on vectors}\label{subsec:diag_funct_w} A crucial ingredient for obtaining characterisations of equivalence in the presence of the $\diag(\cdot)$ operation is that vectors are constant on equitable partitions (Proposition~\ref{prop:determined1} and Lemma~\ref{lem:diag-sim}). In this way, vectors obtained by evaluating expressions on equivalent $A_G$ and $A_H$ are ``almost'' the same, up to the use of indicator vectors. More precisely,
when $e(A_G)=\sum_{i=1}^\ell a_i\times \ones_{V_i}$ then $e(A_H)=\sum_{i=1}^\ell a_i\times \ones_{W_i}$, and vice versa.

 We next show that this tight relationship between vectors computed from $A_G$ and $A_H$ allows us to extend the matrix query languages considered in this section with pointwise function applications on \textit{vectors}. More precisely, we denote by $\Apply_{\mathsf{v}}[f]$, for $f\in\Omega$, that we only allow function applications of the form $e(X):=\Apply_{\mathsf{v}}[f](e_1(X),\ldots,\allowbreak e_p(X))$ where each $e_i(X)$ returns a vector when evaluated on a matrix. 
\begin{proposition}\label{prop:pwvectors} 
	Let $G$ and $H$ be two graphs of the same order. 
	\begin{itemize}
		\item[(1)] $G\equiv_{\ML{\cdot,{}^*,\ones,\diag}} H$ if and only if $G\equiv_{\ML{\cdot,{}^*,\ones,\diag,+,\times,\Apply_{\mathsf{s}}[f],\Apply_{\mathsf{v}}[f],f\in\Omega}} H$. 
		\item[(2)]$G\equiv_{\ML{\cdot,\,\tr,\change{{}^*,}\ones,\diag}} H$ if and only if $G\equiv_{\ML{\cdot,\tr,{}^*,\ones,\diag,+,\times,\,\Apply_{\mathsf{s}}[f],\Apply_{\mathsf{v}}[f],f\in\Omega}} H$. 
	\end{itemize}
\end{proposition}
\begin{proof}
	In view of the previous results, it suffices to show that (1)~$\ML{\cdot,{}^*,\ones,\diag,+,\times,\allowbreak \Apply_{\mathsf{s}}[f],\allowbreak f\in\Omega}$-equivalence implies $\ML{\cdot,{}^*,\ones,\diag,+,\times,\Apply_{\mathsf{s}}[f],\allowbreak \Apply_{\mathsf{v}}[f],f\in\Omega}$-equivalence; and that (2)~$\ML{\cdot,\tr,{}^*,\ones,\diag,+,\times,\allowbreak\,\Apply_{\mathsf{s}}[f],f\in\Omega}$-equivalence implies $\ML{\cdot,\tr,{}^*,\ones,\diag,+,\times,\allowbreak\,\Apply_{\mathsf{s}}[f],\Apply_{\mathsf{v}}[f],f\in\Omega}$-equivalence. Both implication follow if we can show that $\ML{\cdot,\tr,{}^*,\allowbreak \ones,\diag,+,\times,\Apply_{\mathsf{s}}[f],\Apply_{\mathsf{v}}[f],f\in\Omega}$-vectors are constant on equitable partitions and that $\Apply_{\mathsf{v}}[f]$, for $f\in\Omega$, preserves conjugation by quasi doubly-stochastic matrices that are compatible with equitable partitions of \change{$G$ and $H$ that witness that $G$ and $H$ have a common equitable partition.}
	
	For conciseness, let ${\cal L}^{\dagger}$ denote  $\{\cdot,\tr,{}^*,\ones,\diag,+,\times,\Apply_{\mathsf{s}}[f]\Apply_{\mathsf{v}}[f],f\in\Omega\}$, i.e., ${\cal L}^{\dagger}$ consists of all operations considered so far. Proposition~\ref{prop:determined1} (being constant on partitions) trivially generalizes to $\ML{{\cal L}^{\dagger}}$-vectors. 
	Indeed, it suffices to consider the case $$e(X):=\Apply_{\mathsf{v}}[f](e_1(X),\ldots,e_p(X)),$$ where $e_1(X),\ldots,e_p(X)$ are expressions in $\ML{{\cal L}^{\dagger}}$ such that each $e_i(A_G)$ returns a vector. We may assume by induction that for $i=1,\ldots,p$, $e_i(A_G)=\sum_{j=1}^\ell a_{j}^{(i)} \times \ones_{V_i}$ 
for  scalars $a_{j}^{(i)}\in\C$, for $j=1,\ldots,\ell$, \change{and where $\ones_{V_i}$, for $i=1,\ldots, \ell$, are indicator vectors representing an equitable partition of $G$.}
		Since the sets of entries in the indicator vectors holding value $1$ are disjoint for any two different indicator vectors, we have that 
	\begin{linenomath}
\[
			e(A_G)=\sum_{i=1}^\ell \Apply_{\mathsf{s}}[f]\bigl(a_{i}^{(1)},\ldots,a_i^{(p)}) \times \ones_{V_i}.\]
	\end{linenomath}
So, indeed, $\ML{{\cal L}^{\dagger}}$-vectors are constant on equitable partitions.	

\change{Let ${\cal V}=\{V_1,\ldots,V_\ell\}$ and ${\cal W}=\{W_1,\ldots,W_\ell\}$ be equitable partitions of $G$ and $H$, respectively. Suppose that these partitions witness that $G$ and $H$ have a common equitable partition.} That $T$-conjugation, for a quasi doubly-stochastic matrices $T$ that is compatible \change{with ${\cal V}$ and ${\cal W}$},  is also preserved by pointwise function applications on vectors now  easily follows. Indeed, consider $e(X):=\Apply_{\mathsf{v}}[f](e_1(X),\ldots,e_p(X))$. By assumption, $e_i(A_G)=T\mdot e_i(A_H)$ for all $i=1,\ldots,p$. Furthermore,  
$e_i(A_G)=\sum_{j=1}^\ell a_{j}^{(i)} \times \ones_{V_i}$ and $e_i(A_H)=\sum_{j=1}^\ell b_{j}^{(i)} \times \ones_{W_i}$. \change{The indicator vectors $\ones_{V_i}$ and $\ones_{W_i}$, for $i=1,\ldots,\ell$, represent the partitions ${\cal V}$ and ${\cal W}$, as before.}
We have seen in the proof of Lemma~\ref{lem:diag-sim} that  $T$-conjugation of these vectors implies $a_{j}^{(i)}=b_{j}^{(i)}$ for $j=1,\ldots,\ell$ and $i=1,\ldots,p$.
As a consequence, $e(A_G)$ is equal to
	\begin{linenomath}\postdisplaypenalty=0 
	\begin{align*}
		\Apply_{\mathsf{v}}[f](e_1(A_G),\ldots,e_p(A_G))&=\sum_{i=1}^\ell \Apply_{\mathsf{s}}[f]\bigl(a_{i}^{(1)},\ldots,a_i^{(p)}) \times \ones_{V_i} \\
		&=\sum_{i=1}^\ell \Apply_{\mathsf{s}}[f]\bigl(a_{i}^{(1)},\ldots,a_i^{(p)}) \times (T\mdot \ones_{W_i})\\ &= T\mdot \Apply_{\mathsf{v}}[f](e_1(A_H),\ldots,e_p(A_H)), \end{align*}
	\end{linenomath}
	which is equal to $T\mdot e(A_H)$, as desired. 
\end{proof}

Going back to the graphs $G_5$ (\!\raisebox{-1.3ex}{\mbox{ 
\includegraphics[height=0.6cm]{graphG5}}}) and $H_5$ (\!\raisebox{-1.3ex}{\mbox{ 
\includegraphics[height=0.6cm]{graphH5}}}) in Example~\ref{ex:stochequiv}, these cannot even be distinguished by sentences in the large fragments in Proposition~\ref{prop:pwvectors}. In  Section~\ref{sec:C3}, we show that by allowing pointwise function applications on matrices, and in particular the Schur-Hadamard product, we can distinguish these two graphs.

\section{The impact of pointwise multiplication on vectors}\label{sec:pw}
In the preceding section the main use of the $\diag(\cdot)$ operation related to the construction of the coarsest equitable partition (see e.g.,  the proof of Proposition~\ref{prop:equipart}) and more specifically, to the ability to pointwise multiply two vectors (see e.g., Example~\ref{ex:degree3}). In this section we investigate how $\ML{\cdot,{}^*,\ones,\diag}$-equivalence and $\ML{\cdot,\tr,\change{{}^*},\ones,\diag}$-equivalence change if we replace the $\diag(\cdot)$ operation with the operation, denoted by $\odot_v$, which pointwise multiplies vectors. 

We first remark that $e_1(X)\odot_v e_2(X)$ can be expressed as $\diag(e_1(X))\cdot\diag(e_2(X))\cdot\allowbreak\ones(X)$. So surely,
$G\equiv_{\ML{\cdot,{}^*,\ones,\diag}} H$ implies $G\equiv_{\ML{\cdot,{}^*,\ones,\odot_v}} H$. We show below that the converse also holds.
For this fragment, it thus does not matter whether we include $\diag(\cdot)$ or $\odot_v$.
Similarly, $G\equiv_{\ML{\cdot,\tr,\change{{}^*,}\ones,\diag}} H$ implies $G\equiv_{\ML{\cdot,\tr,\change{{}^*,}\ones,\odot_v}} H$. 
 The converse, however, does not hold as we will show below\footnote{It was incorrectly stated in the conference version~\cite{Geerts19} that $\diag(\cdot)$ and $\odot_v$ are interchangeable.}. Finally, $G\equiv_{\ML{\cdot,\tr,{}^*,\ones,\odot_v}} H$ trivially implies $G\equiv_{\ML{\cdot,{}^*,\ones,\odot_v}} H$. The results in this section show that the converse does not hold.

\begin{example}\label{ex:pw}\normalfont 
Consider the graphs 
$G_6$ (\!\!\raisebox{-1.1ex}{\mbox{ 
\includegraphics[height=0.5cm]{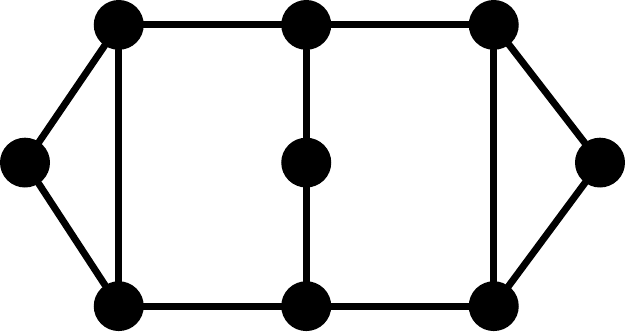}}}) and $H_6$ (\!\raisebox{-1.1ex}{\mbox{ 
\includegraphics[height=0.5cm]{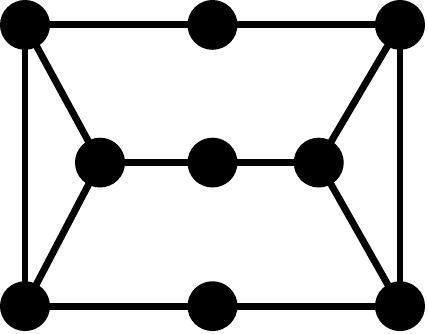}}})\footnote{I am indebted to David E. Roberson for providing these two graphs~\cite{Roberson}.}. On can verify that these graphs are co-spectral and have a common equitable partition (and thus also have co-spectral complements).
Using the diagonal operation we can construct the Laplacian of a graph by simply considering expression $L(X):=(\diag(X\cdot\ones(X))-X$. It is now easy to detect that $G_6$
and $H_6$ have Laplacians that are not co-spectral. Indeed, consider the $\ML{\cdot,\tr,\change{{}^*,}\ones,\diag,+,\times}$ expression $e_{L,k}(X):=\tr(L(X)^k)$. Then, we can check that $e_{L,3}(A_{G_6})=1602\neq 1618=e_{L,3}(A_{H_6})$. The relation between co-spectrality and traces of powers of matrices (cfr.  Proposition~\ref{prop:tracesim}) holds more generally for symmetric matrices (this follows easily from the real version of Specht's Theorem used in the proof of Theorem~\ref{thm:treqpart}).
Hence, we can infer that  the Laplacians of $G_6$ and $H_6$ are not co-spectral.
Another way of verifying this is that $G_6$ and $H_6$ have a different number of spanning trees ($192$ versus $160$) and Kirchhoff's matrix tree theorem (see e.g., Proposition 1.3.4 in~\cite{Brouwer2012}) implies that graphs with co-spectral Laplacians must have the same number of spanning trees. Hence, $G_6$ and $H_6$ can be distinguished by $\ML{\cdot,\tr,\change{{}^*,}\ones,\diag,+,\times}$ (and hence also by sentences in  $\ML{\cdot,\tr,\change{{}^*,}\ones,\diag}$ as we have seen before). Nevertheless, we will see that $G_6$ and $H_6$ cannot be distinguished by sentences in $\ML{\cdot,\tr,{}^*,\ones,\odot_v}$. More generally, we show that two graphs are $\ML{\cdot,\tr,{}^*,\ones,\odot_v}$-equivalent if and only if they are co-spectral and have a common equitable partition (see Proposition~\ref{prop:trstoch} below).~\hfill$\qed$
\end{example}

As mentioned earlier, when considering $\ML{\cdot,{}^*,\ones,\diag}$ one can equivalently use $\odot_v$ instead of $\diag(\cdot)$.
\begin{proposition}\label{prop:pwmultnotr}
Let $G$ and $H$ be two graphs of the same order. Then, $G\equiv_{\ML{\cdot,{}^*,\ones,\odot_v}} H$ if and only if 
$G\equiv_{\ML{\cdot,{}^*,\ones,\diag}} H$. 
\end{proposition}
\begin{proof}
We already observed at the beginning of this section that $G\equiv_{\ML{\cdot,{}^*,\ones,\diag}} H$ implies $G\equiv_{\ML{\cdot,{}^*,\ones,\odot_v}} H$.
\change{It remains to show the reverse implication. This can be proved by
a simple inductive argument, transforming sentences in $\ML{\cdot,{}^*,\ones,\diag,\odot_v}$ into sentences without $\diag(\cdot)$.
With each sentence $e(X)$ in $\ML{\cdot,{}^*,\ones,\diag,\odot_v}$
we associate, in a natural way, the number  of occurrences, denoted by $k$, of the $\diag(\cdot)$ operation in this sentence.}
\change{We will show, by induction on $k$,
that any sentence $e(X)$ in $\ML{\cdot,{}^*,\ones,\diag,\odot_v}$ is equivalent to a sentence $\tilde{e}(X)$ in $\ML{\cdot,{}^*,\ones,\odot_v}$. Here, equivalence of sentences means that $e(A)=\tilde{e}(A)$ for any input matrix $A$. Consider a sentence $e(X)$ in $\ML{\cdot,{}^*,\ones,\diag,\odot_v}$.
Clearly, if $k=0$, then $e(X)$
is already a sentence in $\ML{\cdot,{}^*,\ones,\odot_v}$. Assume now that any sentence
$e(X)$ in $\ML{\cdot,{}^*,\ones,\diag,\odot_v}$ with $k$ occurrences is equivalent to a sentence $\tilde{e}(X)$ in $\ML{\cdot,{}^*,\ones,\odot_v}$. 
Consider a sentence $e(X)$ in $\ML{\cdot,{}^*,\ones,\diag,\odot_v}$ having $k+1$ occurrences of $\diag(\cdot)$. 
Since $e(X)$ is a sentence, we can write
$e(X)$ as $e_1(X)\mdot\diag(e_2(X))\mdot e_3(X)$, where $e_1(X)$, $e_2(X)$ and $e_3(X)$ are expressions in $\ML{\cdot,{}^*,\ones,\diag,\odot_v}$. We will next eliminate the occurrence of $\diag(\cdot)$ in  $\diag(e_2(X))$. We distinguish between the following two cases, depending on the dimension of $e_2(A)$, were $A$ is an  $n\times n$-matrix. Suppose first that $e_2(A)$ is a scalar. Clearly, in this case $\diag(e_2(X))$ is equivalent to $e_2(X)$ and hence $e(X)$ is equivalent to the expression $e'(X):=e_1(X)\mdot e_2(X)\mdot e_3(X)$. Suppose next that $e_2(A)$ evaluates to an $n\times 1$ column vector,  with $n>1$. Then, necessarily, $e_1(A)$ evaluates to a $1\times n$ row vector and
$e_3(A)$ evaluates to an $n\times 1$ column vector. This implies that $e(X)$ is equivalent to $e'(X):=e_1(X)\mdot (e_2(X)\odot_v e_3(X))$. We remark that in both cases, $e'(X)$ is an expression in $\ML{\cdot,{}^*,\ones,\diag,\odot_v}$
consisting of  $k$ occurrences of $\diag(\cdot)$. By the induction hypothesis, $e'(X)$ (and hence also $e(X)$) is indeed equivalent to a sentence $\tilde{e}(X)$ in 
$\ML{\cdot,{}^*,\ones,\odot_v}$, as desired. Hence,  $G\equiv_{\ML{\cdot,{}^*,\ones,\odot_v}} H$ implies $G\equiv_{\ML{\cdot,{}^*,\ones,\diag}} H$.
}
\end{proof}

The above proof fails when the trace operation is present. The reason is that we can have sentences like $e_{L,k}(X)$ in Example~\ref{ex:pw}, \change{which cannot be decomposed as $e_1(X)\mdot \diag(e_2(X))\mdot e_3(X)$. The previous proof crucially relies on the fact that all sentences can be written in that form.}

We next consider $\ML{\cdot,\tr,{}^*,\ones,\odot_v}$. To analyse the distinguishability of graphs by sentences in this fragment, we follow the same approach as for $\ML{\cdot,\tr,\change{{}^*,}\ones,\diag}$.

\begin{corollary}\label{corr:equipw}
	Let $G$ and $H$ be two graphs of the same order. Then, $G\equiv_{\ML{\cdot,\tr,{}^*,\ones,\odot_v}} H$ implies that $G$ and $H$ have a common equitable partition. 
\end{corollary}
\begin{proof}
This is a direct consequence of Propositions~\ref{prop:equipart} and~\ref{prop:pwmultnotr}. Indeed, we know already that
 $G\equiv_{\ML{\cdot,{}^*,\ones,\diag}} H$ implies that 
$G$ and $H$ have a common equitable partition and we have just shown that $G\equiv_{\ML{\cdot,{}^*,\ones,\diag}} H$ if and only if 
$G\equiv_{\ML{\cdot,{}^*,\ones,\odot_v}} H$. It now suffices to observe that $\ML{\cdot,{}^*,\ones,\odot_v}$ is a smaller fragment than $\ML{\cdot,\tr,{}^*,\ones,\odot_v}$.\end{proof}

Furthermore, we can add pointwise vector multiplication to the list of operations in Proposition~\ref{prop:determined1}.  We defer the proof to the appendix.
\begin{proposition}\label{prop:determined2} 
	$\ML{\cdot,\tr,{}^*, \ones,\diag,\odot_v,\diag,+,\times,\Apply_{\mathsf{s}}[f],f\in\Omega}$-vectors are constant on equitable partitions. \qed
\end{proposition}

It remains to identify an appropriate \change{class of matrices for which conjugation is preserved by pointwise vector multiplication.} Let $G$ and $H$ be two graphs that have a common equitable partition.
As before, let ${\cal V}=\{V_1,\ldots, V_\ell\}$ and ${\cal W}=\{W_1,\ldots,W_\ell\}$ be equitable  partitions of $G$ and $H$, respectively. The corresponding
indicator vectors are denoted by $\ones_{V_i}$ and $\ones_{W_i}$, for $i=1,\ldots,\ell$, respectively. We say that a matrix $T$ \textit{preserves the partitions ${\cal V}$ and ${\cal W}$}
if $\ones_{V_i}=T\cdot\ones_{W_i}$
for all $i=1,\ldots,\ell$.
 We note that this condition is weaker than the compatibility notion used before (see the  proof of Lemma~\ref{lem:diag-sim} \change{where} we verified the \textit{preservation} of partitions for matrices that are \textit{compatible} with these partitions).

\begin{lemma}\label{lem:diag-sim2} 
Let $\ML{\cal L}$ be a matrix query language such that $\ML{\cal L}$-vectors are constant on equitable partitions. 	\change{Let $G$ and $H$ be two graphs of the same order and let ${\cal V}$ and ${\cal W}$ be equitable partitions of $G$ and $H$, respectively. Assume that ${\cal V}$ and ${\cal W}$ witness that $G$ and $H$ have a common equitable partition. Furthermore, 
let $T$ be a matrix which preserves ${\cal V}$ and ${\cal W}$.} Finally, Let $e_1(X)$  and $e_2(X)$ be  expressions in $\ML{\cal L}$ which evaluate to vectors. Then, if $e_1(A_G)$ and $e_1(A_H)$ are $T$-conjugate, and $e_2(A_G)$ and $e_2(A_H)$ are $T$-conjugate,
	then also $e_1(A_G)\odot_v e_2(A_G)$ and $e_1(A_H)\odot_v e_2(A_H)$ are $T$-conjugate. 
\end{lemma}
\begin{proof}
The proof is similar to the proof of Lemma~\ref{lem:diag-sim}. Consider $e(X):=e_1(X)\odot_v e_2(X)$. We distinguish between two cases, depending on the dimensions of $e(A_G)$. First, if $e(A_G)$ is a sentence then we know by induction that $e_1(A_G)=e_1(A_H)$ and $e_2(A_G)=e_2(A_H)$. Hence, 
	\begin{linenomath}
		\postdisplaypenalty=0 
		\begin{align*}
 e(A_G)&=e_1(A_G)\odot_v e_2(A_G)=e_1(A_G)\mdot e_2(A_G)\\
 &=e_1(A_H)\mdot e_2(A_H)=e_1(A_H)\odot_v e_2(A_H)=e(A_H). 
 	\end{align*}
	\end{linenomath}
	Next, if $e_1(A_G)$ and $e_2(A_G)$ are (column) vectors, then we know that $e_1(A_G)=T\mdot e_1(A_H)$ and $e_2(A_G)=T\mdot e_2(A_H)$.
	We argued in the proof of Lemma~\ref{lem:diag-sim} that when $\ones_{V_i}=T\cdot\ones_{W_i}$ holds for $i=1,\ldots,\ell$, then since vectors are constant on equitable partitions,
	$e_1(A_G)=\sum_{i=1}^\ell a_i\times \ones_{V_i}=\sum_{i=1}^\ell a_i\times (T\mdot \ones_{W_i})=T\mdot e_1(A_H)$ and 
	$e_2(A_G)=\sum_{i=1}^\ell b_i\times \ones_{V_i}=\sum_{i=1}^\ell b_i\times (T\mdot \ones_{W_i})=T\mdot e_2(A_H)$.	
	We may now conclude that 
	\begin{linenomath}
		\postdisplaypenalty=0 
		\begin{align*}
			e(A_G)&=e_1(A_G)\odot_v e_2(A_G)=\sum_{i=1}^\ell (a_i\times b_i)\times \ones_{V_i}=\sum_{i=1}^\ell (a_i\times b_i)\times (T\mdot \ones_{W_i})\\
			& =T\cdot\bigl(\sum_{i=1}^\ell (a_i\times b_i) \times \ones_{W_i})\bigr)=T\cdot(e_1(A_H)\odot_v e_2(A_H))=T\mdot e(A_H). 
		\end{align*}
	\end{linenomath}
	Hence, $e(A_G)$ and $e(A_H)$ are indeed $T$-conjugate. 
\end{proof}

We can now state a characterisation of $\ML{\cdot,\tr,{}^*,\ones,\odot_v}$-equivalence.
\begin{theorem}\label{thm:pwtrpart}
Let $G$ and $H$ be two graphs of the same order. Let ${\cal V}$ and ${\cal W}$ be equitable partitions of $G$ and $H$, respectively, that witness that $G$ and $H$ have a common equitable partition.
Then,
$G\equiv_{\ML{\cdot,\tr,{}^*,\ones,\odot_v}} H$ if and only if there exists an orthogonal matrix $O$  which preserves ${\cal V}$ and ${\cal W}$ and such that $A_G\mdot O=O\mdot A_H$.
\end{theorem}
\begin{proof}
To show that the existence of a matrix $O$, as stated in the Theorem, implies $G\equiv_{\ML{\cdot,\,\tr,{}^*,\ones,\odot_v}} H$, we argue as before. More precisely, we show that $O$-conjugation and $O^*$-conjugation is preserved by the operations in $\ML{\cdot,\tr,{}^*,\ones,\odot_v}$. This is, however, a direct consequence of Lemmas~\ref{lem:multp-sim},~\ref{lem:trace-sim},~\ref{lem:complextranspose-elim1},~\ref{lem:ones-sim} and~\ref{lem:diag-sim2}. We remark that Proposition~\ref{prop:determined2} guarantees that Lemma~\ref{lem:diag-sim2} can be applied. Indeed, Proposition~\ref{prop:determined2} implies that $\ML{\cdot,\tr,{}^*,\ones,\odot_v}$-vectors are constant on equitable partitions. Furthermore, since $\ones_{V_i}=O\cdot\ones_{W_i}$, for all $i=1,\ldots,\ell$, and $\ones=\sum_{i=1}^\ell \ones_{V_i}=\sum_{i=1}^\ell \ones_{W_i}$, we have that $\ones=O\cdot\ones$. Hence, $O$ is doubly quasi-stochastic and Lemma~\ref{lem:ones-sim} applies. Moreover, if $A_G\mdot O=O\mdot A_H$ then $A_H\mdot O^*=O^*\mdot A_G$ due to the orthogonality of $O$, and $O^*$  \change{preserves ${\cal W}$ and ${\cal V}$}. Hence, Lemma~\ref{lem:complextranspose-elim1} applies.
We may thus conclude that all expressions in $\ML{\cdot,\tr,{}^*,\ones,\odot_v}$ preserve $O$- and $O^*$-conjugation. Hence, $e(A_G)=e(A_H)$ for any sentence $e(X)$ in $\ML{\cdot,\tr,{}^*,\ones,\odot_v}$.
	
	For the converse direction, we need to show that $G\equiv_{\ML{\cdot,{}^*,\,\tr,\ones,\odot_v}} H$ implies that there exists an orthogonal matrix $O$ such that $A_G\mdot O=O\mdot A_H$, and where $O$ \change{preserves the partitions ${\cal V}$ and ${\cal W}$.}
	This can be shown, just like in the proof of Theorem~\ref{thm:treqpart}, by means of trace conditions. In particular, we impose trace conditions such that $O$ satisfies $A_G\mdot O=O\mdot A_H$ and $(\ones_{V_i}\mdot (\ones_{V_i})^*)\mdot O=O\mdot (\ones_{W_i}\cdot(\ones_{W_i})^*)$,
	for $i=1,\ldots,\ell$. These conditions replace conditions~(\ref{eq:cond2}) and~(\ref{eq:cond3}) in the proof of Theorem~\ref{thm:treqpart}. We show in the appendix that this indeed implies that $O$ \change{can be chosen such that it preserves ${\cal V}$ and ${\cal W}$}. The argument is based on a generalisation of Lemma 4 in Th\"une~\cite{Thune2012}. As in the proof of Theorem~\ref{thm:treqpart}, the trace conditions $e_w(X)$ which ensure the existence of an orthogonal matrix $O$ such that $(\ones_{V_i}\mdot (\ones_{V_i})^*)\mdot O=O\mdot (\ones_{W_i}\cdot(\ones_{W_i})^*)$ holds for $i=1,\ldots,\ell$, rely on the expressions $\mathsf{eqpart}_i(X)$ (from the proof of Proposition~\ref{prop:equipart}). These expressions use addition and scalar multiplication. \change{It is now easily verified (see appendix) that $G\equiv_{\ML{\cdot,\tr,{}^*,\ones,\odot_v}} H$
	implies $G\equiv_{\ML{\cdot,\tr,{}^*,\ones,\odot_v,+,\times}} H$} and hence, $G\equiv_{\ML{\cdot,\tr,{}^*,\ones,\odot_v}} H$
 implies that
	$e_w(A_G)=e_w(A_H)$.	\end{proof}

As it turns out, $\ML{\cdot,\tr,{}^*,\ones,\odot_v}$-equivalence precisely captures co-spectral and fractional isomorphic graphs.
\begin{proposition}\label{prop:trstoch} 
	Let $G$ and $H$ be graphs of the same order. Then, $G\equiv_{\ML{\cdot,\tr,{}^*,\ones,\odot_v}} H$ if and only if $G$ and $H$ are co-spectral and have a common equitable partition. 
	\end{proposition}
	We also observe that $G\equiv_{\ML{\cdot,\tr,^*,\ones,\odot_v}} H$ if and only if $\textsf{HOM}_{\cal F}(G)=\textsf{HOM}_{\cal F}(H)$ where
	${\cal F}$ consists of all trees and cycles. This follows from the results in Dell et al.~\cite{Dell2018}.
	
	\begin{proof}
If $G\equiv_{\ML{\cdot,\tr,{}^*,\ones,\odot_v}} H$, then $G$ and $H$ must have a common equitable partition by Corollary~\ref{corr:equipw}. 
Furthermore, we know from \change{Propositions}~\ref{prop:tracesim} and~\ref{theorem:trace}, that $G$ and $H$ must also be co-spectral. 
For the converse, we explicitly construct an orthogonal matrix $O$ such that $A_G\mdot O=O\mdot A_H$ and $O$ preserves \change{equitable partitions of $G$ and $H$ that witness} that $G$ and $H$ have a common equitable partition.
Then, Theorem~\ref{thm:pwtrpart} implies that $G\equiv_{\ML{\cdot,\tr,{}^*,\ones,\odot_v}} H$ holds.

We next construct the matrix $O$.
\change{Let $G$ and $H$ be two graphs of order $n$ and assume that ${\cal V}=\{V_1,\ldots,V_\ell\}$ and ${\cal W}=\{W_1,\ldots,W_\ell\}$ are equitable partitions of $G$ and $H$, respectively, that witness that $G$ and $H$ have a common equitable partition. As before, we denote by $\ones_{V_1},\ldots,\ones_{V_\ell}$ the indicator vectors corresponding to ${\cal V}$ and by
$\ones_{W_1},\ldots,\ones_{W_\ell}$ the indicator vectors corresponding to ${\cal W}$.}

We first observe the following. It is known that, for indicator vectors representing an equitable partition, the subspace $U_G=\textsf{span}(\ones_{V_1},\ldots,\ones_{V_\ell})$ of $\C^n$ is an $A_G$-invariant subspace (see e.g., Lemma 5.2 in~\cite{Chan1997}).
In other words, for any $v\in U_G$, $A_G\mdot v\in U_G$. Furthermore, since $A_G$ is a symmetric matrix, also the orthogonal complement subspace $U_G^\bot$ is $A_G$-invariant
(see e.g., Theorem 36 in~\cite{Kaplansky1974}). Here, $U_G^\bot$ consists of all vectors $v'$ in $\C^n$ that are orthogonal to any vector $v\in U_G$, i.e., such $\tp{v}\mdot v'=0$ holds.
Let us interpret $A_G$ as the linear operator $T_G:\C^n\to \C^n:v\mapsto A_G\mdot v$. This is a diagonalisable operator (because $A_G$ is symmetric) and it is known that
the  restrictions $T_G|_{U_G}$ and $T_G|_{U_G^\bot}$ are also diagonalisable operators (because of the invariance of these two subspaces (see e.g., Corollary 15.9 in~\cite{gorodentsev2016})). This implies, as this is merely a restatement of a linear operator to be diagonalisable,  that
there exists eigenvectors $v_1,\ldots,v_\ell,v_1',\ldots,v_{n-\ell}'$ of $A_G$ such that $U_G=\textsf{span}(v_1,\ldots,v_\ell)$ and $U_G^\bot=\textsf{span}(v_1',\ldots,v_{n-\ell}')$.
Furthermore, if we denote by $P_G$ the matrix with columns $\ones_{V_1},\ldots,\ones_{V_\ell}$, then $A_G\mdot P_G=P_G\mdot C$ with $C$ the $\ell\times\ell$-matrix
such that $C_{ij}=\deg(v,V_j)$ for some (arbitrary) vertex $v\in V_i$ (see e.g., Lemma 6.1 in~\cite{Chan1997}). Also $C_{ij}$ is diagonalisable (this follows from the fact that the characteristic polynomial of $C$ divides that of $A_G$ (see e.g., Theorem 6.2 in~\cite{Chan1997})
and hence there exists $\ell$ linearly independent eigenvectors $c_1,\ldots,c_\ell$ of $C$\footnote{We here use that we work in the algebraically closed field $\C$ for which being diagonalisable coincides with the minimal polynomial being a product of monic linear factors of the form $(x-\lambda)$. So, since
the characteristic polynomial of $C$ divides that of $A_G$, also the minimal polynomial of $C$ divides that of $A_G$. Since $A_G$ is diagonalisable, its minimal polynomial is a product of monic linear factors. Hence, also the minimal polynomial of $C$ has this form and $C$ is diagonalisable as well.}. It is known that $v_i=P_G\mdot c_i$, for $i=1,\ldots,\ell$, are independent
eigenvectors of $A_G$. More precisely, if $C\mdot c_i=\lambda_i \times c_i$ then $A_G\mdot (P_G\mdot c_i)=\lambda_i\times (P_G\mdot c_i)$.
Moreover, $P_G\mdot c_i\in U_G$, for $i=1,\ldots,\ell$. We may thus assume
that $U_G$ is spanned by $P_G\mdot c_1,\ldots, P_G\mdot c_\ell$.

The reasoning above also holds for $A_H$, i.e., there are eigenvectors $w_1,\ldots,w_\ell,w_1',\allowbreak \ldots,w_{n-\ell}'$ of $A_H$ such that $U_H=\textsf{span}(w_1,\ldots,w_\ell)$ and $U_H^\bot=\textsf{span}(w_1',\ldots,w_{n-\ell}')$. 

\change{Important to observe here is that ${\cal V}$ and ${\cal W}$ witness that  $G$ and $H$ have a common equitable partition.} As a consequence, $A_H\mdot P_H=P_H\mdot C$, where $P_H$ is now the matrix with columns $\ones_{W_1},\ldots,\ones_{W_\ell}$ and $C$ is the same $\ell\times\ell$-matrix as used above.
We may thus assume
that $U_H$ is spanned by $P_H\mdot c_1,\ldots, P_H\mdot c_\ell$ and furthermore, $P_G\mdot c_i$ and $P_H\mdot c_i$ are eigenvectors of $A_G$ and $A_H$, respectively, both belonging to
the same eigenvalue $\lambda_i$ of $C$.

We next use that $G$ and $H$ are co-spectral. The argument above, combined with co-spectrality, implies that the (multi-set) of eigenvalues corresponding to the eigenvectors spanning $U_G$ and $U_H$ are
the same. This implies in turn, by co-spectrality, that we may also assume that $A_G\mdot v_i'=\lambda_i\times v_i'$ and $A_H\mdot w_i'=\lambda_i\times w_i'$, for $i=1,\ldots,n-\ell$, for some eigenvalues $\lambda_i$
of $A_G$ (and $A_H$). We recall that $U_G$ and $U_H$ are also spanned by 
$\ones_{V_i},\ldots,\ones_{V_\ell}$ and $\ones_{W_1},\ldots,\ones_{W_\ell}$, respectively. This implies, that the eigenvectors spanning $U_G^\bot$ and $U_H^\bot$ are necessarily orthogonal
to these indicator vectors. 

We define $O$ as the matrix $O_G\mdot \tp{O_H}$, where $O_G$ is the orthogonal matrix whose columns consist of the vectors  $\frac{1}{\sqrt{n_1}}\ones_{V_1},\ldots,\frac{1}{\sqrt{n_\ell}}\ones_{V_\ell},v_1',\ldots,v_{n-\ell}'$
and $O_H$ is the  orthogonal matrix whose columns consist of the vectors  $\frac{1}{\sqrt{n_1}}\ones_{W_1},\ldots,\frac{1}{\sqrt{n_\ell}}\one_{W_\ell},w_1',\ldots,w_{n-\ell}'$, where $n_i=|V_i|=|W_i|$ and \change{where} we assume the
eigenvectors $v_i'$ and $w_i'$ to be normalised. As a consequence, $O$ is clearly an orthogonal matrix and thus $O\mdot \tp{O}=I=\tp{O}\mdot O$ holds. In view of the construction of the eigenvectors,
we have the following more explicit expression for $O$:
\begin{linenomath} 
\[
O=\sum_{j=1}^\ell \Bigl(\frac{1}{n_j}\times (\ones_{V_j}\cdot\tp{\ones}_{W_j})\Bigr) + \sum_{j=1}^{n-\ell} v_j'\cdot\tp{(w_j')}.\]
\end{linenomath} 
We verify the required conditions.
To begin with, we note that $O\cdot\ones_{W_i}=\ones_{V_i}$, for $i=1,\ldots,\ell$. Indeed, this follows from  the fact that 
$\tp{\ones}_{W_j}\mdot \ones_{W_i}$ is zero when $i\neq j$ and is $|W_i|=n_i$ when $i=j$. Moreover, $\tp{(w_j')}\cdot\ones_{W_i}=0$ because of $w_j'\in U_H^\bot$, for all $j=1,\ldots,n-\ell$.
Hence, $O$ indeed preserves \change{${\cal V}$ and ${\cal W}$.}
It remains to verify that
$A_G\mdot O=O\mdot A_H$. We verify this for both terms in the above expression for $O$. Since $v_i'$ and $w_i'$ are eigenvectors of $A_G$ and $A_H$, respectively, belonging to the
same eigenvalue $\lambda_i$, we have for the second term:
\allowdisplaybreaks
\begin{linenomath}
		\postdisplaypenalty=0 
\begin{align*}
A_G\mdot \Bigl(\sum_{j=1}^{n-\ell} v_j'\cdot\tp{(w_j')}\Bigr)&=\sum_{j=1}^{n-\ell} A_G\mdot v_j'\cdot\tp{(w_j')}=\sum_{j=1}^{n-\ell} \lambda_j \times ( v_j'\cdot\tp{(w_j')})\\
&=\sum_{j=1}^{n-\ell} v_j'\cdot\tp{(w_j')}\mdot A_H=\Bigl(\sum_{j=1}^{n-\ell} v_j'\cdot\tp{(w_j')}\Bigr)\mdot A_H.
\end{align*}\end{linenomath} 
For the first term in the expression for $O$, we consider the matrices\allowdisplaybreaks 
\begin{linenomath}
		\postdisplaypenalty=0\begin{align*}
B_G=A_G\mdot \Bigl(\sum_{i=1}^\ell \frac{1}{n_j}\times (\ones_{V_i}\cdot\tp{\ones}_{W_i}) \Bigr)&=\sum_{i=1}^\ell\sum_{j=1}^\ell  \bigl(\frac{1}{n_i} \times \deg(v_i,V_j)\bigr)\times (\ones_{V_j}\cdot\tp{\ones}_{W_i})\\
B_H= \Bigl(\sum_{i=1}^\ell \frac{1}{n_i}\times (\ones_{V_i}\cdot\tp{\ones}_{W_i}) \Bigr)\mdot A_H&=\sum_{i=1}^\ell\sum_{j=1}^\ell  \bigl(\frac{1}{n_i} \times\deg(w_i,W_j)\bigr)\times (\ones_{V_i}\cdot\tp{\ones}_{W_j}),
\end{align*}\end{linenomath} 
for some (arbitrary) vertices $v_i\in V_i$ and $w_i\in W_i$. We here used that the indicator vectors represent equitable partitions.
We now look at the entries in the matrices $B_G$ and $B_H$ and observe that  $J=\sum_{i,j=1}^\ell \ones_{V_j}\cdot\tp{\ones}_{W_i}$. Hence, for each $p,q\in\{1,\ldots,n\}$ we can define
$f(p)$ and $g(q)$ as the unique indexes of indicator vectors $\ones_{V_{f(p)}}$ and $\ones_{W_{g(q)}}$ such that they hold value $1$ at position $p$ and $q$, respectively.
Then, 
\begin{linenomath} 
\[
(B_G)_{p,q}=\frac{1}{n_{f(p)}}\times \deg(v_{f(p)},V_{g(q)})=\frac{1}{n_{f(p)}}\times \deg(w_{f(p)},W_{g(q)})=(B_H)_{p,q},\]
\end{linenomath} 
because the indicator vectors \change{correspond to equitable partitions that witness that $G$ and $H$ have a common equitable partition.} Hence, we may indeed conclude that $A_G\mdot O=O\mdot A_H$.
\end{proof}

\begin{example}\normalfont
We already mentioned that the graphs 
$G_6$ (\!\!\raisebox{-1.1ex}{\mbox{ 
\includegraphics[height=0.5cm]{graphG6}}}) and $H_6$ (\!\raisebox{-1.1ex}{\mbox{ 
\includegraphics[height=0.5cm]{graphH6}}}) are  co-spectral and have a common equitable partition. Proposition~\ref{prop:trstoch} implies that 
$G_6\equiv_{\ML{\cdot,\tr,{}^*,\ones,\odot_v}} H_6$, as anticipated. \hfill~$\qed$
\end{example}

We mention that we can extend $\ML{\cdot,\tr,{}^*,\ones,\odot_v}$ with $+$, $\times$,  and pointwise function applications
on scalars and vectors, without increasing the distinguishing power of the fragments. This can be shown in precisely the same way as before
by showing that $O$- and $O^*$-conjugation is preserved by these operations when $O$ is an orthogonal matrix which preserves 
equitable partitions \change{(that witness that the graphs have a common equitable partition)}.

Finally, we separate distinguishability by $\ML{\cdot,\tr,{}^*,\ones,\odot_v}$ and  $\ML{\cdot,{}^*,\ones,\odot_v}$.
It suffices to consider the graphs 
  $G_3$ (\!\raisebox{-1.3ex}{\mbox{ 
	\includegraphics[height=0.6cm]{graphG3}}}\hspace{.08em}) and $H_3$ (\!\raisebox{-1.3ex}{\mbox{ 
	\includegraphics[height=0.6cm]{graphH3}}}\hspace{.08em}) which are fractionally isomorphic but not co-spectral. 
	Proposition~\ref{prop:pwmultnotr} implies that these are indistinguishable by $\ML{\cdot,{}^*,\ones,\odot_v}$.
	Proposition~\ref{prop:trstoch} implies that they can be distinguished by $\ML{\cdot,\tr,{}^*,\ones,\odot_v}$.

\section{The impact of pointwise functions on matrices}\label{sec:C3} 
The final operation that we consider is pointwise function applications on \emph{matrices}. In particular, we start by considering the Schur-Hadamard product, which we denote by the binary operation $\odot$, i.e., $(A\odot B)_{ij}=A_{ij}B_{ij}$ for matrices $A$ and $B$. \change{Our results will imply that once two
graphs are equivalent with regards to sentences in $\ML{\cdot,\,\tr,{}^*,\ones,\allowbreak \diag,\odot}$, then they will be equivalent with regards to sentences in $\ML{\cdot,\,\tr,{}^*,\ones,\allowbreak \diag,\odot,\Apply[f],f\in\Omega}$
for \textit{any} pointwise function application $\Apply[f]$, be it on scalars, vectors or matrices. That is, the graphs will be $\MLANG$-equivalent.}

\change{This section is structured in a similar way as Section~\ref{sec:diag}. More precisely, we first  illustrate the additional power that the Schur-Hadamard product provides by means of an example in
Section~\ref{subsec:example_hada}. Then, in Section~\ref{subsec:stable_edge_hada} we show that we can compute a so-called 
stable edge partition of a graph by means of expressions in $\ML{\cdot,\,\tr,{}^*,\ones,\allowbreak \diag,\odot,+,\times}$. Such an edge partition can be regarded as a generalisation of the notion of equitable partition. Stable edge partitions can be obtained as the result of the edge colouring or $2$-dimensional Weisfeiler-Lehman ($\textsf{2WL}$) algorithm.
In Section~\ref{subsec:stable_edge_hada}, we also show that  $G\equiv_{\ML{\cdot,\,\tr,{}^*,\ones,\allowbreak \diag,\odot}} H$ implies that $G$ and $H$  are indistinguishable by the 2WL-algorithm. It is known from the seminal paper by Cai, F\"urer and Immerman~\cite{Cai1992}, that this is equivalent to $\CLK{3}$-equivalence.
Based on this, in Section~\ref{subsec:c3} we prove the main result of this section, i.e., that $G\equiv_{\ML{\cdot,\,\tr,{}^*,\ones,\allowbreak \diag,\odot}} H$ if and only if 
 $G\equiv_{\MLANG} H$ if and only if $G\equiv_{\CLK{3}} H$.}

We remark that from the work by Brijder et al.~\cite{Brijder2018} it implicitly follows that $\CLK{3}$-equivalence implies $\MLANG$-equivalence. Our results thus show that converse implication also holds. That is, $\MLANG$-equivalence coincides with
$\CLK{3}$-equivalence. 

\subsection{\change{Example of the impact of the presence of Schur-Hadamard product}}\label{subsec:example_hada}

\change{We start with an example showing what extra information can be computed from graphs when the Schur-Hadamard product is present.}
\begin{example}\normalfont
	We recall that in expression $\#\mathsf{3degr}(X)$ in Example~\ref{ex:degree3}, products of diagonal matrices resulted in the ability to zoom in on \emph{vertices} that carry specific degree information. When diagonal matrices are concerned, the product of matrices coincides with pointwise multiplication of the \emph{vectors} on the diagonals. Allowing pointwise multiplication on matrices has the same effect, but now on \emph{edges} in graphs. As an example, suppose that we want to count the number of ``triangle walks'' in $G$, i.e., walks $(v_0,\ldots,v_k)$ of length $k$ in $G$ such that each edge $\{v_{i-1},v_i\}$ in the  walk is part of a triangle. This can be done by expression
\begin{linenomath}
	\[ \#\Delta\mathsf{paths}_k(X):=\ones(X)^*\cdot((\Apply[f_{>0}](X^2 \odot X))^k\cdot\ones(X), \]
	\end{linenomath}
	where $f_{>0}(x)=1$ if $x\neq 0$ and $f_{>0}(x)=0$ otherwise\footnote{The use of $\Apply[f_{>0}](\cdot)$ is just for convenience. Its application inside sentences can be simulated with operations in $\ML{\cdot,{}^*,\,\tr,\ones,\allowbreak \diag,\odot}$ when evaluated on given adjacency matrices.}. Indeed, when evaluated on adjacency matrix $A_G$, $A_G^2\odot A_G$ extracts from $A_G^2$ only those entries corresponding to paths $(u,v,w)$ of length $2$ such that $(u,w)$ is an edge as well, i.e., it identifies edges involved in triangles in $G$. Then, $\Apply[f_{>0}](A_G^2\odot A_G)$ sets all non-zero entries to $1$. By considering the $k^{\text{th}}$ power of this matrix and summing up all its entries, the number of triangle paths of length $k$ is obtained. It can be verified that for graphs $G_5$ (\!\raisebox{-1.3ex}{\mbox{ 
	\includegraphics[height=0.6cm]{graphG5}}}) and $H_5$ (\!\raisebox{-1.3ex}{\mbox{ 
	\includegraphics[height=0.6cm]{graphH5}}}), $\#\Delta\mathsf{paths}_2(A_{G_5})=[160]\neq[132]=\#\Delta\mathsf{paths}_2(A_{H_5})$ and hence, they can be distinguished when the Schur-Hadamard product is available. Recall that all previous fragments could not distinguish between these two graphs. ~\hfill$\qed$ 
\end{example}

As mentioned earlier, we will use the Schur-Hadamard product to compute  \textit{stable edge partitions} of graphs, obtained as the result of the edge colouring or $2$-dimensional Weisfeiler-Lehman ($\textsf{2WL}$) algorithm~\cite{Bastert2001,Cai1992,O2017,Weisfeiler1968}. Such partitions can be seen as a generalization of equitable partitions, but now partitioning all pairs of vertices, rather than single vertices. \change{We detail these notions next.}

	 \SetAlgoCaptionLayout{small} \SetAlgoVlined 
\begin{algorithm}
	[t] \caption{Computing the stable colouring based on algorithm \textsc{2-Stab}~\cite{Bastert2001}.}\label{alg:stabcol}  \SetKwInOut{Input}{Input}\SetKwInOut{Output}{Output} \ResetInOut{Output}
	 \Input{A graph $G=(V,E)$ of order $n$.} \Output{Stable  colouring $\chi:V\times V\to C$.} Let $\chi:=\chi_0$\; Let $C:=\{0,1,2\}$\; \Repeat{{\normalfont $|C|$ does not change}}{ \For{$(v_1,v_2)\in V\times V$} { Compute $\mathsf{L}^2(v_1,v_2)$ relative to $\chi$\; } Replace $C$ by a minimal set of new colours $C'$ and define $\chi':V\times V\to C'$ such that\\
	\For{pairs $(v_1,v_2)$, $(v_1',v_2')$ in $V\times V$} {$\chi'(v_1,v_2)=\chi'(v_1',v_2') \Leftrightarrow \mathsf{L}^2(v_1,v_2)=\mathsf{L}^2(v_1',v_2')$} Let $C:=C'$\; Let $\chi:=\chi'$\; }
	
	 \end{algorithm}

\subsection{Stable edge partitions}\label{subsec:stable_edge_hada}
\change{We first recall the notion of stable edge partition and define when two graphs are
indistinguishable by the $\textsf{2WL}$ algorithm. Then,
similarly to the proof of Proposition~\ref{prop:equipart}, we show that when two
graphs are indistinguishable by sentences in $\ML{\cdot,\tr,{}^*,\ones,\diag,\odot}$, then they are indistinguishable by the $\textsf{2WL}$ algorithm (Proposition~\ref{prop:equiWL} below).} 

As already mentioned, the stable edge partition of a graph $G=(V,E)$ arises as the result of applying the \textsf{2WL}-algorithm~\cite{Bastert2001,Cai1992,O2017,Weisfeiler1968} on $G$. In Algorithm~\ref{alg:stabcol} we provide the pseudo-code of the algorithm \textsc{2-Stab}, taken from Bastert~\cite{Bastert2001}, which implements  the \textsf{2WL}-algorithm.  In a nutshell, the algorithm starts by assigning every vertex pair a colour, and then revises colourings iteratively based on some structural information. When no revision of the colouring occurs, the colouring has stabilized, the algorithm stops and returns the stable colouring. Colourings naturally induce partitions of $V\times V$, by simply grouping together vertex pairs with the same colour. The stable edge partition of $G$ is the partition induced by the stable colouring returned by \textsc{2-Stab}. The algorithm \textsc{2-Stab} needs at most $n^2$ iterations when evaluated on a graph of order $n$.

More precisely, \change{in this context,} a \textit{colouring} $\chi$ assigns a colour to each vertex pair in $V\times V$, i.e., if we denote by $C$ a set of colours, it is a function $\chi:V\times V\to C$.  The partition of $V\times V$ induced by $\chi$ is denoted by $\Pi_\chi(G)$ and will  be represented by \textit{indicator matrices}, one for each colour $c\in C$. \change{Assume that $V=\{1,\ldots,n\}$. An indicator matrix for a subset $Y\subseteq V\times V$
is a matrix $E_Y$ in $\R^{n\times n}$ such that $(E_Y)_{v_1v_2}=1$ if $(v_1,v_2)\in Y$
and $(E_Y)_{v_1v_2}=0$ otherwise. For a colour $c\in C$, we denote by $E_c$ the indicator matrix for $\{(v_1,v_2)\in V\times V\mid \chi(v_1,v_2)=c\}$.}
Hence,  $\Pi_\chi(G)$ is represented by the indicator matrices $E_{c}$, for $c\in C$.

Algorithm  \textsc{2-Stab} starts (on lines 1 and 2) with an initial colouring $\chi_0:V\times V\to\{0,1,2\}$ encoding adjacency, non-adjacency and loop information. More precisely, for vertices $v,w\in V$, $\chi_0(v,v)=2$, $\chi_0(v,w)=1$ if $(v,w)\in E$, and $\chi_0(v,w)=0$ for $v\neq w$ and $(v,w)\not\in E$. Then,  \textsc{2-Stab} adjusts the current colouring in each iteration, as follows.

Suppose that the current colouring is $\chi:V\times V\to C$. Given this colouring, for each pair of vertices $v_1,v_2\in V$, the so-called \textit{structure list} $\mathsf{L}^2(v_1,v_2)$ is computed (lines 4 and 5). To define these lists, the \textit{structure constants} are needed, which are defined as
\begin{linenomath}
\[
 p_{v_1,v_2}^{c,d}:=|\{ v_3\in V\mid \chi(v_1,v_3)=c, \chi(v_3,v_2)=d\}|, \]\end{linenomath}
for colours $c$ and $d$ in $C$ and vertices $v_1$ and $v_2$ in $V$. These numbers count the number of triangles\footnote{With a triangle one simply means a triple $(v_1,v_2)$, $(v_1,v_3)$ and $(v_2,v_3)$ of vertex pairs, none of which has to be an edge in $G$.}, based on $(v_1,v_2)$ whose other two pairs of vertices $(v_1,v_3)$ and $(v_3,v_2)$ have prescribed colours $c$ and $d$, respectively. Then, in a structure list we simply gather all these numbers for a specific vertex pair. That is, 
\begin{linenomath}
\[ \mathsf{L}^2(v_1,v_2):=\{ (c,d,p_{v_1,v_2}^{c,d})\mid p_{v_1,v_2}^{c,d}\neq 0\}.  \]\end{linenomath}
Based on this information,  \textsc{2-Stab} will assign new colours to pairs of vertices (lines 6--8). More precisely, $C$ is replaced by a minimal set of colours $C'$ such that each unique $\mathsf{L}^2(v_1,v_2)$ corresponds precisely to a single colour $c'$ in $C'$. Hence, the new colouring $\chi':V\times V\to C'$ will assign $(v_1',v_2')$ the colour $c'$, corresponding to $\mathsf{L}^2(v_1,v_2)$,  when  $\mathsf{L}^2(v_1,v_2)= \mathsf{L}^2(v_1',v_2')$. It is easily verified that the partition $\Pi_{\chi'}(G)$ is a refinement of $\Pi_{\chi}(G)$, which in turn is a refinement of $\Pi_{\chi_0}(G)$. 

Algorithm
 \textsc{2-Stab} now replaces $\chi$ by $\chi'$ and $C$ by $C'$ (lines 9 and 10), and repeats this process until the number of colours remains fixed (line 11). In other words, the corresponding partition is not further refined. The algorithm then returns this final (stable) colouring. 
 
 The \textit{stable edge partition of $G$}, denoted by $\Pi(G)$, is now the partition induced by the stable colouring.  It is known that $\Pi(G)$ is the unique coarsest partition of $V\times V$ which refines $\Pi_{\chi_0}(G)$ and corresponding to a colouring satisfying the stability condition stated on lines 7 and 8 in Algorithm~\ref{alg:stabcol}. 

Two graphs $G=(V,E)$ and $H=(W,F)$ of the same order are now said to be \textit{indistinguishable by the \textsf{2WL} algorithm}, denoted by $G\equiv_{\mathsf{2WL}} H$, if the stable
edge partitions $\Pi(G)$ and $\Pi(H)$ of $G$ and $H$, respectively, are (i)~of the form $\Pi(G)=\{E_{c_1},\ldots,E_{c_\ell}\}$ and $\Pi(H)=\{F_{c_1},\ldots,F_{c_\ell}\}$, that is, the parts in the partitions correspond to the same colours; and (ii)~the corresponding parts in these partitions have the same size, that is, $E_{c_i}$ and $F_{c_i}$ have the same number of entries carrying the value $1$, for $i=1,\ldots,\ell$.

In the seminal paper by Cai, F\"urer and Immerman~\cite{Cai1992}, the connection with logical indistinguishability was made.
\begin{theorem}\label{thm:cfi} 
	Let $G$ and $H$ be two graphs of the same order. Then,  $G\equiv_{\mathsf{2WL}} H$  if and only if $G\equiv_{\CLK{3}} H$. ~\hfill$\qed$ 
\end{theorem}

We next show that $\ML{\cdot,\,\tr,{}^*,\ones,\allowbreak \diag,\odot}$-equivalence implies indistinguishability by the \textsf{2WL} algorithm.  
\begin{proposition}\label{prop:equiWL}
Let $G$ and $H$ be graphs of the same order. Then, $G\equiv_{\ML{\cdot,\,\tr,{}^*,\ones,\allowbreak \diag,\odot}} H$ implies that $G\equiv_{\mathsf{2WL}} H$.\end{proposition}
\begin{proof}
	The overall proof is similar  (both in terms of structure as strategy) to the proof of Proposition~\ref{prop:equipart}, but using indicator matrices (representing the edge partitions) instead of indicator vectors (which represented the vertex partitions), and by relying on the algorithm  \textsc{2-Stab} to compute the  stable edge partition of a graph instead of algorithm \textsc{GDCR} (which computed the coarsest equitable partition). First, to simplify the construction of the expressions later on, we allow for 
	addition and scalar multiplication. \change{It can be verified (see appendix) that	$G\equiv_{\ML{\cdot,\,\tr,{}^*,\ones,\allowbreak \diag,\odot}} H$ implies
	$G\equiv_{\ML{\cdot,\,\tr,{}^*,\ones,\allowbreak \diag,\odot,+,\times}} H$. } We may thus indeed consider
 $\ML{\cdot,\,\tr,{}^*,\ones,\allowbreak \diag,\odot,+,\times}$
from now on. 

\change{The proof consists of the following two steps:
		\begin{itemize}
	\item[(a)] We first construct a number of expressions in $\ML{\cdot,\,\tr,{}^*,\ones,\allowbreak \diag,\odot,+,\times}$, denoted by $\mathsf{stabcol}_c(X)$, for $c\in C$ and $C$ a set of colours.
	 The key property of these expressions is that when they are evaluated on the adjacency matrix $A_G$ of $G$, 
	$\mathsf{stabcol}_c(A_G)$, for $c\in C$, correspond to indicator matrices representing the stable edge partition of  $G$. 
	\item[(b)] 
	The construction of the expressions $\mathsf{stabcol}_c(X)$, for $c\in C$, depend on $A_G$. As such, it is not guaranteed that $\mathsf{stabcol}_c(A_H)$, for $c\in C$, correspond to indicator matrices representing the stable edge partition of $H$. We show, however, when  $G\equiv_{\ML{\cdot,\,\tr,{}^*,\ones,\allowbreak \diag,\odot,+,\times}} H$.  holds, then $\mathsf{stabcol}_c(A_H)$, for $c\in C$, indeed correspond to indicator matrices representing the stable edge partition of $H$. To show this, we construct a number of sentences in 
	$\ML{\cdot,\,\tr,{}^*,\ones,\allowbreak \diag,\odot,+,\times}$.
Along the way, based on the construction of the expressions $\mathsf{stabcol}_c(X)$, for $c\in C$, we show that  $G\equiv_{\ML{\cdot,\,\tr,{}^*,\ones,\allowbreak \diag,\odot,+,\times}} H$ implies that $G\equiv_{\mathsf{2WL}} H$ holds.
	\end{itemize}}

	\medskip
	\noindent
		\change{\textbf{(a)~Compute the stable edge partition of a graph.}}
	 	 Given $G=(V,E)$, let $\Pi(G)=\{E_{c_1},\ldots,E_{c_\ell}\}$ be its stable edge partition. We show that we can construct expressions $\mathsf{stabcol}_{c_i}(X)$ in $\ML{\cdot,\tr,{}^*,\allowbreak \ones,\allowbreak \diag,\odot,+,\times}$, such that  $E_{c_i}=\mathsf{stabcol}_{c_i}(A_G)$, for $i=1,\ldots,\ell$. The expressions are constructed by simulating the run of the algorithm  \textsc{2-Stab} on $A_G$.

The initialisation step of  \textsc{2-Stab} is easy to simulate in $\ML{\cdot,\tr,{}^*,\ones,\allowbreak \diag,\odot,+,\times}$. Indeed, we simply consider expressions $\mathsf{stabcol}_2^{(0)}(X):=\diag(\ones(X))$; $\mathsf{stabcol}_1^{(0)}(X):=X$; and finally,  $\mathsf{stabcol}_0^{(0)}(X):=\ones(X)\mdot  (\ones(X))^*-X-\diag(\ones(X))$. Then, the indicator matrices $\mathsf{stabcol}_0^{(0)}(A_G)$, $\mathsf{stabcol}_1^{(0)}(A_G)$, and $\mathsf{stabcol}_2^{(0)}(A_G)$ represent the initial partition $\Pi_{\chi_0}(G)=\{E_0,E_1,E_2\}$ corresponding to the initial colouring $\chi_0$. 

Suppose now that after iteration $i$, the current set of colours  is $C$ and the  colouring is $\chi:V\times V\to C$.  Assume, by induction, that we have expressions $\mathsf{stabcol}_c^{(i)}(X)$ in  $\ML{\cdot,\tr,{}^*,\ones,\allowbreak \diag,+,\times,\odot}$, one for each $c\in C$, such that $\mathsf{stabcol}_c^{(i)}(A_G)$ is an indicator matrix representing the part in the edge partition $\Pi_{\chi}(G)$, induced by $\chi$, for colour $c$. Given these, we next construct  expressions for the refined partition computed by \textsc{2-Stab} in the next iteration. 

First, for  each pair of colours $(c,d)$ in $C$, we consider the expression
\begin{linenomath}\[ P_{c,d}^{(i+1)}(X):= \mathsf{stabcol}_c^{(i)}(X)\mdot \mathsf{stabcol}_d^{(i)}(X). \]
\end{linenomath}
On input $A_G$, it is readily verified that $P_{c,d}^{(i+1)}(A_G)$ is a matrix whose entry corresponding to vertices $v_1$ and $v_2$ holds the value $p_{v_1,v_2}^{c,d}$.

Let ${\cal P}_{c,d}^{(i+1)}$ be the set of  numbers occurring in  $P_{c,d}^{(i+1)}(A_G)$. For each value $p$ in ${\cal P}_{c,d}^{(i+1)}$, we now extract an indicator matrix indicating the positions in $P_{c,d}^{(i+1)}(A_G)$ that hold value $p$. 

This can be done using an expression $\mathsf{ind}_{c,d,p}^{(i+1)}(X)$ which works in a similar way as $\#\mathsf{3deg}(X)$ in Example~\ref{ex:degree3}, but uses the Schur-Hadamard product instead of products of diagonal matrices. The following example
illustrates the underlying idea  (see also the Schur-Wielandt Principle~\cite{Pech2002} mentioned before).
\begin{example}\label{ex:sw-matrix}\normalfont
Consider  $P_{c,d}=\begin{pmatrix} 2 & 0 & 3\\
1 & 3 & 2 \\
0 & 2 & 3\end{pmatrix}$ with ${\cal P}_{c,d}=\{0,1,2,3\}$. Suppose that we want to find all entries holding value $3$. This can be computed, as follows:
\begin{linenomath}
\small \[
\begin{pmatrix} 0 & 0 & 1\\
0 & 1 & 0 \\
0 & 0 & 1\end{pmatrix}=
\frac{1}{6}\times\left(
\begin{pmatrix} 2 & 0 & 3\\
1 & 3 & 2 \\
0 & 2 & 3\end{pmatrix}\odot
\left(\begin{pmatrix} 2 & 0 & 3\\
1 & 3 & 2 \\
0 & 2 & 3\end{pmatrix}- \begin{pmatrix} 1 & 1 & 1\\
1 & 1 & 1 \\
1 & 1 & 1\end{pmatrix}\right)\odot
\left(\begin{pmatrix} 2 & 0 & 3\\
1 & 3 & 2 \\
0 & 2 & 3\end{pmatrix}- \begin{pmatrix} 2 & 2 & 2\\
2 & 2 & 2 \\
2 & 2 & 2\end{pmatrix}\right)\right),
\]\end{linenomath}
where $\frac{1}{6}=\frac{1}{3(3-1)(3-2)}$, just as in  Example~\ref{ex:degree3}.~\hfill$\qed$
\end{example}
More generally, to identify positions that hold a specific value in $P_{c,d}^{(i+1)}(A_G)$, we consider the expression $\mathsf{ind}^{(i+1)}_{c,d,p}(X)$ defined by
\begin{linenomath}\[
\Biggl(\frac{1}{\prod_{p'\in {\cal P}^{(i+1)}_{c,d},p\neq p'}(p-p')}\Biggr) \times\!\!\!\!\!\!\! \bigodot_{p'\in {\cal P}^{(i+1)}_{c,d},p\neq p'} \!\!\!\!\!\!\!\!\!\!\bigl(P^{(i+1)}_{c,d}(X)-p'\times (\ones(X)\mdot (\ones(X))^*)\bigr).
\]\end{linenomath}
It should be clear from Example~\ref{ex:sw-matrix} that $\mathsf{ind}^{(i+1)}_{c,d,p}(A_G)$ indeed results in the desired indicator matrix. We note that the expression
$\mathsf{ind}^{(i+1)}_{c,d,p}(X)$ depends on the values in ${\cal P}^{(i+1)}_{c,d}$ and hence also depends on $A_G$.

Let $C'$ be the new set of colours assigned by \textsf{2-Stab}$(G)$ during the current iteration. As mentioned earlier, each colour $c$ in $C'$ is in correspondence with  $\mathsf{L}^2(v_1,v_2)$ for some vertices $v_1$ and $v_2$. Let us pick a colour $c$ in $C'$ and assume that it corresponds to 
$$\mathsf{L}^2(v_1,v_2)=\{(c_1,d_2,p_{v_1,v_2}^{c_1,d_1}),\ldots,(c_s,d_s,p_{v_1,v_2}^{c_s,d_s})\}.$$
We next use $\mathsf{ind}^{(i+1)}_{c,d,p}(X)$ and the Schur-Hadamard product to identify all vertex pairs that are assigned colour $c$, as follows: 
 \begin{linenomath}\[\mathsf{stabcol}_c^{(i+1)}(X):= \mathsf{ind}^{(i+1)}_{c_1,d_2,p_{v_1,v_2}^{c_1,d_1}}(X)\odot\cdots \odot\mathsf{ind}^{(i+1)}_{c_s,d_s,p_{v_1,v_2}^{c_s,d_s}}(X). 
\]\end{linenomath}
In other words, we use the Schur-Hadamard product to simulate the ``conjunction'' of the binary matrices representing the vertex pairs $(v_1,v_2)$ having non-zero structure constants $p_{v_1,v_2}^{c_i,d_i}$, for $i=1,\ldots,s$.
It is now easily verified that, on input $A_G$, $\mathsf{stabcol}_c^{(i+1)}(A_G)$ returns an indicator matrix in which the entries holding a $1$ correspond precisely to the pairs $(v_1',v_2')\in V\times V$ such that $\mathsf{L}^2(v_1',v_2')=\mathsf{L}^2(v_1,v_2)$ where $\mathsf{L}^2(v_1,v_2)$  corresponds to colour $c$. In other words, $\mathsf{stabcol}_c^{(i+1)}(A_G)$ represents the refined edge partition corresponding to the part associated with colour $c$. We do this for every colour in $C'$. Clearly,  $\mathsf{stabcol}_c^{(i+1)}(A_G)$, for $c\in C'$, represent the refined partition $\Pi_{\chi'}(G)$ corresponding to $\chi':V\times V\to C'$.

We continue in this way until the colouring stabilises. i.e., no further colours are needed. We denote the final set of colours by $C$ and by $\mathsf{stabcol}_{c}(X)$, for $c\in C$, the 
 $\ML{\cdot,\tr,{}^*,\ones,\allowbreak \diag,\odot,+,\allowbreak \times}$ expressions computing the parts $E_c$ in $\Pi(G)$. The correctness of these expressions follows from the previous arguments and the correctness of the algorithm \textsc{2-Stab}.

	\medskip
	\noindent
		\change{\textbf{(b)~Verifying that $G\equiv_{\mathsf{2WL}} H$.}}

Just as in the proof of Proposition~\ref{prop:equipart}, the expressions $\mathsf{stabcol}_c(X)$ depend on $A_G$ since  we explicitly used the values  occurring in  $P^{(i)}_{c,d}(A_G)$ and the colours assigned to vertex pairs  during each iteration $i$ of the execution of  \textsc{2-Stab} on $G$. Let $\Pi(H)$ be stable edge partition of $H$.
We next show that $G \equiv_{\ML{\cdot,\tr,{}^*,\ones,\odot,+,\times}} H$ implies that $\Pi(H)$ consists of $\mathsf{stabcol}_{c}(A_H)$, for $c\in C$. Furthermore, we show that the number of ones in 
 $\mathsf{stabcol}_{c}(A_G)$ and  $\mathsf{stabcol}_{c}(A_H)$  agree for all $c\in C$. Hence, $G$ and $H$ are indistinguishable by  the \textsf{2WL} algorithm.

The proof is by induction on the number of iterations of $\textsc{2-Stab}(G)$ and $\textsc{2-Stab}(H)$.
We denote by $\chi_G^{(i)}:V\times V\to C^{(i)}_G$ and $\chi_H^{(i)}:W\times W\to C^{(i)}_H$ the colourings used in the $i^{\text{th}}$ iteration of  $\textsc{2-Stab}(G)$ and $\textsc{2-Stab}(H)$, respectively. 
The induction hypothesis is that  $G \equiv_{\ML{\cdot,\tr,{}^*,\ones,\odot,+,\times}} H$ implies that $C^{(i)}_G=C^{(i)}_H=C^{(i)}$ and furthermore that for each $c\in C^{(i)}$, $\mathsf{stabcol}_{c}^{(i)}(A_H)$ is an indicator
matrix, and all $\mathsf{stabcol}_{c}^{(i)}(A_H)$ together constitute the edge partition $\Pi_{\chi^{(i)}_H}(H)$. Moreover, we show that for each $c\in C^{(i)}$, $\mathsf{stabcol}_{c}^{(i)}(A_G)$ and $\mathsf{stabcol}_{c}^{(i)}(A_H)$ have the same number of ones. This clearly suffices, for if this holds, $\mathsf{stabcol}_{c}(A_H)$, for $c\in C$, constitute $\Pi(H)$ and $\mathsf{stabcol}_{c}(A_G)$ and  $\mathsf{stabcol}_{c}(A_H)$ have the same number of ones, for all $c\in C$. 

We start by verifying the hypothesis for the base case, i.e., when $i=0$.
Clearly, $\chi_G^{(0)}$ and $\chi_H^{(0)}$ use the same colours $C^{(0)}_G=C^{(0)}_H=C^{(0)}=\{0,1,2\}$. 
By definition of the expressions $\mathsf{stabcol}_c^{(0)}(X)$,  all $\mathsf{stabcol}_c^{(0)}(A_H)$ together represent $\Pi_{\chi^{(0)}_H}(H)$.
Moreover,  by considering the sentences
\begin{linenomath}
\[
\mathsf{\#ones}^{(0)}_c(X):=(\ones(X))^*\mdot \mathsf{stabcol}_c^{(0)}(X)\cdot\ones(X),
\]\end{linenomath}
for $c\in C^{(0)}$, $G \equiv_{\ML{\cdot,\tr,{}^*,\ones,\odot,+,\times}} H$ implies that $\mathsf{\#ones}^{(0)}_c(A_G)=\mathsf{\#ones}^{(0)}_c(A_H)$. 
Hence, we may conclude that $\mathsf{stabcol}_{c}^{(0)}(A_G)$ and $\mathsf{stabcol}_{c}^{(0)}(A_H)$ have the same number of ones, as desired.

Suppose, by induction, that $G \equiv_{\ML{\cdot,\tr,{}^*,\ones,\odot,+,\times}} H$ implies that $\chi_G^{(i)}:V\times V\to C^{(i)}_G$ and $\chi_H^{(i)}:W\times W\to C^{(i)}_H$ with 
$C^{(i)}_G=C^{(i)}_H=C^{(i)}$. Furthermore, the current edge partition $\Pi_{\chi^{(i)}_H}(H)$ of $H$ is represented by  $\mathsf{stabcol}_c^{(i)}(A_H)$, for $c\in C^{(i)}$. Furthermore,
for each $c\in C^{(i)}$, the number of ones in $\mathsf{stabcol}^{(i)}_c(A_H)$ and $\mathsf{stabcol}^{(i)}_c(A_G)$ agree. 

As before, let ${\cal P}^{(i+1)}_{c,d}$ be the set of values occurring in $P_{c,d}^{(i+1)}(A_G)$ and consider the expressions $\mathsf{ind}^{(i+1)}_{c,d,p}(X)$ for $c,d\in C^{(i)}$ and $p\in {\cal P}^{(i+1)}_{c,d}$.
We show that $\mathsf{ind}^{(i+1)}_{c,d,p}(A_H)$ is a binary matrix as well containing the same number of ones as $\mathsf{ind}^{(i+1)}_{c,d,p}(A_G)$. 
This implies that each value $p\in {\cal P}^{(i+1)}_{c,d}$ occurs in $P_{c,d}^{(i+1)}(A_H)$ and
moreover, it occurs the same number of times as in $P_{c,d}^{(i+1)}(A_G)$. Hence, the set of values occurring in $P^{(i+1)}_{c,d}(A_H)$ is the same as those occurring
in $P_{c,d}^{(i+1)}(A_G)$.

To check that $\mathsf{ind}^{(i+1)}_{c,d,p}(A_H)$  is a binary matrix, we use the sentence
\begin{linenomath}\[
\mathsf{binary}(X):=(\ones(X))^*\cdot\bigl((X\odot X- X)\odot(X\odot X-X)\bigr)\cdot\ones(X).
\]\end{linenomath}
This sentence will return $[0]$, when given a real matrix as input, if and only if the input matrix is a binary matrix. We have seen a similar expression in the proof of Proposition~\ref{prop:equipart}.
Indeed, for a binary
matrix $B$, $B\odot B=B$ and hence $B\odot B-B=Z$, where $Z$ is the zero matrix. Since $Z\odot Z=Z$, $\mathsf{binary}(B)=\tp{\ones}\mdot Z\mdot \ones=[0]$. For the converse,
assume that $\mathsf{binary}(B)=[0]$. We observe that each entry in $(B\odot B- B)\odot(B\odot B-B)$ is non-negative value.
Indeed, all entries are squares of real numbers. Hence, when $\mathsf{binary}(B)=[0]$, the sum of all these squared entries must be zero.
This implies that $B\odot B-B=Z$. This in turn implies that $B$ can only 
contain $0$ or $1$ as entries, since these are the only real  values satisfying $x^2-x=0$.
Hence, when $G \equiv_{\ML{\cdot,\tr,{}^*,\ones,\odot,+,\times}} H$ holds, then since all $\mathsf{ind}_{c,d,p}^{(i+1)}(A_G)$, for $c,d\in C^{(i)}$ and $p\in {\cal P}^{(i+1)}_{c,d}$, are binary matrices,
 \begin{linenomath} \[
 \mathsf{binary}(\mathsf{ind}^{(i+1)}_{c,d,p}(A_G))=[0]=\mathsf{binary}(\mathsf{ind}^{(i+1)}_{c,d,p}(A_H)).
 \]\end{linenomath}
So indeed, $\mathsf{ind}^{(i+1)}_{c,d,p}(A_H)$ is a binary matrix as well. 

The new colours in  $\textsc{2-Stab}(G)$ are assigned based on the structure lists $\mathsf{L}^2(v_1,v_2)$. We show that for every unique structure list 
$\mathsf{L}^2(v_1,v_2)$ there is a pair of vertices $w_1,w_2$ in $W$ such that $\mathsf{L}^2(v_1,v_2)=\mathsf{L}^2(w_1,w_2)$. This implies that $\textsc{2-Stab}(H)$
will use the same colours for refining $\chi_H^{(i)}$ as $\textsc{2-Stab}(G)$ uses to refine $\chi_G^{(i)}$. Hence, the revised colourings $\chi_G^{(i+1)}:V\times V\to C_G^{(i+1)}$
and $\chi_H^{(i+1)}:W\times W\to C_H^{(i+1)}$ satisfy indeed that $C_G^{(i+1)}=C_H^{(i+1)}=C^{(i+1)}$. 

Consider a structure list $\mathsf{L}^2(v_1,v_2)$ and assume that it corresponds to a new colour $c\in C_G^{(i+1)}$. We know that $\mathsf{stabcol}_c^{(i+1)}(A_G)$ returns the indicator matrix indicating which vertex pairs in $V\times V$ have
this structure list (colour $c$). The expression $\mathsf{stabcol}_c^{(i+1)}(X)$ consists of the Schur-Hadamard product of $\mathsf{ind}^{(i+1)}_{c,d,p}(X)$ for every $(c,d,p)$ in $\mathsf{L}^2(v_1,v_2)$.
We have shown above that  $\mathsf{ind}^{(i+1)}_{c,d,p}(A_G)$ and  $\mathsf{ind}^{(i+1)}_{c,d,p}(A_H)$ contain the same number of ones, meaning that there are vertex pairs $(w_1,w_2)\in W\times W$
for which $p_{w_1,w_2}^{c,d}=p=p_{v_1,v_2}^{c,d}$.  Furthermore, in a similar way as above, we can show that $G \equiv_{\ML{\cdot,{}^*,\tr,\ones,\odot,+,\times}} H$ implies that 
$\mathsf{stabcol}_c^{(i+1)}(A_H)$ is a binary matrix which consists of the same number of ones as $\mathsf{stabcol}_c^{(i+1)}(A_G)$. So,  $\textsc{2-Stab}(H)$ needs the
same set of  colours $C_G^{(i+1)}$ as  $\textsc{2-Stab}(G)$ in the refinement phase. Hence, we can take $C_G^{(i+1)}=C_H^{(i+1)}=C^{(i+1)}$.

By construction, $\mathsf{stabcol}_c^{(i+1)}(A_H)$ and $\mathsf{stabcol}_{c'}^{(i+1)}(A_H)$ do not have a common entry holding value $1$, for each distinct pair of colours $c,c'\in C^{(i+1)}$. 
We note that the number of entries holding value $1$ in all $\mathsf{stabcol}_c^{(i+1}(A_H)$ combined sum up $n^2$. Indeed, we know that this holds for $\mathsf{stabcol}_c^{(i+1)}(A_G)$ and we have
just shown that $\mathsf{stabcol}_c^{(i+1))}(A_H)$ consists of the same number of ones as $\mathsf{stabcol}_c^{(i+1)}(A_G)$. Hence, $\mathsf{stabcol}_c^{(i+1)}(A_H)$ also represent
a partition of $W\times W$, i.e., the partition $\Pi_{\chi_H^{(i+1)}}(H)$ and the induction hypothesis is  satisfied.
\end{proof}


\subsection{Main result}\label{subsec:c3}
We are now ready to show the main result of this section.
\begin{theorem}\label{thm:c3} 
	Let $G$ and $H$ be two graphs of the same order, then $G \equiv_{\ML{\cdot,\,\tr,{}^*,\ones,\diag,\odot}} H$ if and only if $G \equiv_{\MLANG} H$ if and only  if $G\equiv_{\CLK{3}} H$. 
\end{theorem}
\begin{proof}
We show that $G \equiv_{\ML{\cdot,\,\tr,{}^*,\ones,\diag,\odot}} H$ implies $G\equiv_{\CLK{3}} H$, and that
$G\equiv_{\CLK{3}} H$ implies  $G\equiv_{\MLANG} H$. Since $\ML{\cdot,\tr,{}^*,\ones,\diag,\odot}$ is a smaller fragment than $\MLANG$, $G\equiv_{\MLANG} H$ clearly implies $G \equiv_{\ML{\cdot,\,\tr,{}^*,\ones,\diag,\odot}} H$, resulting in the theorem.

Let us assume that $G \equiv_{\ML{\cdot,\,\tr,{}^*,\ones,\diag,\odot}} H$ holds. Then, the previous proposition implies that $G\equiv_{\mathsf{2WL}} H$.
 Combined with Theorem~\ref{thm:cfi}, this implies that $G\equiv_{\CLK{3}} H$. Next, we assume that $G\equiv_{\CLK{3}} H$  holds. We show that this implies that $G\equiv_{\MLANG} H$. We here rely on the connection between $\MLANG$ and  three-variable logics~\cite{Brijder2018}, which we first recall.
 
 In Proposition 4.2 in Brijder et al.~\cite{Brijder2018} it was shown that for every sentence $e(X)$ in $\MLANG$ there exists an equivalent formula $\varphi_e(z)$ in the relational calculus with aggregates which uses only three ``base variables''.  We will
not recall the syntax of this calculus formally (see~\cite{libkin_sql} for a full definition) but only recall that in this calculus, we have base variables and numerical variables. Base variables can be bound to base columns of relations, and compared for equality. Numerical variables can be bound to numerical columns, and can be equated to function applications and aggregates. 
The free variable $z$ in $\varphi_e(z)$  is a numeric variable since a scalar is returned by $e(X)$.

We now make the connection between matrices, on which $\MLANG$ expressions are evaluated, and such typed relations, on which calculus expressions are evaluated. More specifically, a matrix $A$ is encoded as a ternary relation $\textsf{Rel}(A)$ where two base columns are reserved for the indices of the matrix and the numerical column holds the value in each entry (vectors and scalars are represented analogously). 
It is now understood that the equivalence of  $e(X)$ and $\varphi_e(z)$ means that   $e(A_G)$  and the evaluation of $\varphi_e(z)$ on $\textsf{Rel}(A_G)$ results in the same scalar. Let $c=e(A_G) \in \C$ and consider the calculus sentence $\psi_e:=\exists z\, \varphi_e(z)\land z=c$. Following the arguments in the proof of Proposition \change{4.4} in~\cite{Brijder2018}, which in turn rely on  standard
translation techniques (see e.g.,~\cite{hlnw_aggregate,libkin_sql}), one can show that $\psi_e$ can be equivalently expressed by a sentence $\psi_e'$ in $C_{\infty\omega}^{3}$~\cite{otto_bounded}, i.e., in infinitary counting logic with three distinct (untyped) variables over binary relations. These binary relations encode graphs in a standard way by simply storing the edge relation.
It is known that $G\equiv_{C_{\infty\omega}^{3}} H$ if and only if $G\equiv_{\CLK{3}} H$~\cite{Hella1996}. By assumption  $G\equiv_{\CLK{3}} H$ and hence $G\equiv_{C_{\infty\omega}^{3}} H$. This implies that $\psi_e'(G)=\psi_e'(H)$ since $\psi_e'$ is a sentence in $C_{\infty\omega}^{3}$. Hence, also $\psi_e$ evaluates to true on both $\mathsf{Rel}(A_G)$ and $\mathsf{Rel}(A_H)$,
and $\varphi_e(z)$ returns the value $c$ on both $\mathsf{Rel}(A_G)$ and $\mathsf{Rel}(A_H)$. As a consequence, also $e(A_H)=c$ and $e(A_G)=e(A_H)$. Since this argument works for any $\MLANG$ sentence $e(X)$,  we have that $G\equiv_{\MLANG} H$. \end{proof}

The results by Dell et al.~\cite{Dell2018} also tell that $G\equiv_{\MLANG} H$ if and only if $\textsf{HOM}_{\cal F}(G)=\textsf{HOM}_{\cal F}(H)$ where
	${\cal F}$ consists of all graphs of tree-width at most two.
	
We conclude by providing an algebraic characterisation of $\MLANG$-equivalence based on a result by Dawar et al.~\cite{Dawar2016}. To state this result, we need the notion of 
coherent algebra (see e.g.,~\cite{Friedland1989}). The \textit{coherent algebra} $\mathfrak{C}(A_G)$ associated with $A_G$ is the smallest complex matrix algebra containing $A_G$, $I$, and $J$ and which is closed under the Schur-Hadamard product. \change{The coherent algebra $\mathfrak{C}(A_H)$ associated with $A_H$ is defined similarly.} The algebras $\mathfrak{C}(A_G)$ and $\mathfrak{C}(A_H)$ are said to be \textit{algebraically isomorphic} if there is bijection $\imath:\mathfrak{C}(A_G)\to \mathfrak{C}(A_H)$ which is
an algebra morphism which in addition satisfies: $\imath(J)=J$, $\imath(A^*)=(\imath(A))^*$ and $\imath(A\odot B)=\imath(A)\odot\imath(B)$, for all matrices $A, B\in\mathfrak{C}(A_G)$.

\begin{proposition}[Proposition 7 in Dawar et al.~\cite{Dawar2016}]\label{prop:dawar}
Let $G$ and $H$ be two graphs of the same order. Then, $G\equiv_{\CLK{3}} H$ if and only if  there exists an algebraic isomorphism $\imath:\mathfrak{C}(A_G)\to \mathfrak{C}(A_H)$ such that $\imath(A_G)=\imath(A_H)$.\hfill$\qed$
\end{proposition}

This correspondence can be made a bit more precise and in line with our previous characterisations. 

\begin{proposition}\label{thm:algiso} 
	Let $G$ and $H$ be two graphs of the same order, then $G \equiv_{\MLANG} H$ if and only if there exists an orthogonal matrix $O$ such that $E_c\mdot O=O\mdot F_c$, for $c\in C$, where
$E_c$ and $F_c$, for $c\in C$, constitute  the stable edge partitions $\Pi(G)$ and $\Pi(H)$ of $G$ and $H$, respectively. (Here, $C$ denotes the set of colours used by the colourings that induce
the partitions).
		\end{proposition}

\begin{proof}
We know from Proposition~\ref{prop:equiWL} that $G \equiv_{\MLANG} H$ implies that $G\equiv_{\mathsf{2WL}} H$. Moreover, we can compute $\Pi(G)$ and $\Pi(H)$ by means of the expressions $\mathsf{stabcol}_c(X)$ in $\MLANG$. Let $C=\{c_1,\ldots,c_\ell\}$ be the set of colours used in these partitions. Just as in the proof of Theorem~\ref{thm:treqpart}, we consider sentences
$e_w(X):=\tr\bigl(w(\mathsf{stabcol}_{c_1}(X),\ldots,\allowbreak \mathsf{stabcol}_{c_\ell}(X))\bigr)$ for some word $w$ over $\ell$ variables. Then, $G \equiv_{\MLANG} H$ implies that
$e_w(A_G)=e_w(A_H)$ for any such word $w$, and  thus by the real version of Specht's Theorem, there exists an orthogonal matrix $O$ such that $\mathsf{stabcol}_{c}(A_G)\mdot O=O\mdot \mathsf{stabcol}_{c}(A_H)$ for all $c\in C$, as desired. In the application of Specht's Theorem it is crucial that $\Pi(G)$ and $\Pi(H)$ are closed under transposition. This known to hold, i.e., for every part $E_c$ in $\Pi(G)$ there is a part $E_{c'}$ such that $\tp{E_c}=E_{c'}$. \change{This property also holds for $\Pi(H)$} (see e.g.,~\cite{Bastert2001}). 

For the converse, suppose that  there exists an orthogonal matrix $O$ such that $E_c\mdot O=O\mdot F_c$, for $c\in C$. We note that this implies that $A_G\mdot O=O\mdot A_H$ since $A_G=\sum_{c\in D} E_c$ and $A_H=\sum_{c\in D} F_c$ for some subset of colours $D$ of $C$. This follows the fact that the  \textsf{2WL} algorithm refines the initial colouring, in which edges  are coloured differently than non-edges. 
So, a colour used for an edge in $G$ can only be used for an edge in $H$, and vice versa.
Moreover, it is known that
the binary matrices in $\Pi(G)$ and $\Pi(H)$ form a basis for $\mathfrak{C}(A_G)$ and $\mathfrak{C}(A_H)$, respectively. 
If we now consider $\imath:\mathfrak{C}(A_G)\to \mathfrak{C}(A_H):A\mapsto O\mdot A\cdot\tp{O}$, then this is known to be an algebraic isomorphism between  $\mathfrak{C}(A_G)$ and $\mathfrak{C}(A_H)$~\cite{Friedland1989}. Hence, by Proposition~\ref{prop:dawar}, $G\equiv_{\CLK{3}} H$ and thus also 
$G \equiv_{\MLANG} H$ by Theorem~\ref{thm:c3}.\end{proof}

\begin{remark}
	The orthogonal matrix $O$ in the statement of Proposition~\ref{thm:algiso} can be taken to be \change{compatible with the coarsest equitable partitions of $G$ and $H$, that witness that $G$ and $H$ have a common equitable partition.} This is in agreement with Theorem~\ref{thm:treqpart}. This follows from the fact that there is a subset $K$ of colours such that $I=\sum_{c\in K}E_c=\sum_{c\in K}F_c$~\cite{Bastert2001}. Furthermore, the diagonal matrices $E_c$, for $c\in K$, correspond to $\diag(\ones_{V_c})$ for the coarsest 
	equitable partition ${\cal V}=\{V_c\mid c\in K\}$ of $G$. Similarly, for $c\in K$, $F_c=\diag(\ones_{W_c})$ correspond to the coarsest equitable partition ${\cal W}=\{W_c\mid c\in K\}$ of $H$~\cite{Bastert2001}. 
\end{remark}

\begin{remark}
The proof of Proposition~\ref{thm:algiso} relied on results by Brijder et al.~\cite{Brijder2018} and Dawar et al.~\cite{Dawar2016} in which connections with $\CLK{3}$-equivalence were made. A direct proof of Proposition~\ref{thm:algiso} is possible. Indeed, it suffices to show that $O$-conjugation, for an orthogonal matrix $O$ such that $E_c\mdot O=O\mdot F_c$ holds for each colour $c\in C$, is preserved by all operations in $\MLANG$, including arbitrary pointwise functions on matrices.
We do not detail this further in this paper in order to keep the paper of reasonably length (the proof consists of a long case analysis in which all previous conjugation-preserving conditions need to be verified in the context of stable edge partitions). The crucial ingredient in all this is that one can verify that for any expression $e(X)$ in $\MLANG$, such that $e(A_G)$ returns a matrix,
we can write $e(A_G)=\sum_{c\in C} a_c\times E_c$ and $e(A_H)=\sum_{c\in C} a_c\times F_c$. This is generalization $\ML{\cal L}$-vectors being constant on equitable partitions, but now for  $\ML{\cal L}$-matrices being constant on stable edge partitions. The ability to rewrite $e(A_G)$ (and $e(A_H)$) in terms of the indicator matrices allows to show, e.g., that  $O$-conjugation is preserved by the Schur-Hadamard product and, more generally, by any pointwise function application on matrices.
\end{remark}

\vspace{-3ex}
\section{Conclusion} 
We have characterised $\ML{\cal L}$-equivalence for undirected graphs and identified what additional distinguishing power each of the operations in $\MLANG$ has. 
\change{Some of the results  generalise to directed graphs (with asymmetric adjacency matrices) or even arbitrary matrices. This is explored in an upcoming paper~\cite{Geerts20}.}
The extension to the case when queries can have multiple inputs is wide open.

Of interest may also be to connect $\ML{\cal L}$-equivalence to fragments of first-order logic (without counting). A possible line of attack could be to work over the boolean semiring instead of over the complex numbers (see Grohe and Otto~\cite{grohe_otto_2015} for a similar approach). More general semirings could open the way for modelling and querying labeled graphs using matrix query languages (see also~\cite{Brijder-abs-1904-03934}).

Another question is which additional linear algebra operations should be added to the matrix language $\MLANG$ such that $\CLK{k}$-equivalence can be captured, for $k\geq 4$. We refer to~\cite{AtseriasM13,grohe_otto_2015,MALKIN2014} characterisation of $\CLK{k}$-equivalence in terms of solutions of linear systems of equations, which may serve as inspiration.
Finally, connections between $\MLANG$ and rank logics, as studied in the context of the descriptive complexity of linear algebra~\cite{dghl_rank,Dawar2008,Dawar2017,GradelP15,GroheP17,holm_phd}, are worth exploring.

\medskip
{\bf Acknowledgement.}The author is grateful to Joeri Rammelaere (and his Python skills) for computing the numerical quantities of the example graphs used in this paper.

%

\appendix{}

\section*{Proof of Lemma~\ref{lem:multp-sim} } 
\setcounter{lemma}{0} 
\setcounter{section}{5}

\renewcommand{\thelemma}{\arabic{section}.\arabic{lemma}} 
\begin{lemma}
	Let $\ML{\cal L}$ be any matrix query language fragment and let $G$ and $H$ be graphs of the same order.
		Consider expressions  $e_1(X)$ and $e_2(X)$ in $\ML{\cal L}$. If $e_i(A_G)$ and $e_i(A_H)$ are $T$-similar, for $i=1,2$, for an \textit{arbitrary matrix} $T$, then $e_1(A_G)\mdot e_2(A_G)$ is also $T$-similar to $e_1(A_H)\mdot e_2(A_H)$.
\end{lemma}
\begin{proof}
	To show this lemma, we distinguish between the following cases, depending on the dimensions of $e_1(A_G)$ and $e_2(A_G)$ (or equivalently, the dimensions of $e_1(A_H)$ and $e_2(A_H)$). Let $e(X):=e_1(X)\mdot e_2(X)$. Let $n$ be the order of $G$ (and $H$). 
	\begin{itemize}
		\item ($\mathbf{n\times n, n\times n}$): $e_1(A_G)$ and $e_2(A_G)$ are of dimension $n\times n$. By assumption, $e_1(A_G)\mdot T=T\mdot e_1(A_H)$ and $e_2(A_G)\mdot T=T\mdot e_2(A_H)$. Hence, 
		\begin{linenomath}
			\[ e(A_G)\mdot T=e_1(A_G)\mdot e_2(A_G)\mdot T=e_1(A_G)\mdot T\mdot \allowbreak e_2(A_H)=T\mdot e_1(A_H)\mdot \allowbreak e_2(A_H)\allowbreak=T\mdot e(A_H).\]
		\end{linenomath}
		\item ($\mathbf{n\times n, n\times 1}$): $e_1(A_G)$ is of dimension $n\times n$ and $e_2(A_G)$ is of dimension $n\times 1$. By assumption, $e_1(A_G)\mdot T=T\mdot e_1(A_H)$ and $e_2(A_G)=T\mdot e_2(A_H)$. Hence, 
		\begin{linenomath}
			\[ e(A_G)=\allowbreak e_1(A_G)\mdot \allowbreak e_2(A_G)=\allowbreak e_1(A_G)\mdot T\mdot \allowbreak e_2(A_H)\allowbreak=T\mdot e_1(A_H)\cdot\allowbreak e_2(A_H)=T\mdot e(A_H).\]
		\end{linenomath}
		\item ($\mathbf{n\times 1, 1\times n}$): $e_1(A_G)$ is of dimension $n\times 1$ and $e_2(A_G)$ is of dimension $1\times n$. By assumption, $e_1(A_G)=T\mdot e_1(A_H)$ and $e_2(A_G)\mdot T=e_2(A_H)$. Hence, 
		\begin{linenomath}
			\[ e(A_G)\mdot T=e_1(A_G)\mdot e_2(A_G)\mdot T\allowbreak=e_1(A_G)\mdot e_2(A_H)\allowbreak= T\mdot e_1(A_H)\mdot e_2(A_H)=T\mdot e(A_H)).\]
		\end{linenomath}
		\item ($\mathbf{n\times 1, 1\times 1}$): $e_1(A_G)$ is of dimension $n\times 1$ and $e_2(A_G)$ is of dimension $1\times 1$. By assumption, $e_1(A_G)=T\mdot e_1(A_H)$ and $e_2(A_G)=e_2(A_H)$. Hence, 
		\begin{linenomath}
			\[ e(A_G)=e_1(A_G)\mdot e_2(A_G)=e_1(A_G)\mdot \allowbreak e_2(A_H)=T\mdot \allowbreak e_1(A_H)\mdot e_2(A_H)\allowbreak=T\mdot e(A_H)).\]
		\end{linenomath}
		\item ($\mathbf{1\times n, n\times n}$): $e_1(A_G)$ is of dimension $1\times n$ and $e_2(A_G)$ is of dimension $n\times n$. By assumption, $e_1(A_G)\mdot T=e_1(A_H)$ and $e_2(A_G)\mdot T=T\mdot e_2(A_H)$. Hence, 
		\begin{linenomath}
			\[ e(A_G)\mdot T=e_1(A_G)\mdot e_2(A_G)\mdot T=e_1(A_H)\cdot\allowbreak T \mdot e_2(A_H)=e_1(A_H)\mdot e_2(A_H)=e(A_H)).\]
		\end{linenomath}
		\item ($\mathbf{1\times n, n\times 1}$): $e_1(A_G)$ is of dimension $1\times n$ and $e_2(A_G)$ is of dimension $n\times 1$. By assumption, $e_1(A_G)\mdot T=e_1(A_H)$ and $e_2(A_G)=T\mdot e_2(A_H)$. Hence, 
		\begin{linenomath}
			\[ e(A_G)=e_1(A_G)\mdot e_2(A_G)=e_1(A_G)\mdot \allowbreak T\cdot\allowbreak e_2(A_H)= e_1(A_H)\mdot e_2(A_H)=e(A_H).\]
		\end{linenomath}
		
		\item ($\mathbf{1\times 1, 1\times n}$): $e_1(A_G)$ is of dimension $1\times 1$ and $e_2(A_G)$ is of dimension $1\times n$. By assumption, $e_1(A_G)=e_1(A_H)$ and $e_2(A_G)\mdot T=e_2(A_H)$. Hence, 
		\begin{linenomath}
			\[ e(A_G)\mdot T=e_1(A_G)\mdot e_2(A_G)\mdot T=e_1(A_G)\mdot e_2(A_H)=e_1(A_G)\mdot \allowbreak e_2(A_H)=e(A_H).\]
		\end{linenomath}
		
		\item ($\mathbf{1\times 1, 1\times 1}$): $e_1(A)$ and $e_2(A)$ are of dimension $1\times 1$. By assumption, $e_1(A_G)=e_1(A_H)$ and $e_2(A_G)=e_2(A_H)$. Hence, $ e(A_G)=e_1(A_G)\mdot e_2(A_G)=e_1(A_H)\mdot \allowbreak e_2(A_H)=e(A_H)$. 
	\end{itemize}
	This concludes the proof. 
\end{proof}

\section*{Continuation of the proofs of Proposition~\ref{prop:equipart} and Theorems~\ref{thm:pwtrpart} and~\ref{prop:equiWL}}

In all three proofs we relied on the presence of addition and scalar multiplication to compute either equitable partitions or stable edge partitions. Since addition and scalar multiplication are
used in the proofs before the proper conjugacy notion was identified, we cannot simply rely on Lemma~\ref{cor:ridoflinear}. We therefore show that when ${\cal L}$ is either $\{\mdot,{}^*,\ones,\diag\}$, $\{\mdot,\tr,{}^*,\diag\}$ (for Proposition~\ref{prop:equipart}), $\{\mdot,\tr,{}^*,\ones,\odot_v\}$ (Theorem~\ref{thm:pwtrpart}) or 
$\{\mdot,\tr,{}^*,\ones,\diag,\odot\}$ (Theorem~\ref{prop:equiWL}), 
 that $G\equiv_{\ML{\cal L}} H$ implies  $G\equiv_{\ML{{\cal L}^+}} H$, where
${\cal L}^+={\cal L}\cup \{+,\times\}$.

 We show this by verifying that any expression $e(X)$ in $\ML{{\cal L}^+}$ can be \textit{equivalently} written as a linear combination of expressions in $\ML{\cal L}$. We denote equivalence by $\equiv$ and 
 $e(X)\equiv e'(X)$ means that $e(A)=e'(A)$ for all matrices $A$.
 We verify our claim by induction on the structure of expressions.

\smallskip 
\noindent (\textbf{base case}) Let $e(X):=X$. This already has the desired form.

\smallskip 
\noindent (\textbf{matrix multiplication}) Let $e(X):=e_1(X)\cdot e_2(X)$. By induction we have that $e_1(X)\equiv \sum_{i=1}^p a_i\times e_i^1(X)$ and
$e_2(X)\equiv\sum_{i=1}^q b_i\times e_i^2(X)$. Hence, $e(X)\equiv \sum_{i=1}^p\sum_{j=1}^{q} (a_i\times b_i) \times (e_i^1(X)\mdot e_j^2(X)$). 

\smallskip 
\noindent (\textbf{complex conjugate transposition}) Let $e(X):=(e_1(X))^*$. By induction, $e_1(X)\equiv\sum_{i=1}^p a_i\times e_i^1(X)$. Then, $e(X)\equiv \sum_{i=1}^p \bar{a}_i\times (e_i(X))^*$.

\smallskip 
\noindent (\textbf{trace}) Let $e(X):=\tr(e_1(X))$. By induction,
$e_1(X)\equiv\sum_{i=1}^p a_i\times e_i^1(X)$. Then, $e(X)\equiv \sum_{i=1}^p a_i\times \tr(e_i(X))$.

\smallskip 
\noindent (\textbf{ones-vector}) Let $e(X):=\ones(e_1(X))$. By induction,
$e_1(X)\equiv\sum_{i=1}^p a_i\times e_i^1(X)$. Then, $e(X)\equiv \ones(e_1^1(X))$.

\smallskip 
\noindent (\textbf{diag}) Let $e(X):=\diag(e_1(X))$.
By induction,
$e_1(X)\equiv\sum_{i=1}^p a_i\times e_i^1(X)$. Then, $e(X)\equiv \sum_{i=1}^p a_i\times \diag(e_i^1(X))$.

\smallskip 
\noindent (\textbf{pointwise vector product}) Let $e(X):=e_1(X)\odot_v e_2(X)$. By induction we have that $e_1(X)\equiv \sum_{i=1}^p a_i\times e_i^1(X)$ and
$e_2(X)\equiv\sum_{i=1}^q b_i\times e_i^2(X)$. Hence, $e(X)\equiv \sum_{i=1}^p\sum_{j=1}^{q} (a_i\times b_i) \times (e_i^1(X)\odot_v e_j^2(X)$).

\smallskip 
\noindent (\textbf{Schur-Hadamard})
\smallskip 
\noindent (\textbf{pointwise vector product}) Let $e(X):=e_1(X)\odot e_2(X)$. By induction we have that $e_1(X)\equiv \sum_{i=1}^p a_i\times e_i^1(X)$ and
$e_2(X)\equiv\sum_{i=1}^q b_i\times e_i^2(X)$. Hence, $e(X)\equiv \sum_{i=1}^p\sum_{j=1}^{q} (a_i\times b_i) \times (e_i^1(X)\odot e_j^2(X)$). 

This concludes the proof.\qed

\section*{Proof of Proposition~\ref{prop:determined1}} 
\setcounter{proposition}{3} 
\setcounter{section}{7} 
\renewcommand{\theproposition}{\arabic{section}.\arabic{proposition}}
\begin{proposition}
	\ $\ML{\cdot,{}^*,\tr, \ones,\diag,+,\times,\Apply_{\mathsf{s}}[f],f\in\Omega}$-vectors are constant on equitable partitions. 
\end{proposition}
\begin{proof}
	Let ${\cal L}^{\#}$ denote $\{\cdot,{}^*,\tr, \ones,\diag,+,\times,\Apply_{\mathsf{s}}[f],f\in\Omega\}$. Consider a graph $G$ of order $n$ with equitable partition ${\cal V}=\{V_1,\ldots,V_\ell\}$. As before, let $\ones_{V_1},\ldots,\ones_{V_\ell}$ be the corresponding indicator vectors. We will show that for any expression $e(X)\in \ML{{\cal L}^{\#}}$ such that $e(A_G)$ is an $n\times 1$-vector, $e(A_G)$ can be uniquely written in the form $\sum_{i=1}^\ell a_i\times \ones_{V_i}$ for scalars $a_i\in \C$.
	
	We show, by induction on the structure of expressions in $\ML{{\cal L}^{\#}}$, that the following properties hold; 
	\begin{enumerate}
		\item[(a)] if $e(A_G)$ returns an $n\times n$-matrix, then for any pair $i,j=1,\ldots,\ell$ there exists a scalars $a_{ij}, b_{ij}\in\C$ such that
		
		\begin{linenomath}
			\[ \diag(\ones_{V_i})\mdot e(A_G)\mdot \ones_{V_j}= a_{ij} \times \ones_{V_i} \text{ and }\, \tp{\ones_{V_j}}\mdot e(A_G)\mdot \diag(\ones_{V_i})= b_{ij} \times \tp{\ones_{V_i}}\]
		\end{linenomath}
		\item[(b)] if $e(A_G)$ returns an $n\times 1$-vector, then for any $i=1,\ldots,\ell$, there exists a scalar $a_i\in C$ such that
		
		\begin{linenomath}
			\[ \diag(\ones_{V_i})\mdot e(A_G)=a_i\times \ones_{V_i}.\]
		\end{linenomath}
		
	\end{enumerate}
	
	Clearly, if (b) holds for every $i=1,\ldots,\ell$, then, $e(A_G)=\sum_{i=1}^\ell a_i\times\ones_{V_i}$ because $I=\sum_{i=1}^\ell \diag(\ones_{V_i})$. We remark these properties can be seen as generalization of the known fact that the vector space spanned by indicator vectors of an equitable partition of $G$ is invariant under multiplication by $A_G$ (See e.g., Lemma 5.2 in~\cite{Chan1997}). That is, for any linear combination $v=\sum_{i=1}^\ell a_i\times \ones_{V_i}$ we have that $A\mdot v=\sum_{i=1}^\ell b_i\times \ones_{V_i}$. In our setting, (a) and (b) imply that $e(A_G)\mdot v$ is again a linear combination of indicator vectors, when $e(A_G)$ returns an $n\times n$-matrix. We next verify properties (a) and (b). We often use that $I=\sum_{i=1}^\ell \diag(\ones_{V_i})$ and $\ones=\sum_{i=1}^\ell \ones_{V_i}$.

	\smallskip 
	\noindent (\textbf{base case}) Let $e(X):=X$. The required property is simply a restatement of the being equitable. That is,
	
	\begin{linenomath}
		\[ \diag(\ones_{V_i})\mdot e(A_G)\mdot \ones_{V_j}=\degr(v,V_j) \times \ones_{V_i},\]
	\end{linenomath}
	for an arbitrary vertex $v\in V_i$. So, we can take $a_{ij}=\degr(v,V_j)$. Similarly, because we $A_G$ is a symmetric matrix, 
	\begin{linenomath}
		\[ \tp{\ones_{V_j}}\mdot e(A_G)\mdot \diag(\ones_{V_i})=\tp{(\diag(\ones_{V_i})\mdot e(A_G)\mdot \ones_{V_j})}=\degr(v,V_j) \times \tp{\ones_{V_i}},\]
	\end{linenomath}
	for an arbitrary vertex $v\in V_i$. So, we can take $a_{ij}=\degr(v,V_j)$.
	
Below, for condition (a) we only verify that $\diag(\ones_{V_i})\mdot e(A_G)\mdot \ones_{V_j}= a_{ij} \times \ones_{V_i}$ holds. The verification of $\tp{\ones_{V_j}}\mdot e(A_G)\mdot \diag(\ones_{V_i})= b_{ij} \times \tp{\ones_{V_i}}$ is entirely similar.

\smallskip 
\noindent (\textbf{multiplication}) Let $e(X):=e_1(X)\mdot e_2(X)$. We distinguish between a number of cases, depending on the dimensions of $e_1(A_G)$ and $e_2(A_G)$. We first check the cases when $e(A_G)$ returns an $n\times n$-matrix and need to show that property (a) holds. 
\begin{itemize}
	\item ($\mathbf{n\times n, n\times n}$): $e_1(A_G)$ and $e_2(A_G)$ are of dimension $n\times n$. By induction, $\diag(\ones_{V_i})\mdot e_1(A_G)\mdot \ones_{V_j}=a_{ij}\times \ones_{V_j}$ and 
	
	$\diag(\ones_{V_i})\mdot e_2(A_G)\mdot \ones_{V_j}=b_{ij}\times \ones_{V_i}$.
	Then, $\diag(\ones_{V_i})\mdot e(A_G)\cdot\ones_{V_j}$ is equal to \allowdisplaybreaks 
	\begin{linenomath}
		\postdisplaypenalty=0
		\begin{align*}
			\diag(\ones_{V_i})\mdot e_1(A_G)\mdot e_2(A_G)\mdot \ones_{V_j}&=\sum_{k=1}^\ell \diag(\ones_{V_i})\mdot e_1(A_G)\mdot \diag(\ones_{V_k})\mdot e_2(A_G)\mdot \ones_{V_j}\\
			&=\sum_{k=1}^\ell b_{kj}\times \bigl(\diag(\ones_{V_i})\mdot e_1(A_G)\mdot \ones_{V_k}\bigr)=\Bigl(\sum_{k=1}^\ell b_{kj} \times a_{ik}\Bigr)\times \ones_{V_i},
			\end{align*}
	\end{linenomath}
	as desired.
		\item ($\mathbf{n\times 1, 1\times n}$): $e_1(A_G)$ is of dimension $n\times 1$ and $e_2(A_G)$ is of dimension $1\times n$. By induction we have that $\diag(\ones_{V_i})\mdot e_1(A_G)=a_i \times \ones_{V_i}$ and $\diag(\ones_{V_i})\mdot \tp{(e_2(A_G))}=b_i \times \ones_{V_i}$. Hence, $\diag(\ones_{V_i})\mdot e(A_G)\cdot\ones_{V_j}$ is equal to 
	\begin{linenomath}
		\postdisplaypenalty=0
		\begin{align*}
			\diag(\ones_{V_i})\mdot e_1(A_G)\mdot e_2(A_G)\cdot\ones_{V_j}& =a_i \times( \ones_{V_i}\mdot e_2(A_G)\mdot \ones_{V_j})\\
			&=\sum_{k=1}^\ell a_i \times( \ones_{V_i}\mdot e_2(A_G)\mdot \diag(\ones_{V_k})\mdot \ones_{V_j})\\
			&=\sum_{k=1}^\ell a_i \times( \ones_{V_i}\mdot \tp{\bigl(\diag(\ones_{V_k})\mdot \tp{(e_2(A_G))}\bigr)}\mdot \ones_{V_j})\\
			&=\sum_{k=1}^\ell (a_i\times b_k) \times( \ones_{V_i}\mdot \tp{\ones_{V_k}}\mdot \ones_{V_j})\\
			&=(a_i\times b_j \times |V_i|)\times \ones_{V_i}
			\end{align*}
	\end{linenomath}
	as desired.
	Here we used that $\tp{\ones_{V_k}}\mdot \ones_{V_j}$ is either $0$, in case that $k\neq j$, or $|V_j|$ in case that $j=k$.
\end{itemize}

\noindent We next check that condition (b) holds when $e(A_G)$ returns an $n\times 1$-vector. 
\begin{itemize}
	\item ($\mathbf{n\times n, n\times 1}$): $e_1(A_G)$ is of dimension $n\times n$ and $e_2(A_G)$ is of dimension $n\times 1$. By induction, we have that $\diag(\ones_{V_i})\mdot e_1(A_G)\mdot \ones_{V_j}=a_{ij}\times \ones_{V_i}$ and
			$\diag(\ones_{V_i})\mdot e_2(A_G)=b_{i}\times \ones_{V_i}$. Hence, $\diag(\ones_{V_i})\mdot e(A_G)$ is equal to
			\allowdisplaybreaks 
	\begin{linenomath}
		\postdisplaypenalty=0
		\begin{align*}
			\diag(\ones_{V_i})\mdot e_1(A_G)\mdot e_2(A_G)&=\sum_{j=1}^\ell \diag(\ones_{V_i})\mdot e_1(A_G)\mdot \diag(\ones_{V_j})\mdot e_2(A_G)\\
			&= \sum_{j=1}^\ell b_j \times (\diag(\ones_{V_i})\mdot e_1(A_G)\cdot\ones_{V_j}) = \sum_{j=1}^\ell (a_{ij}\times b_j) \times \ones_{V_i},
		\end{align*}
	\end{linenomath}
	as desired.
	
	\item ($\mathbf{n\times 1, 1\times 1}$): $e_1(A_G)$ is of dimension $n\times 1$ and $e_2(A_G)$ is of dimension $1\times 1$. By induction we have that $\diag(\ones_{V_i})\mdot e_1(A_G)=a_i\times \ones_{V_i}$ and $e_2(A_G)=b\in\C$. Hence, 
	\begin{linenomath}
		\[ \diag(\ones_{V_i})\mdot e(A_G)=\diag(\ones_{V_i})\mdot e_1(A_G)\mdot e_2(A_G)=(a_i\times b) \times \ones_{V_i},\]
	\end{linenomath} 
	as desired. 
\end{itemize}

\smallskip 
\noindent (\textbf{ones vector}) $e(X):=\ones(e_1(X))$. We only need to consider the case when $e_1(A_G)$ is an $n\times n$-matrix or $n\times 1$-vector. In both cases, it suffices to observe that $\ones=\sum_{i=1}^\ell \ones_{V_ii}$. Indeed,
\[ \diag(\ones_{V_i})\mdot e(A_G)=\diag(\ones_{V_i})\mdot \ones=\ones_{V_i}. \]

\smallskip 
\noindent (\textbf{conjugate transpose}) $e(X):=(e_1(X))^*$. If $e_1(A_G)$ returns a $1\times n$-vector, then $\diag(\ones_{V_i})\mdot \tp{(e_1(A_G))}=a_i\times\ones_{V_i}$. Hence, $\diag(\ones_{V_i})\mdot e_1(A_G)=a_i^*\times\ones_{V_i}$. If $e_1(A_G)$ returns an $n\times n$-matrix, then by induction, $\tp{\ones_{V_j}}\mdot e_1(A_G)\cdot\diag(\ones_{V_i})=b_{ij}\times\tp{\ones_{V_i}}$. Hence,
\[ \diag(\ones_{V_i})\mdot e(A_G)\cdot\ones_{V_j}=(\tp{\ones_{V_j}}\mdot e_1(A_G)\cdot\diag(\ones_{V_i}))^*=b_{ij}^*\times \ones_{V_i}, \]
as desired.

\smallskip 
\noindent (\textbf{diag operation}) $e(X):=\diag(e_1(X))$ where $e_1(A_G)$ is an $n\times 1$-vector. By induction, $\diag(\ones_{V_i})\mdot e_1(A_G)=a_i \times \ones_{V_i}$. Hence, in view of the linearity of the diagonal operation,
\[ \diag(\ones_{V_i})\mdot e(A_G)\cdot\ones_{V_j}=\sum_{k=1}^\ell a_i\times \bigl(\diag(\ones_{V_i})\cdot\diag(\ones_{V_k}) \cdot\ones_{V_j}\bigr)=a_i\times \ones_{V_i},
\]
since $\diag(\ones_{V_k})\cdot\ones_{V_j}$ is $\ones_{V_j}$ when $k=j$ and the zero vector otherwise. 

\smallskip 
\noindent (\textbf{addition}) $e(X):=e_1(X)+e_2(X)$. Clearly, when condition (a) or (b) hold for $e_1(A_G)$ and $e_2(A_G)$, they remain to hold for $e(A_G)$.

\smallskip 
\noindent (\textbf{scalar multiplication}) $e(X):=a\times e_1(X)$. Clearly, when condition (a) or (b) hold for $e_1(A_G)$, they remain to hold for $e(A_G)$.

\smallskip 
\noindent (\textbf{trace}) $e(X):=\tr(e_1(X))$. Such sub-expressions do not return matrices or vectors.

\smallskip 
\noindent (\textbf{pointwise function applications}) $e(X):=\Apply_{\mathsf{s}}[f](e_1(X),\dots,e_p(X))$ where each $e_i(X)$ is a sentence. Again, such sub-expressions do not return matrices or vectors. 
\end{proof}

\section*{Proof of Proposition~\ref{prop:determined2}} 
\setcounter{proposition}{1} 
\setcounter{section}{8} 
\renewcommand{\theproposition}{\arabic{section}.\arabic{proposition}}
\begin{proposition}
$\ML{\cdot,{}^*,\tr, \ones,\odot_v,\diag,+,\times,\Apply_{\mathsf{s}}[f],f\in\Omega}$-vectors are constant on equitable partitions 
\end{proposition}
\begin{proof}
Given that we verified this property of all operations except for $\odot_v$ in the proof of Proposition~\ref{prop:determined1}, we only need to verify that $\odot_v$ can be added to the list of supported operations. We use the same induction hypotheses as in the proof of Proposition~\ref{prop:determined1} and verify that these hypotheses remain to hold for $\odot_v$:

\smallskip 
\noindent (\textbf{pointwise vector multiplication}) $e(X):=e_1(X)\odot_v e_2(X)$ where $e_1(X)$ and $e_2(X)$ return vectors. By induction we have that $\diag(\ones_{V_i})\mdot e_1(A_G)=a_i\times \ones_{V_i}$ and $\diag(\ones_{V_i})\mdot e_2(A_G)=b_i\times \ones_{V_i}$. As a consequence, 
\begin{linenomath}
	\postdisplaypenalty=0 
	\begin{align*}
		\diag(\ones_{V_i})\mdot e(A_G)&=\diag(\ones_{V_i})\mdot e_1(A_G)\odot_v e_2(A_G)=a_i \times (\ones_{V_i}\odot_v e_2(A_G))\\
		&\sum_{j=1}^\ell a_i \times (\ones_{V_i}\odot_v (\diag(\ones_{V_j})\mdot e_2(A_G))=\sum_{j=1}^\ell (a_i\times b_j)\times (\ones_{V_i}\odot_v\ones_{V_j})\\
		& (a_i\times b_i)\times \ones_{V_i}, 
	\end{align*}
\end{linenomath}
because $\ones_{V_i}\odot_v\ones_{V_j})$ is either $\ones_{V_i}$ when $i=j$, or the zero vector when $i\neq j$. 
\end{proof}

\section*{Continuation of the proof of Theorem~\ref{thm:pwtrpart}} 
\setcounter{theorem}{0} 
\setcounter{section}{8} 
\renewcommand{\thetheorem}{\arabic{section}.\arabic{theorem}}

In the proof in the main body of the paper we left open the verification that $(\ones_{V_i}\cdot\tp{\ones}_{V_i})\mdot O=O\mdot (\ones_{W_i}\cdot\tp{\ones}_{W_i})$, for $i=1,\ldots,\ell$, implies that $O$ preserves the coarsest equitable partitions of $G$ and $H$. In particular, we need to verify that $\ones_{V_i}=O\cdot\ones_{W_i}$, for $i=1,\ldots,\ell$. This can be easily shown, just as in the proof of Theorem~\ref{thm:treqpart} (based on Lemma 4 in Th\"une~\cite{Thune2012}), in which we verified that $J\mdot O=O\mdot J$ implies that $\ones=O\cdot\ones$. 

First, we observe that $(\ones_{V_i}\cdot\tp{\ones}_{V_i})\mdot O\cdot\ones_{W_i}=\ones_{V_i}\mdot (\tp{\ones}_{V_i}\mdot O\mdot \ones_{W_i})=\alpha_i \times \ones_{V_i}$ with $\alpha_i=\tp{\ones}_{V_i}\mdot O\cdot\ones_{W_i}$ and $(\ones_{V_i}\cdot\tp{\ones}_{V_i})\mdot O\cdot\ones_{W_i}=O\mdot (\ones_{W_i}\cdot\tp{\ones}_{W_i})\cdot\ones_{W_i}=(\tp{\ones}_{W_i}\cdot\ones_{W_i})\times O\cdot\ones_{W_i}$. In other words, $O\mdot \ones_{W_i}=\frac{\alpha_i}{n_i}\times\ones_{V_i}$ where $\tp{\ones}_{W_i}\cdot\ones_{W_i}=|W_i|=n_i$. Furthermore, because $\tp{\ones}_{V_i}\mdot \tp{O}\cdot\ones_{W_i}$ is a scalar, $\tp{\ones}_{W_i}\mdot \tp{O}\cdot\ones_{V_i}=(\tp{\ones}_{V_i}\mdot O\cdot\ones_{W_i})^{\mathsf{t}}=\tp{\ones}_{V_i}\mdot O\cdot\ones_{W_i}=\alpha_i$. We next show that $\alpha=\pm n_i$. Indeed, since $O$ is an orthogonal matrix 
\begin{linenomath}
\[ n_i=\tp{\ones}_{V_i}\mdot I\cdot\ones_{W_i}=\tp{\ones}_{V_i}\mdot \tp{O}\mdot O\mdot \ones_{W_i}=\frac{\alpha_i}{n_i}\times (\tp{\ones}_{V_i}\mdot \tp{O}\mdot \ones_{V_i})=\frac{\alpha_i^2}{n_i}, \]
\end{linenomath}
and thus $\alpha_i^2=n_i^2$ or $\alpha_i=\pm n_i$. Hence, $O\cdot\ones_{W_i}=\pm \ones_{V_i}$. We note that $\ones=\sum_{i=1}^\ell \ones_{V_i}=\sum_{i=1}^\ell \ones_{W_i}$. We now argue that either $\ones_{V_i}=O\cdot\ones_{W_i}$ for all $i=1,\ldots,\ell$, or $-\ones_{V_i}=O \mdot \ones_{W_i}$ for all $i=1,\ldots,\ell$. Indeed, suppose that we have $\ones_{V_i}=O\cdot\ones_{W_i}$ for $i\in K\subset \{1,\ldots,\ell\}$ and $-\ones_{V_i}=O \mdot \ones_{W_i}$ for $i\in \bar K=\{1,\ldots,\ell\}\setminus K$, for some non-empty subset $K$ of $\{1,\ldots,\ell\}$. Then $\sum_{i\in K}\ones_{V_i}=O\mdot \bigl(\sum_{i\in K}\ones_{W_i}\bigr)$ and hence since $\sum_{i\in \bar K}\ones_{V_i}=\ones-\sum_{i\in K}\ones_{V_i}$ and $\sum_{i\in \bar K}\ones_{W_i}=\ones-\sum_{i\in K}\ones_{W_i}$, 
\begin{linenomath}
\[ \sum_{i\in\bar K}\ones_{V_i}=O\mdot \bigl(\sum_{i\in \bar K}\ones_{W_i}\bigr). \]
\end{linenomath}
This contradicts that $-\sum_{i\in\bar K}\ones_{V_i}=O\mdot \bigl(\sum_{i\in \bar K}\ones_{W_i}\bigr)$. Hence, when $\ones_{V_i}=O\cdot\ones_{W_i}$ for all $i=1,\ldots,\ell$, $O$ satisfies the desired property already. Otherwise, when $-\ones_{V_i}=O \mdot \ones_{W_i}$ for all $i=1,\ldots,\ell$, we simply replace $O$ by $(-1)\times O$ to obtain that $O\cdot\ones_{W_i}=\ones_{V_i}$. This rescaling does not impact that $A_G\mdot O=O\mdot A_H$ and we can thus indeed conclude that $O$ preserves the coarsest equitable partitions of $G$ and $H$. \hfill$\qed$

\end{document}